\documentclass[english]{article}
\usepackage[T1]{fontenc}
\usepackage[latin9]{inputenc}
\usepackage{geometry}
\usepackage{color}
\usepackage{babel}
\usepackage{float}
\usepackage{amsmath}
\usepackage{amsthm}
\usepackage{amssymb}
\usepackage{graphicx} 
\usepackage{algorithm}
\usepackage{algpseudocode}
\usepackage{comment}

\usepackage[round]{natbib}
\usepackage{enumitem} 

\usepackage[unicode=true,pdfusetitle,
 bookmarks=true,bookmarksnumbered=false,bookmarksopen=false,
 breaklinks=false,pdfborder={0 0 0},pdfborderstyle={},backref=false,colorlinks=true]
 {hyperref}
\hypersetup{
 citecolor=blue, linkcolor=red}

\makeatletter

\newcommand{\noun}[1]{\textsc{#1}}
\providecommand{\tabularnewline}{\\}

\numberwithin{equation}{section}
\theoremstyle{definition}
\newtheorem{defn}{\protect\definitionname}[section]
\theoremstyle{remark}
\newtheorem{rem}{\protect\remarkname}[section]
\theoremstyle{plain}
\newtheorem{prop}{Proposition}[section]
\theoremstyle{remark}

\theoremstyle{plain}
\newtheorem{thm}{Theorem}[section]
\theoremstyle{plain}
\newtheorem{cor}{\protect\corollaryname}[section]
\theoremstyle{plain}
\newtheorem{lem}{\protect\lemmaname}[section]

\newcommand*{\QEDA}{\hfill\ensuremath{\square}}

\makeatother

\providecommand{\corollaryname}{Corollary}
\providecommand{\definitionname}{Definition}
\providecommand{\lemmaname}{Lemma}
\providecommand{\notationname}{Notation}

\providecommand{\remarkname}{Remark}

\newcommand{\E}{\mathbb{E}}
\newcommand{\EQ}{\mathbb{E}^{\mathbb Q}}

\usepackage{todonotes}

\usepackage{mathtools}
\usepackage{authblk}
\title{The Volterra Stein-Stein model with stochastic interest rates}
\author[1]{Eduardo Abi Jaber\thanks{\emph{eduardo.abi-jaber@polytechnique.edu.}}}
\author[2]{Donatien Hainaut\thanks{\emph{donatien.hainaut@uclouvain.be}}}
\author[2]{Edouard Motte\thanks{\emph{Corresponding author, edouard.motte@uclouvain.be}}}
\affil[1]{Ecole Polytechnique, CMAP}
\affil[2]{Universit\'e Catholique de Louvain, LIDAM-ISBA,
Louvain-la-Neuve, Belgium}

\begin{document}
\maketitle

\begin{abstract}
We introduce the Volterra Stein-Stein model with stochastic interest rates, where both volatility and interest rates are driven by correlated Gaussian Volterra processes. This framework unifies various well-known Markovian and non-Markovian models while preserving analytical tractability for pricing and hedging financial derivatives. We derive explicit formulas for pricing zero-coupon bond and interest rate cap or floor, along with a semi-explicit expression for the characteristic function of the log-forward index using Fredholm resolvents and determinants. This allows for fast and efficient derivative pricing and calibration via Fourier methods. We calibrate our model to market data and observe that our framework is flexible enough to capture key empirical features, such as the humped-shaped term structure of ATM implied volatilities for cap options and the concave ATM implied volatility skew term structure (in log-log scale) of the S\&P 500 options. Finally, we establish connections between our characteristic function formula and expressions that depend on infinite-dimensional Riccati equations, thereby making the link with conventional linear-quadratic models. \end{abstract}
\noun{Keywords:} \emph{Gaussian Volterra processes, 
volatility, interest rate, memory, Fredholm resolvents and determinants, Fourier pricing, Riccati equations.}

\section{Introduction}

Modeling the joint dynamics of interest rates and equity is crucial for pricing and hedging hybrid derivatives. These products, which include equity-linked interest rate derivatives, callable hybrid structures, and quanto options, depend on the interaction between stock prices, interest rates, and their volatilities. A realistic model must  capture the well-documented stylized facts of each underlying as well as accurately describe their dependency structure to ensure consistent pricing and risk management across asset classes.\\

Empirical evidence reveals persistent memory effects in historical time series of interest rates and asset volatility and slow decay of their autocorrelations structures, see \cite{cont2001empirical, key-10, key-19}.\\ 

Beyond historical data, robust and universal patterns emerge in the term structures of option prices across a broad range of maturities. Specifically, 
\begin{enumerate}[label=(\roman*)]
    \item \label{item1intro} for interest rates: the term structure of the implied  volatilities of  cap and floor options is typically hump-shaped, with a steep increase at short maturities followed by a more gradual decay at longer maturities, see for instance \cite{key-10, de2004information};
    \item for stock indices such as the S\&P 500: the term structure of the At-The-Money (ATM) skew of the  implied volatility of call and put options, when plotted in log-log scale, exhibits typically a concave shape. In addition, the ATM skew displays an approximately linear decrease at long maturities suggesting a power-law decay, though only for sufficiently large maturities, see \cite{bergomi2015stochastic, key-8, delemotte2023yet, guyon2022does,abi2024volatility}.
\end{enumerate}

Various approaches have been developed to jointly model equity and interest rate dynamics, often by extending stochastic volatility frameworks to include stochastic interest rates. For instance, \cite{boyarchenko2007american,grzelak2011affine,boyarchenko2013american,key-23} incorporated stochastic rates into the \cite{heston1993closed} model, while \cite{key-25,key-26} adapted the \cite{key-24} and \cite{key-22} volatility models with multi-factor \cite{hull1993one} interest rates. However, these frameworks rely on Markovian dynamics with a single characteristic time scale for at least one of the two risk factors, and fail to capture the complex shapes of term structures observed on the market: they lead either to term structures of cap/floor implied volatilities that decrease monotonically - at odds with the empirically observed hump-shaped patterns -  or ATM skews of stock indices implied volatility that decay too quickly for longer maturities. While such models offer analytical tractability, their inability to align with empirical data highlights the need for alternative approaches that incorporate more flexibility.\\

 To address these limitations, we propose a flexible yet analytically tractable model based on Volterra processes, capable of accurately reproducing these empirical term structure features while remaining computationally efficient for option pricing and model calibration. Our motivation for incorporating Volterra processes stems from their proven ability to capture key stylized facts of volatility    (\cite{ key-8,key-14, guyon2023volatility,abi2024volatility})  and interest rates (\cite{corcuera2013short,benth2019non,key-5,key-16}). \\

We propose a hybrid equity-rate modeling framework where both volatility and interest rates are driven by (possibly correlated) Gaussian Volterra processes, unifying a broad class of Markovian and non-Markovian models. Building on the Volterra Stein-Stein model with constant interest rates studied by  \cite{key-4}, we extend the framework to incorporate stochastic interest rates, capturing key market features more effectively. Our main contributions can be summarized as follows: 
\begin{itemize}
    \item \textbf{Mathematical tractability and pricing:} Despite the non-Markovian nature of the processes, we derive explicit pricing formulas for zero-coupon bonds   and call and put options on zero-coupon bonds, see Propositions~\ref{prop:zero-coupon-price} and \ref{prop:pricing_ZC_options}. We  obtain a semi-explicit expression for the characteristic function of the log-forward index in terms of  \cite{fredholm1903classe} resolvents and determinants, enabling Fourier-based pricing methods, see Theorem~\ref{thm:chf_general}.  This result extends the formula derived by  \cite{key-4} for constant interest rates. 
    
   \item \textbf{Flexibility and joint calibration:} We calibrate our model  to market data and achieve excellent fits for: (i) the humped-shaped term structure of ATM implied volatilities for cap options by incorporating  in interest rates mean reversion as well as long-range memory with a fractional kernel, see Figure~\ref{fig:USD-cap-implied}, and (ii) the concave ATM implied volatility skew term structure (in a log-log plot) of S\&P 500 options using a shifted fractional kernel, see Figure~\ref{fig:S=000026P500-ATM-skew}. On the selected calibration date, our estimated parameters yield a negative implied correlation between short rates and the index, aligning with historical observations. Furthermore, we compare the impact of a singular power-law kernel for volatility (rough volatility) and demonstrate that rough models underperform our non-rough counterparts in capturing the entire volatility surface, confirming within our framework the findings of \cite{ delemotte2023yet, guyon2023volatility,abi2024volatility}.

\item \textbf{Link with conventional linear-quadratic models:}     We establish connections between our characteristic function  formula and expressions
that depend on infinite-dimensional  Riccati equations, see  Proposition~\ref{prop:-chf_riccati_1} for general kernels and Proposition~\ref{prop:riccati_ODE_last} for completely monotone convolution kernels. The latter formula establishes the link with conventional linear-quadratic models (Proposition~\ref{prop:chf_multi_factors_ODE}), allows us to recover  closed-form solutions in specific cases such as in \cite{key-25} (Corollary~\ref{corr:chf_stein-stein}), and leads to another numerical approximation method based on  multi-factor approximations in the spirit of \cite{key-1}, see Section~\ref{subsec:Multi-factor-approximation}.
\end{itemize}

The paper is outlined as follows. In Section \ref{sec:Mathematical-financial-framework}, we  introduce the Volterra Stein-Stein model with stochastic rates, derive pricing formulas for zero-coupon bonds, and analyze the forward index dynamics. Section \ref{sec:The-characteristic-function} provides a semi-explicit expression for the characteristic function of the log-index. In Section \ref{sec:Numerical-illustration-in}, we validate the Gaussian Volterra framework by calibrating the interest rate and volatility models to market data. Finally, Section \ref{sec:Link-with-conventional_models} establishes connections with conventional linear-quadratic models, derives simplified forms for Markovian volatility models, and introduces a multi-factor approximation for the characteristic function for  completely monotone kernels.

\section{The Volterra Stein-Stein and Hull-White model\protect\label{sec:Mathematical-financial-framework}}

We fix a finite horizon $T>0$ and a filtered probability space $(\Omega,\mathcal{F},(\mathcal{F}_{t})_{0\leq t\leq T},\mathbb{Q})$
where $\mathbb{Q}$ stands for one risk-neutral probability measure
and the filtration $(\mathcal{F}_{t})_{0\leq t\leq T}$ satisfies the usual conditions. We consider a 
financial market with a financial index denoted by $(I_{t})_{0\leq t\leq T}$
that depends on both stochastic interest rate $(r_{t})_{0\leq t\leq T}$
and volatility $(\nu_{t})_{0\leq t\leq T}$ with dynamic
\[
dI_{t}=r_{t}\:I_{t}\:dt+\nu_{t}I_{t}\:dW_{I}^{\mathbb{Q}}(t).
\]

In this paper, the full triplet $(I,r,\nu)$ refers to a Volterra Stein-Stein model with stochastic interest rates since, to incorporate memory effects, the two risk factors $(r,\nu)$ are modeled by Gaussian Volterra processes:
\begin{align}
r_{t} & =r_{0}(t)+\int_{0}^{t}G_{r}(t,s)\kappa_{r}\:r_{s}ds+\int_{0}^{t}G_{r}(t,s)\:\eta_{r}\:dW_{r}^{\mathbb{Q}}(s),\label{eq:IR_HW_multi_factor}\\
\nu_{t} & =g_{0}(t)+\int_{0}^{t}G_{\nu}(t,s)\kappa_{\nu}\:\nu_{s}ds+\int_{0}^{t}G_{\nu}(t,s)\:\eta_{\nu}\:dW_{\nu}^{\mathbb{Q}}(s),\label{eq:vol_risk-neutral_measure}
\end{align}
where $(W_{I}^{\mathbb{Q}},W_{r}^{\mathbb{Q}},W_{\nu}^{\mathbb{Q}})$
are correlated Brownian motions such that 
\begin{align*}
d\langle W_{l}^{\mathbb{Q}},W_{k}^{\mathbb{Q}}\rangle_{t}=\rho_{lk}\:dt,    
\end{align*}
with $\rho_{lk}\in[-1,1],\:l,k\in\{I,r,\nu\}$ such that the correlation matrix is semi-definite positive, $\kappa_r,\kappa_{\nu} \in \mathbb R$ controlling the speed of mean-reversion, $\eta_r,\eta_{\nu} \in \mathbb R$ the volatility parameters,  and $G_{r}(.)$ and $G_{\nu}(.)$ the kernels that
satisfy the following definition. 
\begin{defn}
\label{def:L2_kernel-1}A kernel
$G:[0,T]^2\to\mathbb{R}$ is a Volterra kernel of continuous and
bounded type in $L^{2}$ if $G(t,s)=0$ for $s\geq t$ and 
\[
\sup_{t\leq T}\int_{0}^{T}\left(|G(t,s)|^{2} +|G(s,t)|^2 \right)ds<+\infty, \quad \lim_{h\to0}\int_{0}^{T}|G(u+h,s)-G(u,s)|^{2}ds=0,\quad u\leq T.
\]
\end{defn}
Finally, $r_{0}$ is a time-dependent 
function used to perfectly fit
the initial term-structure of market bond prices, see Remark~\ref{R:stripr0}, and $g_{0}$ is such that 
\begin{equation}
g_{0}(t)=\nu_{0}+\theta_{\nu}\int_{0}^{t}G_{\nu}(t,s)ds,\quad \nu_{0}\in\mathbb{R}_{+},\:\theta_{\nu}\in\mathbb{R}.\label{eq:g_0_t}
\end{equation}
Note that, in general, the input curve $g_{0}$ can be non-parametric and used to fit at-the-money implied volatility term-structure observed in the market. \\

Under Definition~\ref{def:L2_kernel-1} and for $r_0,g_0$ locally square integrable,  it can be shown there exists a strong solution to
\eqref{eq:IR_HW_multi_factor} and \eqref{eq:vol_risk-neutral_measure},
such that 
\[
\sup_{0\leq t\leq T}\E^{\mathbb{Q}}\left[|r_{t}|^{p}+|\nu_{t}|^{p}\right]<+\infty,\quad p\geq1,
\]
this follows from an adaption of \cite[Theorem A.3]{key-4}. In particular, the joint process $(r,\nu)$ is a Gaussian process. This class of models encompasses well-known Markov
and non-Markovian Volterra processes. For an introduction to Volterra processes,
we refer among others to  \cite[Chapter 9]{key-16}. Several kernels satisfy Definition \ref{def:L2_kernel-1} such as: 
\begin{itemize}
\item \textbf{the constant kernel:} $G(t,s)=1_{s<t}$, then we obtain Markovian models such as the
classical  \cite{key-24} or  \cite{key-22}
model for volatility and the \cite{hull1993one} model for short-term interest
rate, {such models have been studied by \cite{key-25, grzelak2012extension},}
\item \textbf{the exponential kernel:} $G(t,s)=1_{s<t}\exp(-\beta(t-s))$ with $\beta\in\mathbb{R}_{+}$,
then we  also obtain Markovian models but with modified mean-reversion
level and reversion speed,
\item \textbf{the fractional kernel:} $G(t,s)=1_{s<t}\frac{(t-s)^{H-1/2}}{\Gamma(H+1/2)}$ with a Hurst index $H \in (0,1)$ then we
obtain non-semimartingale and non-Markovian long or short memory fractional processes whenever $H \neq 1/2$. For $H<1/2,$ the kernel is singular at $t \to s$,
we get rough models as for example the rough Stein-Stein model studied by
 \cite{key-4} and for $H>1/2$, we obtain long-memory processes
as, for example, the fractional version of the Hull and White model, as studied by \cite{jacquier2024interest}.
\item \textbf{the shifted fractional kernel:} $G(t,s)=1_{s<t}\frac{(t -s+\varepsilon)^{H-1/2}}{\Gamma(H+1/2)},$
with $\varepsilon>0$ and $H \in \mathbb R$, then we obtain path-dependent processes that
are semi-martingales but non-Markovian processes. Compared to the fractional kernel, the shifted fractional kernel has no singularity as $t \to s$, and can achieve faster decays than the fractional kernel for larger $t$ with coefficients $H \leq 0$. 
\end{itemize}

\begin{rem}
    We remark that as $\nu$ is a Gaussian process, it may become negative with positive probability. However, this feature is standard in Gaussian stochastic volatility models and does not create mathematical issues. Moreover, in stochastic volatility models, the important quantity is the sign of the covariance between the spot and the variance, $\langle \log(I), \nu^2 \rangle_t$. In particular, in the Markovian setting, we have
    \[
    d\langle \log(I), \nu^2 \rangle_t = 2\eta_\nu \rho_{I\nu}\,\nu_t^2\,dt.
    \]
    Hence, even if $\nu$ may change sign, the sign of the spot-variance covariance remains determined by the parameter $\rho_{I\nu}$, which is typically negative in order to reproduce the leverage effect. Furthermore, in the following, no truncation or correction would be applied during calibration or pricing, as this would compromise the analytical tractability of the model and could introduce pricing biases.
\end{rem}

\subsection{Pricing Zero-Coupon bonds and interest rate derivatives}
We now consider the pricing of zero-coupon bonds as well as interest
rate derivatives in our framework. The price at time $t\leq T$ of a $T-$maturity
zero-coupon bond is denoted by $P(t,T)$ and is defined by
\[
P(t,T)=\EQ \left[ e^{-\int_{t}^{T}r_{s}ds}\mid \mathcal{F}_{t} \right].
\]
We now show that even considering a general Volterra Hull and White
model for interest rate dynamics, the zero-coupon bond price admits
an analytic expression. To this end, we introduce the concept of resolvent 
associated to a kernel and give some examples in Table \ref{tab:Resolvents_example}.
Note that we consider the resolvent definition of \cite{key-4}, which
may differ from other papers. 
\begin{defn}
Let $G$ be a kernel satisfying Definition \ref{def:L2_kernel-1}.
The resolvent of $G$ is the kernel $R_{G}$ satisfying Definition \ref{def:L2_kernel-1} and is such that 
\begin{equation} \label{eq: def_resolvent}
R_{G}=G+G\star R_{G},
\end{equation}
where the $\star-$product is such that, for $(s,w)\in[0,T]^{2},$
\[
(G\star R_{G})(s,w)=\int_{0}^{T}G(s,z)R_{G}(z,w)dz.
\]
\end{defn}

\footnotesize
\begin{table}[H]
\centering{}%
\begin{tabular}{|c|c|c|}
\hline 
{Kernel} $G(t,s)$ & {Resolvent} $R_{G}(t,s)$ & $B_G(t,T):=\int_t^T R_G(s,t) ds$\tabularnewline
\hline 
$1_{s<t}\:c$ & $1_{s<t}\:c e^{c(t-s)}$ & $e^{c(T-t)}-1$\tabularnewline
\hline 
$1_{s<t}\:c\exp(-\beta(t-s))$ & $1_{s<t}\:c e^{(t-s)(c-\beta)}$ & $\frac{c}{c-\beta} (e^{(c-\beta)(T-t)}-1)$\tabularnewline
\hline 
$1_{s<t}\:c\frac{(t-s)^{\alpha-1}}{\Gamma(\alpha)}$ & $1_{s<t}\:c(t-s)^{\alpha-1}E_{\alpha,\alpha}(c(t-s)^{\alpha})$ & $c (T-t)^{\alpha} E_{\alpha,\alpha+1}(c(T-t)^{\alpha})$\tabularnewline
\hline 
\end{tabular}\caption{\protect\label{tab:Resolvents_example}Particular kernels, the associated resolvents and the integral of these resolvents, where $E_{\beta,\gamma}(x):=\sum_{n=0}^{\infty}\frac{x^{n}}{\Gamma(n\beta+\gamma)}$ is the two-parameter Mittag-Leffler function}.
\end{table}
\normalsize

\begin{prop}\label{prop:zero-coupon-price} 
We define, for $t\in[0,T],$ $f_{t}(.)$
by
\begin{equation}
f_{t}(s):=r_{0}(s)+\int_{0}^{t}R_{\kappa_{r}G_{r}}(s,u)r_{0}(u)du+\frac{\eta_{r}}{\kappa_{r}}\int_{0}^{t}R_{\kappa_{r}G_{r}}(s,u)\;dW_{r}^{\mathbb{Q}}(u),\;t\leq s,\label{eq:forward_rate}
\end{equation}
with $R_{\kappa_{r}G_{r}}$ the resolvent associated to $\kappa_{r}G_{r}$
with the convention $\frac{R_{\kappa_{r}G_{r}}}{\kappa_{r}}=G_{r}$
if $\kappa_{r}=0$\footnote{This convention follows from $\lim_{\kappa_r \to 0}R_{\kappa_{r}G_{r}}/\kappa_{r} = G_r$}. The price of the zero-coupon bond at time $t\in[0,T],$
is given by 
\begin{equation}
P(t,T)=A(t,T)\exp\bigg(-\int_{t}^{T}f_{t}(s)\:ds\bigg),\label{eq:zero_coupon_expression}
\end{equation}
with $A(t,T)$ the time-dependent function given by 
\[
A(t,T):=\exp\bigg(-\int_{t}^{T}\bigg(\int_{u}^{T}R_{\kappa_{r}G_{r}}(s,u)r_{0}(u)ds+\frac{\eta_{r}^{2}}{2}B_{G_{r}}^{2}(u,T)\:\bigg)\:du\bigg).
\]
Moreover, the dynamics of ($P(t,T))_{t\in[0,T]}$ are given by 
\[
\frac{dP(t,T)}{P(t,T)}=r_{t}\:dt-\eta_{r}B_{G_{r}}(t,T)dW_{r}^{\mathbb{Q}}(t),
\]
with 
\begin{align}\label{eq:BGr}
B_{G_{r}}(t,T):=\frac{1}{\kappa_{r}}\int_{t}^{T}R_{\kappa_{r}G_{r}}(s,t)\;ds.    
\end{align}
\end{prop}
\begin{proof}
Using the resolvent associated to the kernel $\kappa_{r}G_{r}$ and
using  \cite[Theorem A.3]{key-4}, we have that $(r_{t})_{0\leq t\leq T}$
is the strong solution of the following equation
\[
r_{t}=r_{0}(t)+\int_{0}^{t}R_{\kappa_{r}G_{r}}(t,s)r_{0}(s)ds+\frac{1}{\kappa_{r}}\int_{0}^{t}R_{\kappa_{r}G_{r}}(t,s)\eta_{r}\;dW_{r}^{\mathbb{Q}}(s),\:0\leq t\leq T,
\]
with the convention $\frac{R_{\kappa_{r}G_{r}}}{\kappa_{r}}=G_{r}$
if $\kappa_{r}=0.$ In this case, we have that, for $0\leq t\leq T,$
\[
\int_{t}^{T}r_{s}ds=\int_{t}^{T}r_{0}(s)ds+\int_{t}^{T}\int_{0}^{s}R_{\kappa_{r}G_{r}}(s,u)r_{0}(u)\:duds+\frac{1}{\kappa_{r}}\int_{t}^{T}\int_{0}^{s}R_{\kappa_{r}G_{r}}(s,u)\eta_{r}\;dW_{r}^{\mathbb{Q}}(u)ds.
\]
Using the definition of $f_{t}$ given by \eqref{eq:forward_rate} together with an application of stochastic Fubini's theorem,
we obtain that 
\[
\int_{t}^{T}r_{s}ds=\int_{t}^{T}f_{t}(s)ds+\int_{t}^{T}\int_{u}^{T}R_{\kappa_{r}G_{r}}(s,u)r_{0}(u)dsdu+\eta_{r}\int_{t}^{T}\underbrace{\frac{1}{\kappa_{r}}\int_{u}^{T}R_{\kappa_{r}G_{r}}(s,u)\;ds}_{:=B_{G_{r}}(u,T)}dW_{r}^{\mathbb{Q}}(u).
\]
Thus, for $t\in[0,T],$ conditional on $\mathcal{F}_{t}$, the random variable $\int_{t}^{T}r_{s}ds|\mathcal{F}_t$ is
Gaussian with mean 
\[
\E^{\mathbb{Q}}\left[\int_{t}^{T}r_{s}ds\mid \mathcal{F}_{t}\right]=\int_{t}^{T}f_{t}(s)ds+\int_{t}^{T}\int_{u}^{T}R_{\kappa_{r}G_{r}}(s,u)r_{0}(u)dsdu,
\]
and variance
\[
\mathbb{V}^{\mathbb{Q}}\left[\int_{t}^{T}r_{s}ds\mid\mathcal{F}_{t}\right]=\eta_{r}^{2}\int_{t}^{T}B_{G_{r}}^{2}(u,T)\:du.
\]
Therefore, we readily deduce that 
\[
P(t,T)=\exp\bigg(-\int_{t}^{T}f_{t}(s)ds-\int_{t}^{T}\bigg(\int_{u}^{T}R_{\kappa_{r}G_{r}}(s,u)r_{0}(u)ds+\frac{\eta_{r}^{2}}{2}B_{G_{r}}^{2}(u,T)\:\bigg)\:du\bigg),
\]
and we finally obtain that 
\[
\frac{dP(t,T)}{P(t,T)}=r_{t}dt-\eta_{r}\int_{t}^{T}\frac{1}{\kappa_{r}}R_{\kappa_{r}G_{r}}(s,t)ds\:dW_{r}^{\mathbb{Q}}(t).
\]
\end{proof}

\begin{rem}\label{R:stripr0}
Proposition \ref{prop:zero-coupon-price} provides a natural
way to estimate the function $r_{0}(.)$ such that the interest rate
model matches perfectly the initial term-structure of market bond
prices. In fact from \eqref{eq:zero_coupon_expression}, we deduce
that 
\[
\int_{0}^{t}r_{0}(s)\:\bigg(1+\kappa_{r}B_{G_{r}}(s,t)\bigg)\:ds=-\ln P(0,t)-\int_{0}^{t}\frac{\eta_{r}^{2}}{2}B_{G_{r}}^{2}(s,t)\:ds.
\]
Thus, if we assume that, for some given maturities $0=T_{0}<T_{1}<...<T_{n}$,
we have the initial term-structure of market bond prices $P^{Market}(0,T_{i})$,
and that $r_{0}(t)$ is piecewise constant such that
\begin{equation}
r_{0}(t)=-\frac{\ln P^{Market}(0,T_{i})-\ln P^{Market}(0,T_{i-1})+\frac{\eta_{r}^{2}}{2}\int_{T_{i-1}}^{T_{i}}B_{G_{r}}^{2}(u,T_{i})\:du}{(T_{i}-T_{i-1})+\kappa_{r}\int_{T_{i-1}}^{T_{i}}B_{G_{r}}(s,T_{i})\,ds},\:t\in[T_{i-1},T_{i}),\label{eq:link_r0_ZC}
\end{equation}
then $P(0,T_{i})=P^{Market}(0,T_{i}),\:i=1,...,n.$ 
\end{rem}
Based on the expression for the pricing of zero-coupon bonds, we can
now easily deduce explicit expressions for the pricing of zero-coupon
bond call and put options. 
\begin{prop}
\label{prop:pricing_ZC_options}We consider $T-$maturity zero-coupon
bond call option $(P(T,S)-K)_{+}$ and put option $(K-P(T,S))_{+}$
with $S>T$ and $K$ the strike of the option. The arbitrage-free
price at time $t\in[0,T]$ of the call option is given by 
\[
\text{Call}_{ZC}(t,T,S,K)=P(t,S)\phi(-d_{1}(t,T))-KP(t,T)\phi(-d_{2}(t,T)),
\]
and the price of the put option is given by 
\[
\text{Put}_{ZC}(t,T,S,K)=KP(t,T)\phi(d_{2}(t,T))-P(t,S)\phi(d_{1}(t,T)),
\]
where $\phi(.)$ is the cumulative function of a normal distribution
and 
\begin{equation}
d_{1}(t,T)=d_{2}(t,T)-v(t, T, S),\label{eq:d_1}
\end{equation}
\begin{equation}
d_{2}(t,T)=\frac{1}{v(t,T,S)\sqrt{T-t}}\log\bigg(\frac{KP(t,T)}{P(t,S)}\bigg)+\frac{v(t,T,S)\sqrt{T-t}}{2},\label{eq:d_2}
\end{equation}
where 
\begin{equation}
v(t,T,S)^{2}=\frac{\eta_{r}^{2}\int_{t}^{T}\bigg(B_{G_{r}}(s,T)-B_{G_{r}}(s,S)\bigg)^{2}ds}{T-t}\label{eq:variance_ZC_bond_option_pricing}.
\end{equation}
\end{prop}
\begin{proof}
We consider the call option. We know that the arbitrage-free price
of a call on zero-coupon bond satisfies 
\[
\text{Call}_{ZC}(t,T,S,K)=P(t,T)\E^{\mathbb{Q}^{T}}\left[(P(T,S)-K)_{+} \mid \mathcal F_t\right].
\]
From Proposition \ref{prop:zero-coupon-price}, we know that under
the risk-neutral measure $\mathbb{Q},$
\[
\frac{dP(t,S)}{P(t,S)}=r_{t}\:dt-\eta_{r}B_{G_{r}}(t,S)dW_{r}(t).
\]
Therefore, for $S>T,$ we obtain that under the $T-$forward measure
$\mathbb{Q}^{T}$\footnote{We introduce the forward measure in Section \ref{sec:forward_measure}.}, the zero-coupon bond dynamics is given by 
\[
\frac{dP(t,S)}{P(t,S)}=\bigg(r_{t}+\eta_{r}^{2}B_{G_{r}}(t,S)B_{G_{r}}(t,T)\bigg)\:dt-\eta_{r}B_{G_{r}}(t,S)dW_{r}^{\mathbb{Q}^{T}}(t).
\]
Moreover, under $\mathbb{Q}^{T},$ we have that 
\[
d\bigg(\frac{P(t,S)}{P(t,T)}\bigg)=\frac{P(t,S)}{P(t,T)}\eta_{r}\bigg(B_{G_{r}}(t,T)-B_{G_{r}}(t,S)\bigg)dW_{r}^{\mathbb{Q}^{T}}(t),
\]
and we deduce that 
\[
\E^{\mathbb{\mathbb{Q}}^{T}}\left[(P(T,S)-K)_{+}\mid \mathcal{F}_{t}\right]=\frac{P(t,S)}{P(t,T)}\phi(-d_{1}(t,T))-K\phi(-d_{2}(t,T)),
\]
where $\phi(.)$ is the cumulative function of a normal distribution
and $d_{1},d_{2}$ are given by \eqref{eq:d_1}-\eqref{eq:d_2}. Finally,
using the call/put parity, we easily deduce the form of the put price. 
\end{proof}
Based on zero-coupon bonds call or put options, we can price cap and
floor options since, as explained in \cite{key-9}, cap and floor options
can be decomposed into a sum of zero-coupon bonds options. In fact, 
the cap and floor payoffs are of the form

\[
\text{Cap}=\sum_{i=1}^{\beta}(T_{i}-T_{i-1})\:(L(T_{i-1},T_{i})-K)_{+},
\]
\[
\text{Floor}=\sum_{i=1}^{\beta}(T_{i}-T_{i-1})\:(K-L(T_{i-1},T_{i}))_{+},
\]
where $L$ is the reference rate and $(T_{1},...,T_{\beta})$ are
the payment dates. Using \cite{key-9}, we have that the price, at time $t< T_0$, of cap and floor options is given by
\[
\text{Cap}(t,\beta,K)=\sum_{i=1}^{\beta}(1+K(T_{i}-T_{i-1}))\:\text{Put}_{ZC}\left(t,T_{i-1},T_i,\frac{1}{1+K(T_{i}-T_{i-1})}\right). 
\]
\[
\text{Floor}(t,\beta,K)=\sum_{i=1}^{\beta}(1+K(T_{i}-T_{i-1}))\:\text{Call}_{ZC}\left(t,T_{i-1},T_i,\frac{1}{1+K(T_{i}-T_{i-1})}\right).
\]
This enable us to calibrate our Gaussian Volterra model for interest rate to the market data of cap and floor options for any Volterra kernel $G_{r}(.)$. 

\subsection{Forward measure and pricing of derivatives on the index}\label{sec:forward_measure}

We now focus on the pricing of financial derivatives on the index $I$ in our framework.
Standard arguments imply that the arbitrage-free price of a financial
derivative is obtained by taking the discounted value of a payoff
under a risk-neutral measure $\mathbb{Q}$. In this case, if we consider
a general $T-$maturity derivative of the form 
\[
H_{T}:=h(I_{T})\in L^{2}(\mathbb{Q}),
\]
where $h(.)$ is a continuous positive function, the arbitrage-free
price at time $t\in[0,T],$ is given by 
\begin{equation}
V_{t}=\E^{\mathbb{Q}}\left[e^{-\int_{t}^{T}r_{s}ds}H_{T}\mid \mathcal{F}_{t}\right].\label{eq:price_derivatives}
\end{equation}
In the presence of stochastic interest rates, except for linear payoffs, it is impossible to
obtain a more explicit form of the price of financial derivatives on $I$ if we consider the formulation
under the risk-neutral measure. One technique for dealing with this
problem is to switch from the risk-neutral measure to the $T-$forward
neutral measure. Thus, we introduce the  $T-$forward neutral
measure denoted by $\mathbb{Q}^{T}$ which is equivalent to the risk-neutral
measure $\mathbb{Q}$. The change of measure from the risk-neutral
measure $\mathbb{Q}$ to the $T-$forward measure $\mathbb{Q}^{T}$
satisfies 
\[
\frac{d\mathbb{Q}^{T}}{d\mathbb{\mathbb{Q}}}\bigg|_{t}:=e^{-\int_{0}^{t}r(s)\:ds}\frac{P(t,T)}{P(0,T)}.
\]

We can now introduce the $T-$forward index denoted by $(I_{t}^{T})_{0\leq t\leq T}$
and defined as, for $0\leq t\leq T,$
\[
I_{t}^{T}:=\E^{\mathbb{Q}^{T}}\left[I_{T}\mid \mathcal{F}_{t}\right]=\frac{I_{t}}{P(t,T)}.
\]
Since under the forward measure the financial index divided by the
zero-coupon bond is a martingale, we obtain that the process $(I_{t}^{T})_{0\leq t\leq T}$
has the following dynamic 
\begin{align}\label{eq:dyn_forward_inflation}
\frac{dI_{t}^{T}}{I_{t}^{T}}=\nu_{t}\:dW_{I}^{\mathbb{Q}^{T}}(t)+\eta_{r}B_{G_{r}}(t,T)\:dW_{r}^{\mathbb{Q}^{T}}(t),    
\end{align}
with
\begin{equation}
\nu_{t}=g_{0}^{T}(t)+\int_{0}^{t}G_{\nu}(t,s)\kappa_{\nu}\:\nu_{s}ds+\int_{0}^{t}G_{\nu}(t,s)\:\eta_{\nu}\:dW_{\nu}^{\mathbb{Q}^{T}}(s),\label{eq:vol_forward_measure}
\end{equation}
 where $(W_{r}^{\mathbb{Q}^{T}},W_{I}^{\mathbb{Q}^{T}},W_{\nu}^{\mathbb{Q}^{T}})$
is a $3-$dimensional correlated Brownian motion under the $T-$forward
measure $\mathbb{Q}^{T}$ and 
\[
g_{0}^{T}(t):=g_{0}(t)-\eta_{\nu}\eta_{r}\rho_{r\nu}\int_{0}^{t}G_{\nu}(t,s)\:B_{G_{r}}(s,T)\:ds.
\]
 
\begin{rem}
    By Girsanov's theorem, the $3-$dimensional correlated Brownian motion under the $T-$forward measure $(W_{r}^{\mathbb{Q}^{T}},W_{I}^{\mathbb{Q}^{T}},W_{\nu}^{\mathbb{Q}^{T}})$ satisfies, for $t\leq T$,
    \begin{equation*}
        \begin{aligned}
            &dW_{r}^{\mathbb{Q}^{T}}(t) = dW_{r}^{\mathbb{Q}}(t)+ \eta_r B_{G_r}(t,T) dt,\\
            &dW_{I}^{\mathbb{Q}^{T}}(t) = dW_{I}^{\mathbb{Q}}(t)+ \rho_{Ir}\eta_r B_{G_r}(t,T) dt,\\
            &dW_{\nu}^{\mathbb{Q}^{T}}(t) =dW_{\nu}^{\mathbb{Q}}(t)+\rho_{r\nu} \eta_r B_{G_r}(t,T) dt. 
        \end{aligned}
    \end{equation*}
\end{rem}

By combining the forward measure and the $T-$forward financial index,
we obtain that the arbitrage-free price of a derivative (\ref{eq:price_derivatives})
satisfies 
\[
V_{t}=P(t,T)\E^{\mathbb{Q}^{T}}\left[\,h(e^{\log I_{T}^{T}})\,\mid \mathcal{F}_{t}\right].
\]
By switching from the risk-neutral measure to the $T-$forward measure,
we get a simpler framework for pricing financial derivatives. In the
following, we will show that the characteristic function of the log-forward
index $(\log I_{t}^{T})_{0\leq t\leq T}$ admits an explicit form.
This opens the way to a fast pricing of financial derivatives using
Fourier methods. Moreover, using hedging approaches developed in
several papers \cite{key-11,key-20,key-7}, some hedging strategies
can also be deduced by Fourier methods for contingent claims that
admit a Fourier representation. \\

For convenience of the reader, we end this section by collecting in Table \ref{tab: notation} the main notation for the objects used in the sequel. 

\begin{table}[h] 
\centering
\begin{tabular}{ll}
\hline
\textbf{Object} & \textbf{Notation} \\
\hline
Kernel & $G_.$ \\
Resolvent defined by \eqref{eq: def_resolvent} & $R_.$ \\
Integral of the resolvent defined by \eqref{eq:BGr} & $B_{G_.}$ \\
Speed of mean reversion & $\kappa_.$ \\
Risk-factor volatility & $\eta_.$ \\
Correlation & $\rho_.$ \\

\hline
\end{tabular}
\caption{\protect \label{tab: notation} Summary of notation.}
\end{table}

\section{The characteristic function of the log-forward index\protect\label{sec:The-characteristic-function} }

The aim of this section is to derive an analytical form of the characteristic
function of the log-forward index $(I_{t}^{T})_{0\leq t\leq T}$ with  dynamics 
\eqref{eq:dyn_forward_inflation} under $\mathbb{Q}^{T}$  and apply it for  Fourier pricing of  derivatives on the index.  \\

We start by recalling some results on operator theory in Hilbert spaces as well as introducing some notations. Let $\mathbf{A}$ be a linear compact operator acting on $L^{2}([0,T],\mathbb{C})$. Then, $\mathbf{A}$ is a bounded operator i.e. there exists $C>0$ such that, for all $f\in L^{2}([0,T],\mathbb{C})$, $||\mathbf{A}f||\leq C ||f||$ and for $B:=\{f\in L^{2}([0,T],\mathbb{C}): ||f||\leq 1\}$, the closure of $\mathbf{A}(B)$ is compact in $L^{2}([0,T],\mathbb{C})$. $\mathbf{A}$ is an integral operator if $\mathbf{A}$ a linear operator induced by a kernel $G\in L^{2}([0,T]^2,\mathbb{R})$ such that, for $f\in L^{2}([0,T],\mathbb{C})$ and $s\in[0,T],$
\[
(\mathbf{A}f)(s)=\int_{0}^{T}G(s,w)f(w)dw.
\]
Moreover, if $\mathbf{A}$ is a integral operator, then $\mathbf{A}$ is a Hilbert-Schmidt operator on  $L^{2}([0,T],\mathbb{C})$ and is in particular compact. The trace of an operator, denoted by $\text{Tr}(.)$, is defined for operators of trace class where a compact operator $\mathbf{A}$ is said to be of trace class if
\[
\text{Tr}(\mathbf{A})=\sum_{n\geq 1} \langle \mathbf{A} v_n, v_n \rangle <\infty,
\]
for a given orthonormal basis $(v_n)_{n\geq 1}$. For more details about Hilbert-Schmidt operators and trace class operators,  we refer to \cite{gohberg1978introduction, simon2005trace,gohberg2012traces}, \cite[Proposition 3]{brislawn1988kernels}, \cite[Section 2 and 3]{bornemann2010numerical}, and also \cite[Section
A]{key-4}.\\

We now introduce some notations and properties:
\begin{itemize}
\item $\langle \cdot,\cdot\rangle_{L^{2}}$ denotes the following product on $L^{2}([0,T],\mathbb{C})$
\[
\langle f,g\rangle_{L^{2}}=\int_{0}^{T}f(s)g(s)ds,\quad f,g\in L^{2}([0,T],\mathbb{C}).
\]
Note that $\langle\cdot,\cdot\rangle_{L^{2}}$ is the inner product in $L^{2}([0,T],\mathbb{R})$ but not in $L^{2}([0,T],\mathbb{C})$. 
\item For any kernels $G_{1},G_{2}\in L^{2}([0,T]^{2},\mathbb{\mathbb{R}}),$
the $\star-$product is defined by
\[
(G_{1}\star G_{2})(s,w)=\int_{0}^{T}G_{1}(s,z)G_{2}(z,w)dz, \quad (s,w)\in[0,T]^{2}.
\]
\item For a kernel $G,$ $G^{*}$ denotes the adjoint kernel such that,
for $(s,w)\in[0,T]^{2},$
\[
G^{*}(s,w)=G(w,s),
\]
and $\mathbf{G}^{*}$ is the operator induced by $G^{*}.$
\item $\text{id}$ denotes the identity operator such that $(\text{id}f)=f$,  $f\in L^{2}([0,T],\mathbb{C})$. 
\end{itemize}

\subsection{Semi-explicit expression}

We now derive a semi-explicit form for the characteristic function of
the log-forward index in our framework. \\

To this end, we consider an
approach similar to \cite{key-4} which deduces the explicit form
of the log-price characteristic function in equity markets for the
Volterra Stein-Stein volatility model. Here, we show that this result
can be extended to a framework where interest rates are also stochastic.
As in \cite{key-4}, we consider the expression of
the adjusted conditional mean of $\nu_{t}$ and then we define a linear
operator $\mathbf{\Psi}_{t}^u$ in $L^{2}([0,T],\mathbb{C})$ that will
be useful for the expression of the characteristic function.
The adjusted conditional mean of $\nu_{t}$ denoted by $g_{t}$ is
given by 
\[
g_{t}(s)=1_{t\leq s}\E^{\mathbb{Q}^{T}}\left[\nu_{s}-\int_{t}^{T}G_{\nu}(s,w)\kappa_{\nu}\nu_{w}dw\mid \mathcal{F}_{t}\right],\quad s,t\leq T,
\]
and we can easily show that $g_{t}$ reduces to 
\[
g_{t}(s)=1_{t\leq s}\bigg(g_{0}^{T}(s)+\int_{0}^{t}G_{\nu}(s,w)\kappa_{\nu}\nu_{w}dw+\int_{0}^{t}G_{\nu}(s,w)\eta_{\nu}\:dW_{\nu}^{\mathbb{Q}^{T}}(w)\bigg),\quad s,t\leq T.
\]

\begin{defn}
For $u\in\mathbb{C}$ such that $0\leq\mathcal{\mathfrak{R}}(u)\leq1,$
the operator $\mathbf{\Psi}_{t}^{u}$ acting on $L^{2}([0,T],\mathbb{C})$
is defined by 
\begin{equation}
\mathbf{\Psi}_{t}^{u}:=(\text{id}-b^{u}\mathbf{G}_{\nu}^{*})^{-1}a^{u}(\text{id}-2a^{u}\tilde{\mathbf{\Sigma}}_{t}^{u})^{-1}(\text{id}-b^{u}\mathbf{G}_{\nu})^{-1},\quad t\leq T,\label{eq:def_psi_t}
\end{equation}
where:

\begin{itemize}
\item $\mathbf{G}_{\nu}$ is the integral operator induced by $G_{\nu}(.)$
and $\mathbf{G}_{\nu}^{*}$ the adjoint operator,
\item $a^{u}:=\frac{1}{2}(u^{2}-u)$ and $b^{u}:=\kappa_{\nu}+\eta_{\nu}u\rho_{I\nu},$
\item $\tilde{\mathbf{\Sigma}}_{t}^{u}$ is the adjusted covariance integral
operator such that 
\[
\tilde{\mathbf{\Sigma}}_{t}^{u}=(\text{id}-b^{u}\mathbf{G}_{\nu})^{-1}\mathbf{\Sigma}_{t}(\text{id}-b^{u}\mathbf{G}_{\nu}^{*})^{-1},
\]
with $\mathbf{\Sigma_{t}}$ the integral operator associated with
the covariance kernel given by 
\[
\Sigma_{t}(s,w)=\eta_{\nu}^{2}\int_{t}^{T}G_{\nu}(s,z)G_{\nu}(w,z)dz,\;t\leq s,u\leq T.
\] 
\end{itemize}
Note that using similar arguments to those in the proof of \cite[Lemma 5.6.]{abi2021markowitz}, we can prove that $\mathbf{\Psi}_{t}^{u}$ is well defined and is a bounded linear operator acting on $L^{2}([0,T],\mathbb{C})$. 
\end{defn}

Based now on $g_{t}$ and $\mathbf{\Psi}_{t}^{u}$, we deduce a semi-explicit
expression of the characteristic function. The methodology used in
\cite{key-4} to deduce the analytical expression of the characteristic
function can be extended to a framework with stochastic interest rates.
Nevertheless, to extend these results, we need to also account for correlation between processes and consider in this sense the process
$h_{t}^{u}(.)$ defined, for $u\in\mathbb{C}$ such that $0\leq\mathcal{\mathfrak{R}}(u)\leq1,$
as
\[
h_{t}^{u}(s):=g_{t}(s)+1_{t\leq s}\bigg(\rho_{Ir}\eta_{r}B_{G_{r}}(s,T)-\int_{t}^{s}G_{\nu}(s,w)\:(b^{u}\rho_{Ir}-u\eta_{\nu}\rho_{\nu r})\eta_{r}B_{G_{r}}(w,T)\:dw\bigg),\;s,t\leq T,
\]
with $B_{G_r}$ given by \eqref{eq:BGr}. In particular when $\rho_{Ir}=\rho_{\nu r}=0$, $h_{t}^{u}(.)$ reduces
to $g_{t}(.)$ and we recover the framework of \cite{key-4}. 

\begin{thm}
\label{thm:chf_general} Let $g_{0}(.)$ be given by \eqref{eq:g_0_t}
and $G_{\nu}(.)$ a Volterra kernel as in Definition \ref{def:L2_kernel-1}. Fix
$u\in\mathbb{C}$ such that $0\leq\mathcal{\mathfrak{R}}(u)\leq1,$
then, for all $t\leq T$,
\begin{equation}
\E^{\mathbb{Q}^{T}}\bigg[\exp\bigg(u\log\frac{I_{T}^{T}}{I_{t}^{T}}\bigg)\bigg|\mathcal{F}_{t}\bigg]=\exp\left(\phi_{t}^{u}+\chi_{t}^{u}+\langle h_{t}^{u},\mathbf{\Psi}_{t}^{u}h_{t}^{u}\rangle_{L^{2}}\right),\label{eq:chf_explicit}
\end{equation}
with:
\end{thm}
\begin{itemize}
\item $\mathbf{\Psi}_{t}^{u}$ given by \eqref{eq:def_psi_t},
\item $\phi_{t}^{u}=-\int_{t}^{T}\text{Tr}(\mathbf{\Psi}_{s}^{u}\dot{\mathbf{\Sigma}_{s}})ds,$
where $\dot{\mathbf{\Sigma}_{t}}$ is the strong derivative of $t\to\Sigma_{t}$
induced by the kernel i.e. 
\[
\dot{\mathbf{\Sigma_{t}}}(s,u)=-\eta_{\nu}^{2}G_{\nu}(s,t)G_{\nu}(u,t),\;a.e.
\]
\item $\chi_{t}^{u}=\frac{1}{2}(u^{2}-u)\int_{t}^{T}(1-\rho_{Ir}^{2})\eta_{r}^{2}B_{G_{r}}(s,T)^{2}ds.$
\end{itemize}
\begin{proof}
The proof is given in Section \ref{sec:Proofs}. 
\end{proof}

\subsection{Numerical implementation and Fourier pricing  \protect\label{subsec:numerical_implementation_general_kernels}}

For a practical perspective, we propose a natural way of approximating
the expression of the characteristic function when we consider general
Volterra kernels that lead, for example, to non-Markovian models such
as in the case of the fractional or shifted fractional kernels. This approximation is
based on a natural discretization of the inner product $\langle.,.\rangle_{L^{2}}$.
In this section, for sake of simplicity, we fix $t=0$. \\
\\
As proposed by \cite{key-4}, we can approximate the explicit
expression of the characteristic function \eqref{eq:chf_explicit}
by an approximate closed form solution using a simple disctretisation
of the operator $\mathbf{\Psi}_{0}^{u}$ \emph{\`a la} Fredholm. To
this end, for fixed $N\in\mathbb{N},$ we consider a partition  of $[0,T]$
with $t_{i}=i\frac{T}{N}$ for $i=0,\ldots,N.$ Then,
for $u\in\mathbb{C}$ such that $0\leq\mathcal{\mathfrak{R}}(u)\leq1,$
we can approximate the operator $\mathbf{\Psi}_{0}^{u}$ by a $N\times N$
matrix denoted $\mathbf{\Psi}_{0}^{u;\,N}$ and defined by
\begin{equation}
\mathbf{\Psi}_{0}^{u;\,N}:=(I_{N}-b^{u}(\mathbf{G}_{\nu}^{N})^{T})^{-1}a^{u}(I_{N}-2a^{u}\tilde{\mathbf{\Sigma}}_{0}^{u;\,N})^{-1}(I_{N}-b^{u}\mathbf{G}_{\nu}^{N})^{-1},\label{eq:approx_operator}
\end{equation}
where $I_{N}$ is the $N\times N$ identity matrix, $\mathbf{G}_{\nu}^{N}$
is the lower triangular matrix such that 
\begin{equation}
G_{\nu,\:ij}^{N}:=1_{j\leq i-1}\int_{t_{j-1}}^{t_{j}}G_{\nu}(t_{i-1},s)ds,\;1\leq i,j\leq N, \label{eq:G_matrix}
\end{equation}
and 
\[
\tilde{\mathbf{\Sigma}}_{0}^{u;\,N}:=\frac{T}{N}(I_{N}-b^{u}\mathbf{G}_{\nu}^{N})^{-1}\mathbf{\Sigma}_{0}^{N}(I_{N}-b^{u}(\mathbf{G}_{\nu}^{N})^{T})^{-1},
\]
with $\mathbf{\Sigma}_{0}^{N}$ the $N\times N$ discretized covariance
matrix given by 
\begin{equation}
\mathbf{\Sigma}_{0}^{i,j;\:N}:=\eta_{\nu}^{2}\int_{0}^{T}G_{\nu}(t_{i-1},s)G_{\nu}(t_{j-1},s)ds,\;1\leq i,j\leq N. \label{eq:Sigma_matrix}
\end{equation}

When we consider kernels of the form 
\[
G_\nu(t,s)=\frac{(t-s+\varepsilon)^{H-1/2}}{\Gamma(H+1/2)},~\varepsilon\geq 0, 
\] with $H\in(0,1)$ for $\epsilon=0$ and $H\in \mathbb{R}$ for $\varepsilon>0$, we observe that the $N\times N$ matrix \eqref{eq:G_matrix}-\eqref{eq:Sigma_matrix} can be computed in closed form
\[
G_{\nu,\:ij}^{N}=1_{j\leq i-1} \frac{(t_{i-1}-t_{j-1}+\varepsilon)^{\alpha}-(t_{i-1}-t_{j}+\varepsilon)^{\alpha}}{\Gamma(\alpha+1)},~1\leq i, j\leq N, 
\]
\[
\mathbf{\Sigma}_{0}^{i,j;\:N}=\frac{\eta_\nu^2~t_{i-1}~(t_{i-1}+\varepsilon)^{\alpha-1} (t_{j-1}+\varepsilon)^{\alpha-1}  }{\Gamma(\alpha)^2} F_1(1,1-\alpha,1-\alpha,2,\frac{t_{i-1}}{t_{j-1}+\varepsilon},\frac{t_{i-1}}{t_{i-1}+\varepsilon}),~\mathbf{\Sigma}_{0}^{i,j;\:N}=\mathbf{\Sigma}_{0}^{j,i;\:N}, ~1\leq i\leq j\leq N, 
\]
where $\alpha:=H+1/2$ and $F_1(a,b,c,d,x,y)$ is the Appell hypergeometric function of the first kind\footnote{The Appell hypergeometric function of the first kind can be evaluated numerically using the \texttt{appellf1} function available in the Python library \texttt{mpmath} (Python 3.13).} given by 
\[
F_1(\alpha,\beta,\gamma,\delta,x,y):=\sum_{m,n=0}^{\infty} \frac{(\alpha)_{m+n}(\beta)_m(\gamma)_n}{(\delta)_{m+n} m!n!} x^m y^n ,
\]
with $(.)_n$ the Pochhammer symbol such that 
\[
(x)_n=\prod_{k=1}^n (x-k+1). 
\]\\
Based on the approximate operator, we now deduce a closed form approximate solution to the characteristic function. Relying on 
\cite{abijaber_guellil}, we have that $\phi_{t}^{u}$, that appears in \eqref{eq:chf_explicit},
can be rewritten in term of the Fredholm determinant such that $$  \phi_{t}^{u}=i\pi k_t(u)-\log\left(\det\left(\text{id}-2a^u\tilde{\mathbf{\Sigma}}_{t}^{u}\right)^{1/2}\right),$$
with $k_t(u)\in [-1,1]$ where the sign is determined by the number of crossings of the nonpositive real line $\mathbb{R}_-$ by the Fredholm determinant. The reason we have rewritten $\phi_t^u$ is that this enables us to use the Fredholm determinant to calculate $\phi_t^u$, instead of having to discretize the trace-dependent expression which would require the computation of $\mathbf{\Psi}_{t}^{u}$ for several values of $t$. Moreover, when approximating operators by matrices, following \cite[Proposition 3.4.]{abijaber_guellil}, we obtain a prefactor-free formula such that 
\begin{equation*}
    e^{-i\pi k_0(u)} \det\left(\mathbf{\Phi}_{0}^{u;\,N}\right)^{1/2}=\det\left((\mathbf{\Phi}_{0}^{u;\,N})^{1/2}\right), 
\end{equation*}
with $\mathbf{\Phi}_{0}^{u;\,N}:=(I_{N}-2a^{u}\tilde{\mathbf{\Sigma}}_{0}^{u;\,N})$. 
In this case, if for $N\in\mathbb{N}$ and for $u\in\mathbb{C}$ such
that $0\leq\mathcal{\mathfrak{R}}(u)\leq1,$ we consider the $N-$dimensional
vector $h_{0}^{u;\,N}=(h_{0}^{u}(t_{0}),...,h_{0}^{u}(t_{N}))$ with
$t_{i}=i\frac{T}{N}$, a natural approximation of the characteristic
function is given by
\begin{equation}
\E^{\mathbb{Q}^{T}}\bigg[\exp\bigg(u\log\frac{I_{T}^{T}}{I_{0}^{T}}\bigg)\bigg]=\exp(\phi_{0}^{u}+\chi_{0}^{u}+\langle h_{0}^{u},\mathbf{\Psi}_{0}^{u}h_{0}^{u}\rangle_{L^{2}})\approx \frac{\exp\left(\chi_{0}^{u}+\frac{T}{N}(h_{0}^{u;\,N})^{T}\mathbf{\Psi}_{0}^{u;\,N}h_{0}^{u;N}\right)}{\det\left((\mathbf{\Phi}_{0}^{u;\,N})^{1/2}\right)},\label{eq:approx_fractional_kernel_op_disc}
\end{equation}
with:
\begin{itemize}
\item $\mathbf{\Psi}_{0}^{u;\,N}$ the approximate operator defined by \eqref{eq:approx_operator},
\item $\mathbf{\Phi}_{0}^{u;\,N}:=(I_{N}-2a^{u}\tilde{\mathbf{\Sigma}}_{0}^{u;\,N}),$ 
\item $\chi_{0}^{u}=\frac{1}{2}(u^{2}-u) (1-\rho_{Ir}^{2})\eta_{r}^{2}\int_{0}^{T}B_{G_{r}}(s,T)^2ds.$
\end{itemize}

\begin{rem}
To compute $\chi_{0}^{u}$, we need to compute the integral of the form $\int_{0}^{T}B_{G_{r}}(s,T)^2ds$, where $B_{G_{r}}(s,T)$ is defined by \eqref{eq:BGr} and admits an explicit expression for some well-known kernels as revealed by Table \ref{tab:Resolvents_example}. However, except for the indicator or exponential kernels, the integral $\int_{0}^{T}B_{G_{r}}(s,T)^2ds$ does not have an explicit expression and needs to be computed numerically.
\end{rem}

We now discuss the pricing of equity derivatives. From Section
\ref{sec:forward_measure}, we know that the arbitrage-free
price, at time $t\in[0,T]$, of a financial derivative $H_{T}=h(I_{T})\in L^{2}$
is given by
\[
V_{t}=P(t,T)\:\E^{\mathbb{\mathbb{Q}}^{T}}\bigg(h(e^{\log I_{T}^{T}})\:|\mathcal{F}_{t}\bigg),
\]
where $P(t,T)$ is the price of a zero-coupon bond of maturity $T$.
Exploiting the Fourier link between the density function and the characteristic
function, we can deduce how to compute the price
of vanilla options such as call and put options using a Fourier method.
For $u\in\mathbb{C}$, we consider $\varphi_{t}^{u}$ defined
by 
\[
\varphi_{t}^{u}=\E^{\mathbb{Q}^{T}}\bigg[\exp\bigg(iu\log\frac{I_{T}^{T}}{I_{t}^{T}}\bigg)\bigg|\mathcal{F}_{t}\bigg],
\]
then, from \cite{key-18}, the price at time $t\in[0,T],$ of a call
option $(I_{T}-K)_{+}$ is given by 
\begin{equation}
V_{t}=P(t,T)\:\bigg[I_{t}^{T}-\frac{K}{\pi}\int_{0}^{+\infty}\mathfrak{R}\bigg(e^{(iu+1/2)k_{t}}\varphi_{t}^{u-i/2}\bigg)\:\frac{du}{u^{2}+1/4}\bigg],\label{eq:lewis_price_call}
\end{equation}
where $k_{t}:=\log\frac{I_{t}^{T}}{K}$. Using a numerical integration
of \eqref{eq:lewis_price_call}, we can efficiently price call options
as well as calibrate the model to equity market data. In this paper,
we use the Gauss-Laguerre quadrature numerical integration which has
been demonstrated to be efficient in the context of option pricing
(see \cite{key-7}). In practice, for a general kernel function $G_{\nu}(.)$,
such as the fractional or path-dependent kernels, we have to approximate
the expression of the characteristic function of $\log\frac{I_{T}^{T}}{I_{t}^{T}}$
that appears in the integral \eqref{eq:lewis_price_call}. To this
end, we  use the discretization approach detailed above.
Note that such approximation procedure is not always necessary. In
fact, as we will see later in Section \ref{subsec:Riccati-equations-for-Markov-vol},
if we consider the Stein-Stein volatility model i.e. $G_{\nu}(t,s)=1_{s<t}$,
the analytic expression of the characteristic function is more explicit
as shown by Corollary \ref{corr:chf_stein-stein}.

\begin{algorithm*}
\caption{Lewis pricing with Gauss-Laguerre quadrature}\label{algo:lewis_price_implementation}
\begin{algorithmic}[1]  % The number tells where the line numbering should start
\State For a given quadrature level $L\in \mathbb{N}_0$, obtain Gauss-Laguerre quadrature discretization points and weights $(u_j,w_j)_{j=1,...,L}$.
\State For a given $N\in \mathbb{N}_0$ and for each point $u_j$, approximate the charateristic function of the log-forward index ${\varphi}_{0}^{u_j-i/2}$ by $\tilde{\varphi}_{0}^{u_j-i/2}$, using the operator discretization approach \eqref{eq:approx_fractional_kernel_op_disc}. 
\State Approximate the price of call options \eqref{eq:lewis_price_call} such as 
$$V_{0}\approx P(0,T)\:\bigg[I_{0}^{T}-\frac{K}{\pi}\sum_{j=1}^L \mathfrak{R}\bigg(e^{(iu_j+1/2)k_{0}}\tilde{\varphi}_{0}^{u_j-i/2}\bigg) \:\frac{w_j~e^{u_j}}{u_j^{2}+1/4}\bigg].$$ 
\Statex  {\textbf{Complexity:} The complexity of the algorithm is typically $\mathcal{O}(L^2 + LN^3)$, where $\mathcal{O}(L^2)$ corresponds to the computation of the Gauss-Laguerre nodes and weights, and $\mathcal{O}(LN^3)$ arises from the evaluation of the characteristic function at each quadrature point using the operator discretisation method, which requires inversion and multiplication of $N \times N$ dense matrices.} 
\end{algorithmic}
\end{algorithm*}

The numerical implementation of the \cite{key-18} pricing method with Gauss-Laguerre quadrature, is summarized in Algorithm \ref{algo:lewis_price_implementation}. Care must be taken when selecting the quadrature level $L$ for approximating the integral. As emphasized by \cite{boyarchenko2024correct}, the same quadrature level cannot be used uniformly across all maturities. In particular, shorter maturities require a higher number of quadrature points. This is due to the fact that the integrand in \eqref{eq:lewis_price_call} decays to zero more slowly for small maturities than for large ones, as illustrated in Figure \ref{fig:integrand_decreasing} for the Stein-Stein model\footnote{In this model, the characteristic function of the log-forward index admits an explicit expression.}. Therefore, in what follows, we adopt $L=60$ for maturities $T<0.25$, and $L=40$ for $T \geq 0.25$. As shown in Figures \ref{fig:GL_pricing_T_002}, \ref{fig:GL_pricing_T_025}, and \ref{fig:GL_pricing_T_1}, the Lewis method with Gauss-Laguerre quadrature produces stable pricing results under the Stein-Stein model with these parameter choices.

The closed form approximate solution of the characteristic function \eqref{eq:approx_fractional_kernel_op_disc} has been proposed in different
papers \cite{key-4,key-5}. However, it is important to note that
theoretically, a general convergence result when $N$ tends to infinity
has not yet been demonstrated. We nevertheless
verify empirically the convergence of this approximation method as
$N\to\infty$. For this purpose, we restrict ourselves to a one-factor
Hull and White model for the interest rate i.e. $G_{r}(t,s)=1_{s<t}$
and we consider the fractional kernel for the volatility $G_{\nu}(t,s)=1_{s<t}\frac{(t-s)^{H_{\nu}-1/2}}{\Gamma(H_{\nu}+1/2)}$.
Without loss of generality, we fix arbitrarily the model parameters
to 
\begin{align*}
 & \kappa_{r}=-0.03,\:\eta_{r}=0.01,\:I_{0}^{T}=100,\:\nu_{0}=0.2,\:\kappa_{\nu}=0,\:\theta_{\nu}=0.1,\:\sigma_{\nu}=0.2,\\
 & \rho_{I\nu}=-0.7,\:\rho_{Ir}=-0.25,\:\rho_{\nu r}=-0.25.
\end{align*}

Table \ref{tab:op_disc_error} compares the error committed when approximating the characteristic function of the log-forward index in the Gauss-Laguerre numerical integration for $H_\nu=0.5$. We observe that the error rapidly decreases as the number of discretizations $N$ increases. Moreover, Figure \ref{fig:Implied-volatility-dynamics_op_discr} presents the
implied volatility dynamics generated by the operator discretization for
$H_{\nu}=0.5$ and Figure \ref{fig:Implied-volatility-dynamics_op_discr_H_03} for $H_{\nu}=0.3$. As expected, we observe
a relative fast convergence as the number of discretization factors
$N$ increases. 

\begin{table}[h!]
\centering
\label{tab:mytable}
    \begin{tabular}{|c||c|c|c||c|c|c|}
    \hline
    & \multicolumn{3}{c||}{$T=0.05$} & \multicolumn{3}{c|}{$T=0.25$} \\
    \hline
    $k$\textbackslash $N$ & $10$ & $40$ & $100$ & $10$ & $40$ & $100$ \\ \hline
    $0$ & $1.06\times 10^{-4}$ & $5.65\times 10^{-5}$ & $4.71 \times 10^{-5}$ & $8.08 \times 10^{-4}$ & $2.58 \times 10^{-4}$ & $1.52 \times 10^{-4}$ \\
    $0.2$ & $4.36 \times 10^{-6}$ & $1.25 \times 10^{-6}$ & $5.73 \times 10^{-7}$ & $1.34 \times 10^{-3}$ & $3.55 \times 10^{-4}$ & $1.57 \times 10^{-4}$ \\
    $-0.2$ & $1.07 \times 10^{-7}$ & $1.68 \times 10^{-11}$ & $3.91 \times 10^{-12}$ & $2.58 \times 10^{-4}$ & $1.45 \times 10^{-5}$ & $3.25 \times 10^{-6}$ \\
    \hline
    \end{tabular}
    \caption{\protect \label{tab:op_disc_error} $\left|\sum_{j=1}^L \left[\mathfrak{R}\left(e^{(iu+1/2)k}\tilde{\varphi}_{0}^{u_j-i/2}\right)-\mathfrak{R}\left(e^{(iu+1/2)k}{\varphi}_{0}^{u_j-i/2}\right)\right]\:\frac{w_j~e^{u_j}}{u_j^{2}+1/4}\right|$, for $H_\nu=0.5$.}
\end{table}

\begin{figure}[H]
\centering{}\includegraphics[width=0.45\textwidth]{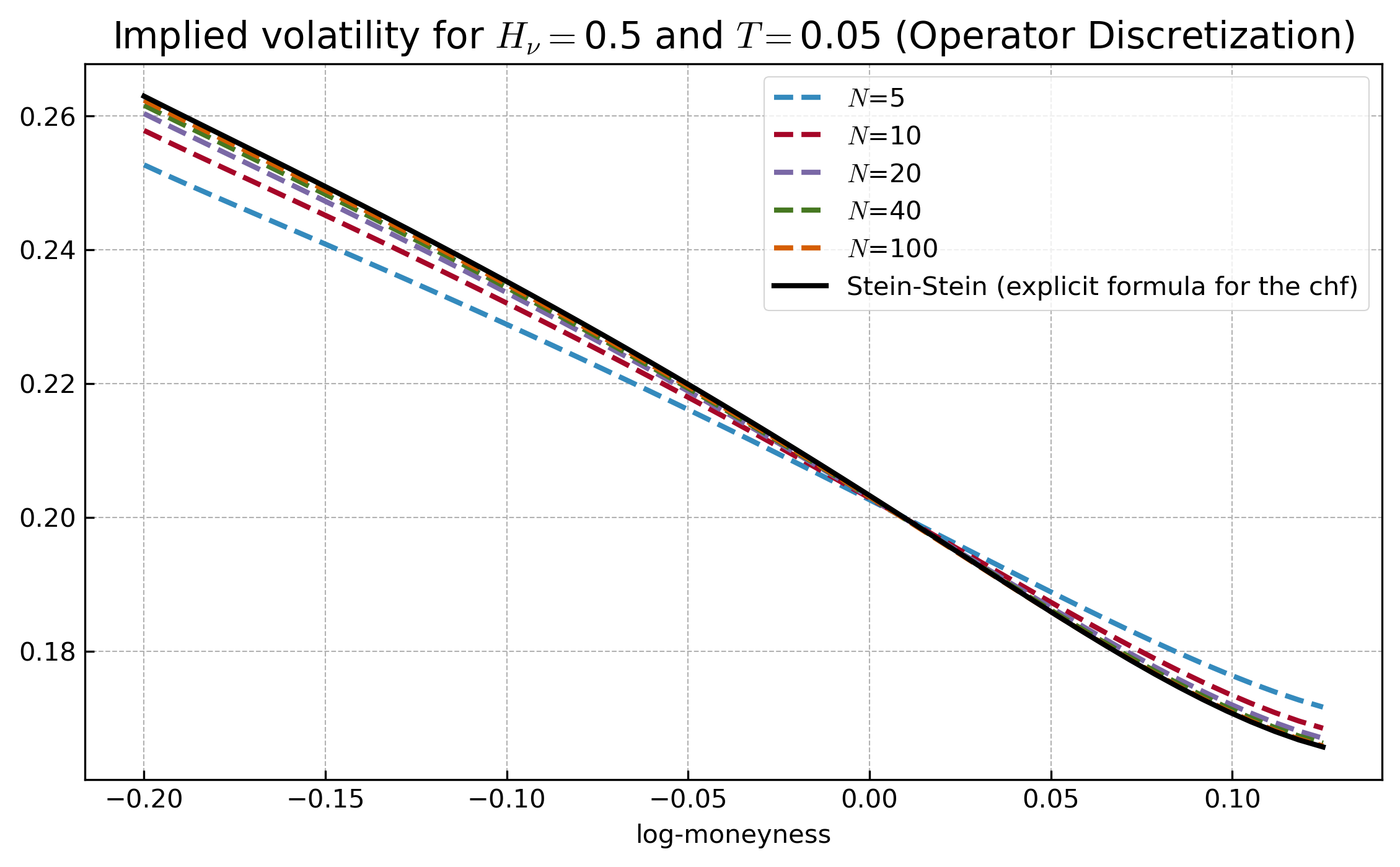}\includegraphics[width=0.45\textwidth]{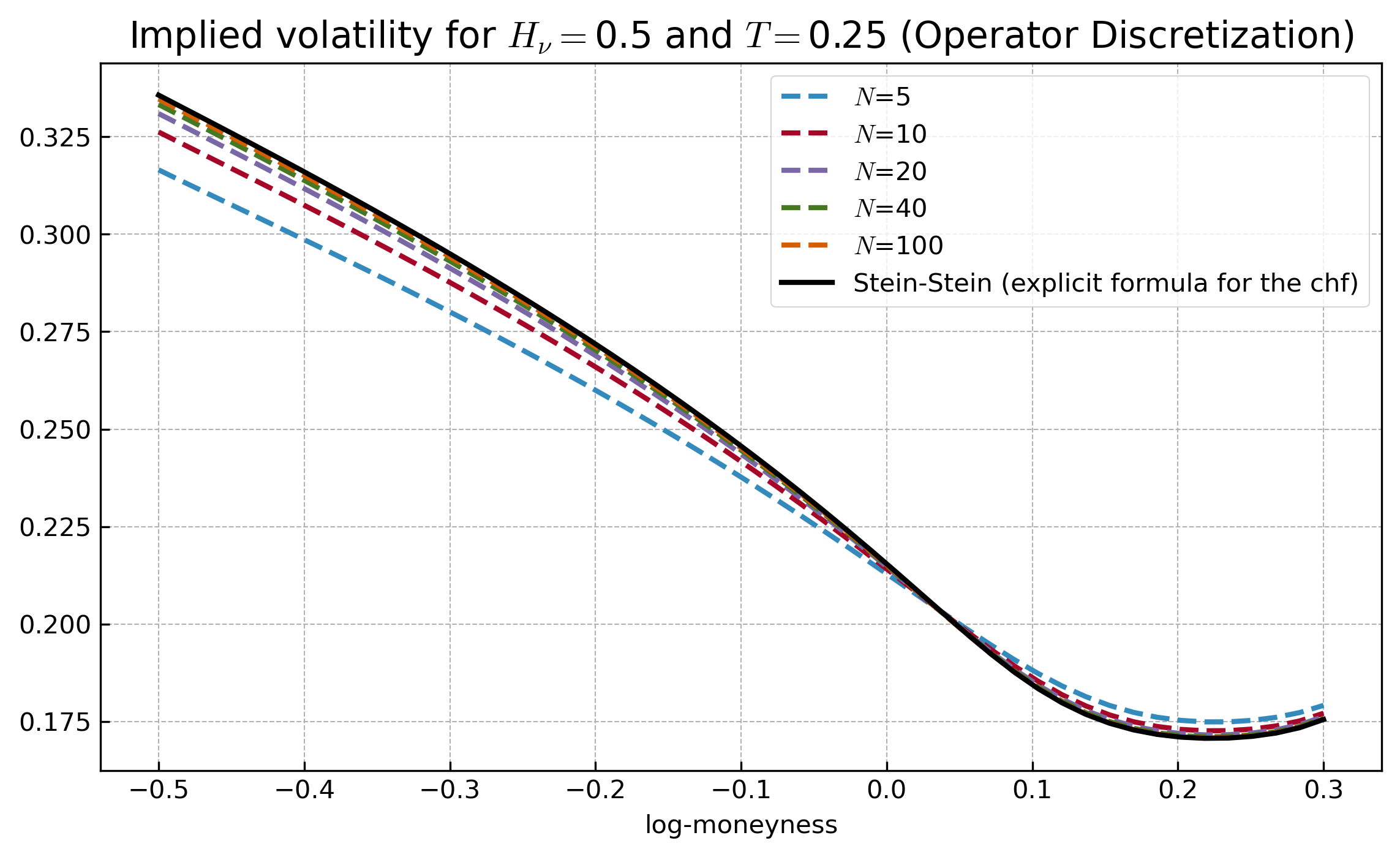} \\ \includegraphics[width=0.45\textwidth]{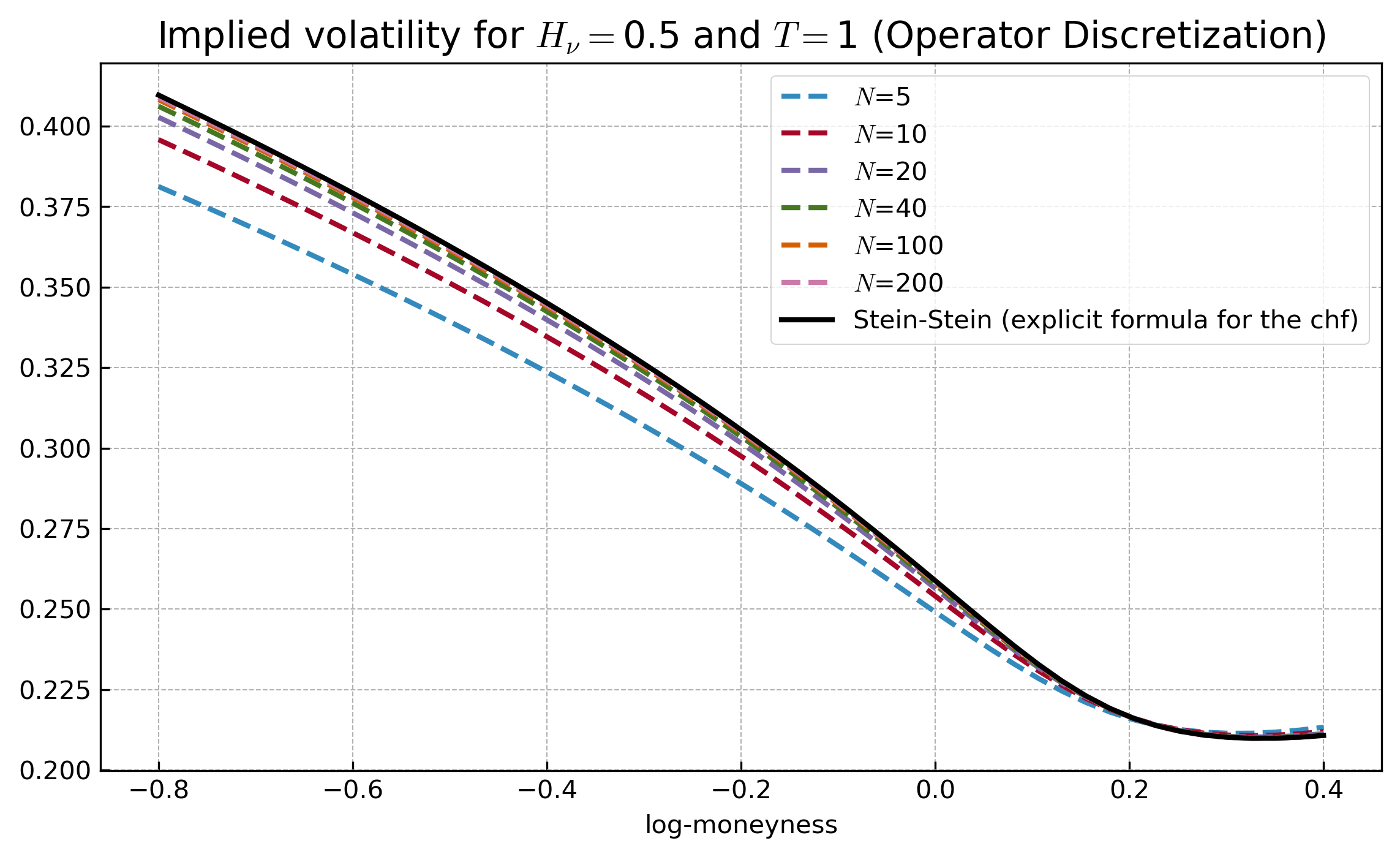} \caption{\protect\label{fig:Implied-volatility-dynamics_op_discr}Implied volatility
dynamics generated by the operator discretization method with $H_{\nu}=0.5$
and the Stein-Stein model.}
\end{figure}

\begin{figure}[H]
\centering{}\includegraphics[width=0.45\textwidth]{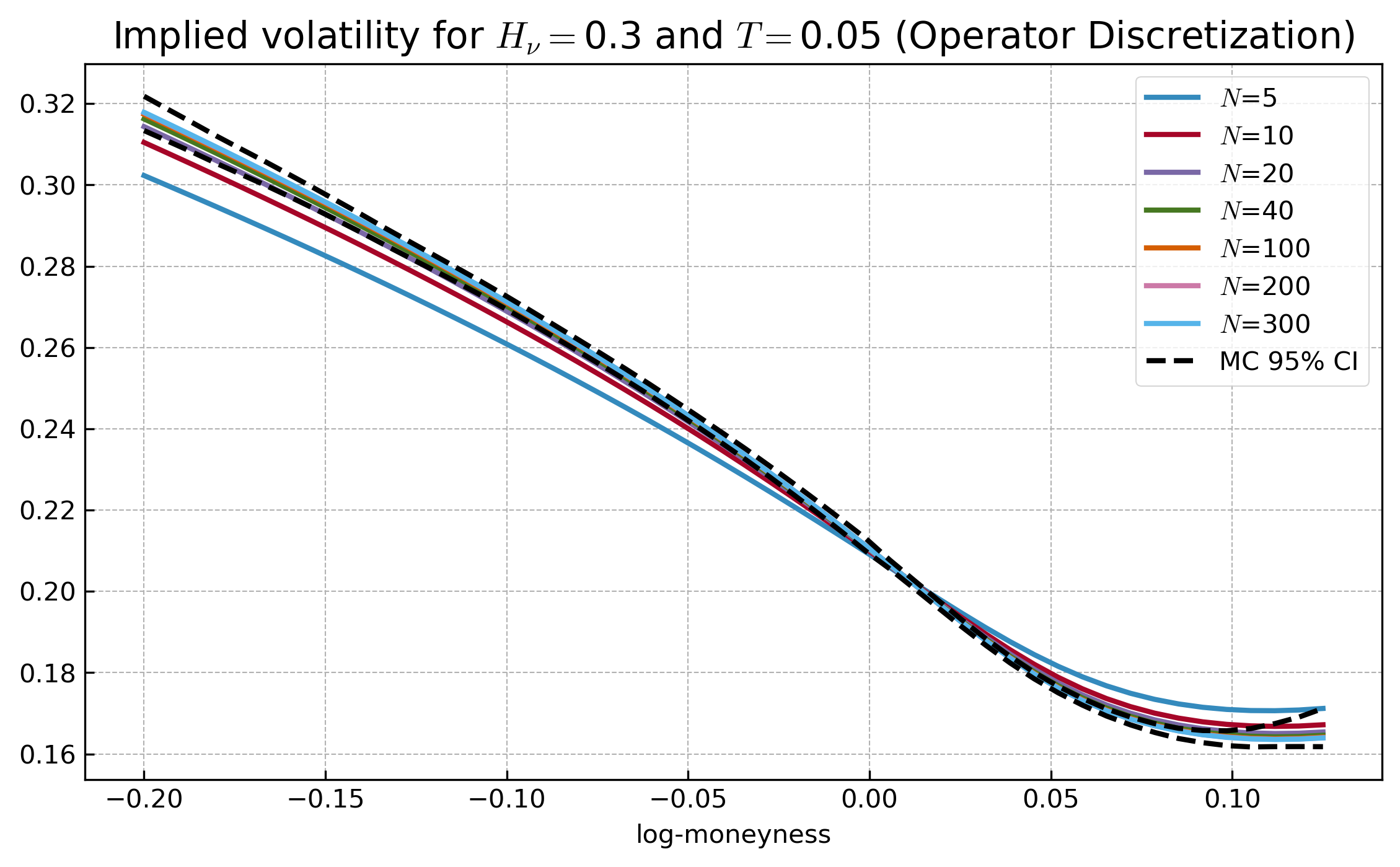}\includegraphics[width=0.45\textwidth]{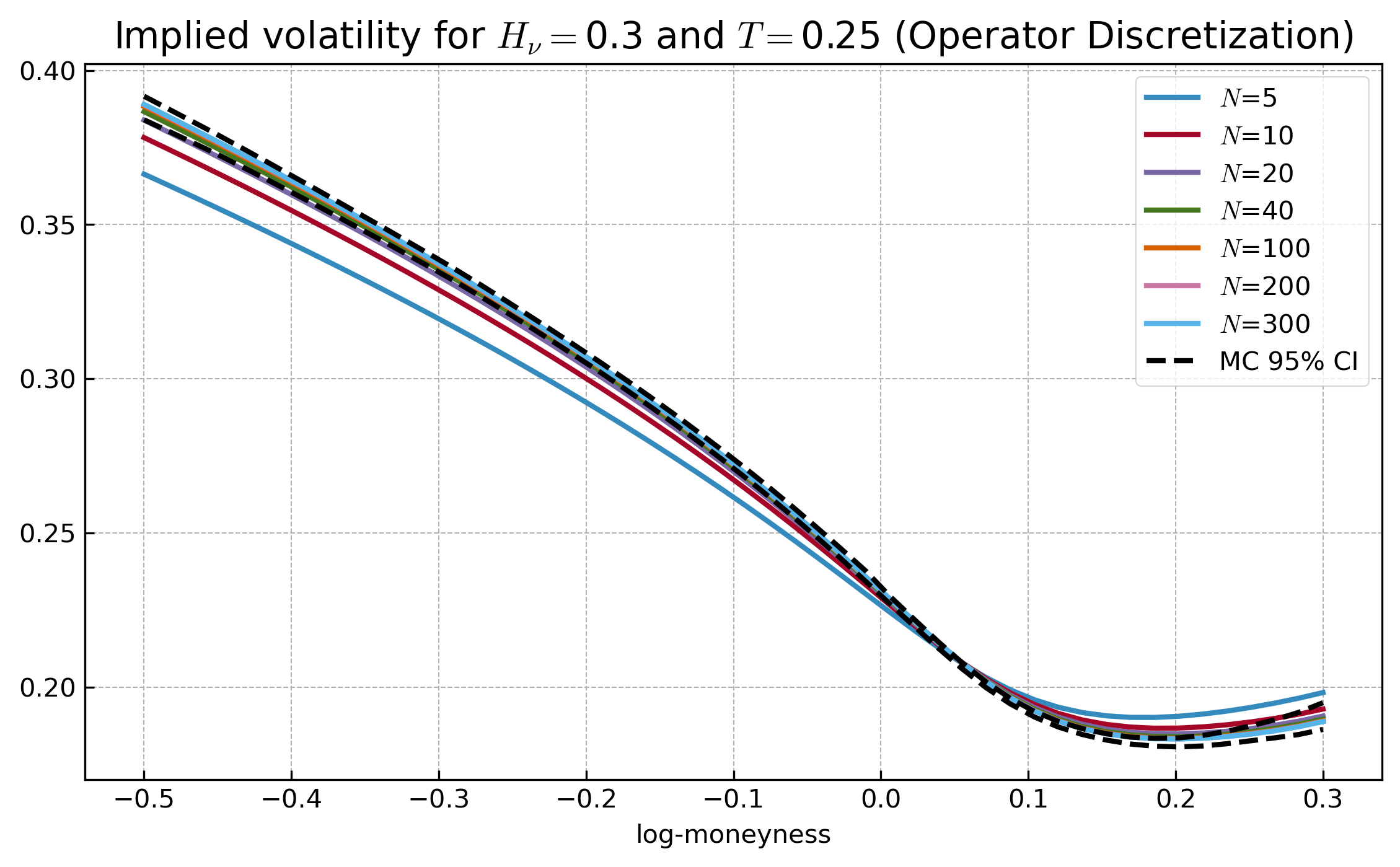} \\\includegraphics[width=0.45\textwidth]{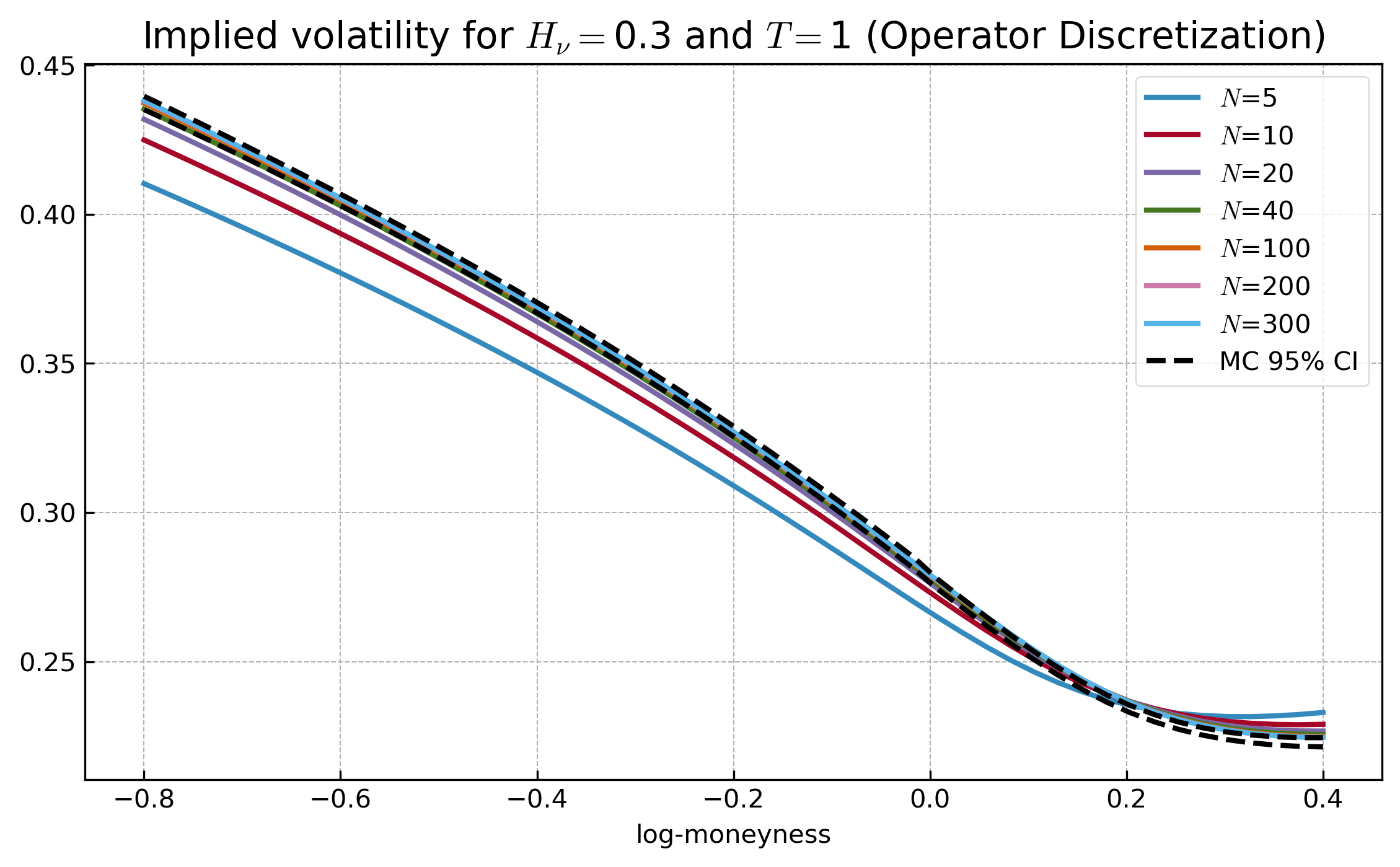}\caption{\protect\label{fig:Implied-volatility-dynamics_op_discr_H_03} Implied volatility dynamics generated
by the operator discretization method with $H_{\nu}=0.3$.
Monte Carlo confidence intervals are generated with $200\:000$ simulations
of risk processes using a Euler scheme with $500$ as discretization steps.}
\end{figure}

\section{Calibration to market data\protect\label{sec:Numerical-illustration-in}}

In this section, we highlight the relevance of our framework for calibration to 
market data. To this end, we place ourselves in an equity market context
where the process $(I_{t})_{0\leq t\leq T}$ represents an equity
stock index. Our calibration instruments from  market data, are  interest rate options, caps and floors, 
as well as equity vanilla call and put options. We now calibrate the models to the market data in order to validate
the relevance of the introduced framework. To this end, we consider
market data of 25/08/2022 consisting of USD 3M Libor yield curve and
ATM USD cap implied volatility data from Bloomberg, and S\&P500 implied
volatility data purchased from the CBOE website\footnote{https://datashop.cboe.com/.}.
The calibration procedure takes place in two stages. First we calibrate
the interest rate parameters using the USD 3M Libor yield curve and
ATM cap data. Then, using these calibrated interest rate parameters,
we calibrate the volatility and correlation parameters based on S\&P500
implied volatility data. We further note that the market data observed on 25/08/2022 exhibit the usual shapes encountered in practice for both the yield curve and implied volatility surfaces. As a result, we expect the qualitative conclusions of this calibration to remain valid for other market dates. \\

The calibration of the model parameters is performed using the standard Python optimization routine \texttt{scipy.optimize.minimize} with the \texttt{L-BFGS-B} algorithm. Since \texttt{L-BFGS-B} is a local optimization method, different initial values may lead to different local minima. To mitigate this sensitivity, we employ a multi-start procedure with $10$ independent initializations. Specifically, the initial values of the parameters are drawn independently from uniform distributions over their respective admissible ranges, which are reported in Table \ref{tab:calibration_ranges}. The optimization is then performed from each starting point, and we retain the parameter set yielding the lowest calibration error. Moreover, for the correlation parameters, we employ a sine/cosine parametrization, such that
\begin{equation*}
    \rho_{I\nu} = \cos(\alpha_{I\nu}), 
    \qquad
    \rho_{\nu r} = \cos(\beta_{\nu r}),
    \qquad \text{ and }  \rho_{I r}
    =
    \cos(\alpha_{I\nu})\cos(\beta_{\nu r})
    +
    \sin(\alpha_{I\nu})\sin(\beta_{\nu r})\cos(\gamma_{I r}), 
\end{equation*}
with $(\alpha_{I\nu},\beta_{\nu r},\gamma_{Ir})\in [0,\pi]^3$. This enables the correlation coefficients to be represented through angular parameters $(\alpha_{I\nu},\beta_{\nu r},\gamma_{Ir})$, while ensuring that the resulting correlation matrix remains positive semidefinite throughout the optimization.\\ 

\begin{table}[htbp]
\centering
\label{tab:calibration_ranges}
\begin{tabular}{lcc}
\hline
Parameter & Lower bound & Upper bound \\
\hline
$\kappa_r$     & $-2$       & $2$    \\
$\theta_r$     & $-0.5$     & $0.5$  \\
$\eta_r$       & $0.05$     & $0.75$  \\
$H_r$          & $0.01$     & $0.99$ \\
$\nu_0$        & $0.05$     & $0.5$  \\
$\theta_\nu$   & $-0.5$     & $0.5$  \\
$\kappa_\nu$   & $-2$       & $2$    \\
$\eta_\nu$     & $0.05$     & $0.75$  \\
$H_\nu$          & $0.01$     & $0.99$ \\
$\alpha_{I\nu}$  & $0$    & $\pi$  \\
$\beta_{\nu r}$   & $0$    & $\pi$  \\
$\gamma_{I r}$ & $0$    & $\pi$  \\
\hline
\end{tabular}
\caption{Parameter ranges used for the calibration.}
\end{table}

\begin{figure}[h!]
\centering{}\includegraphics[width=0.55\textwidth]{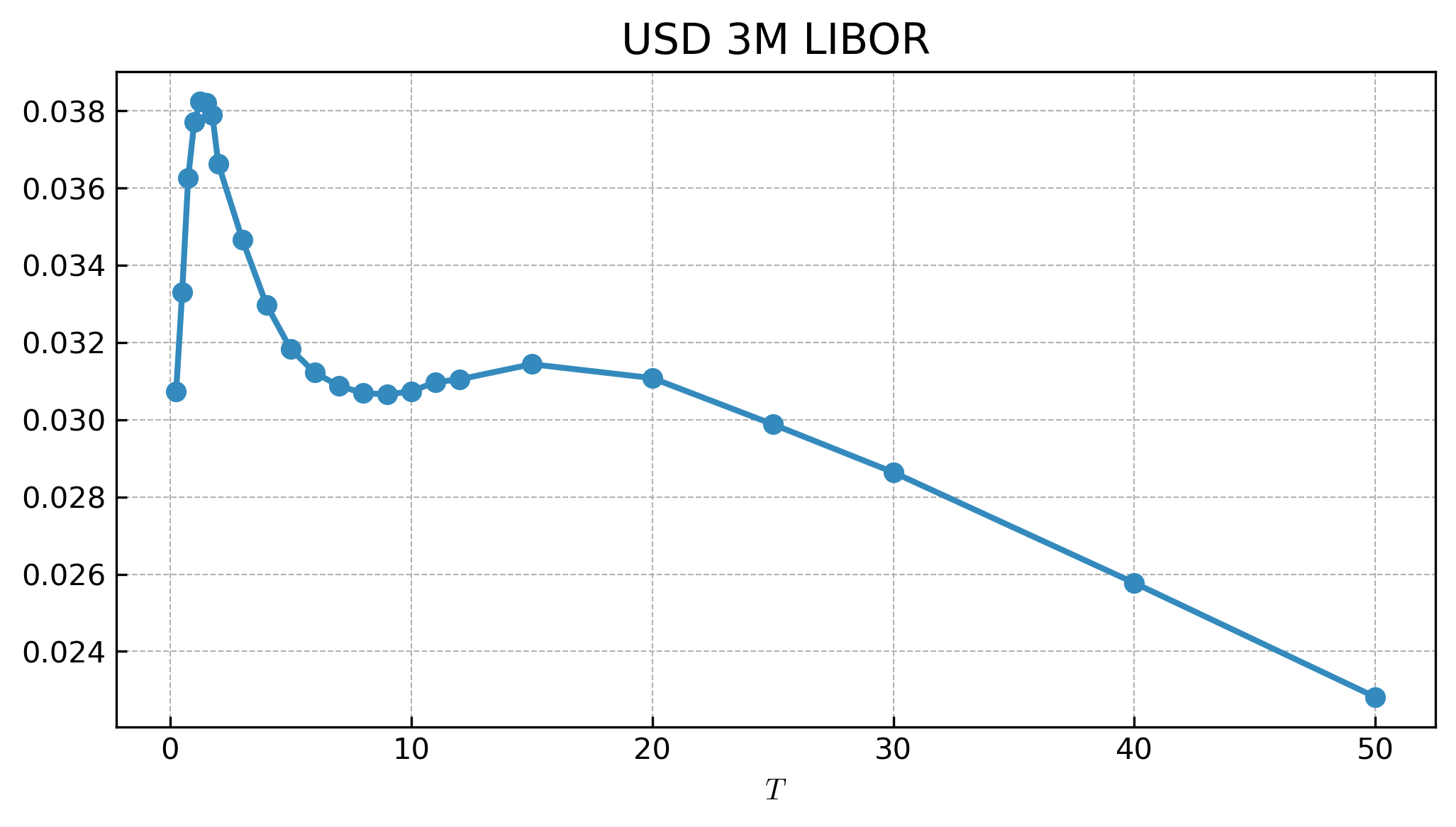}\caption{\protect\label{fig:USD-OIS-curve}USD 3M Libor yield curve of 25/08/2022.}
\end{figure}

\begin{figure}[h!]
\centering{}\includegraphics[width=0.5\textwidth,totalheight=0.22\textheight]{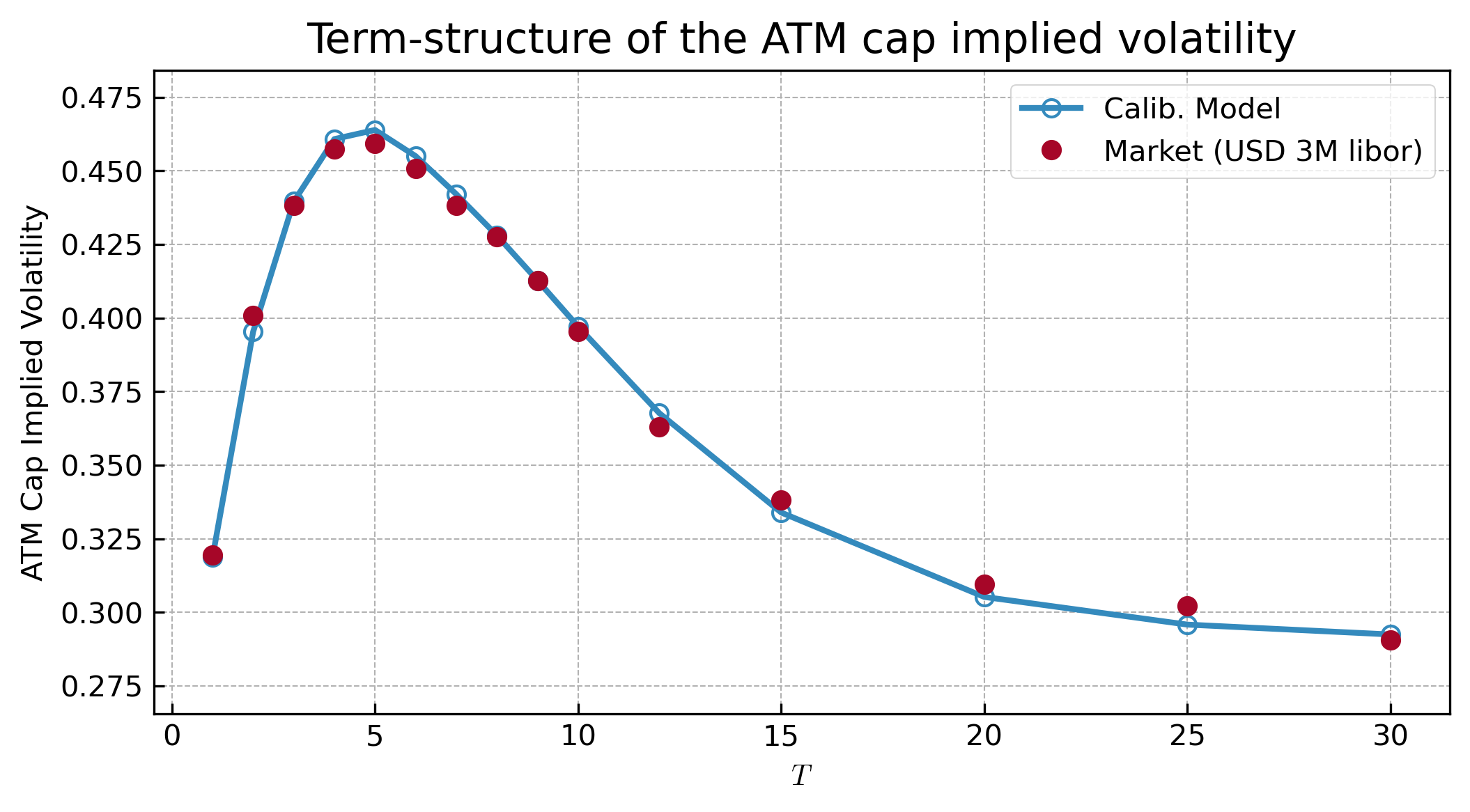}\\
\includegraphics[width=0.5\textwidth,totalheight=0.22\textheight]{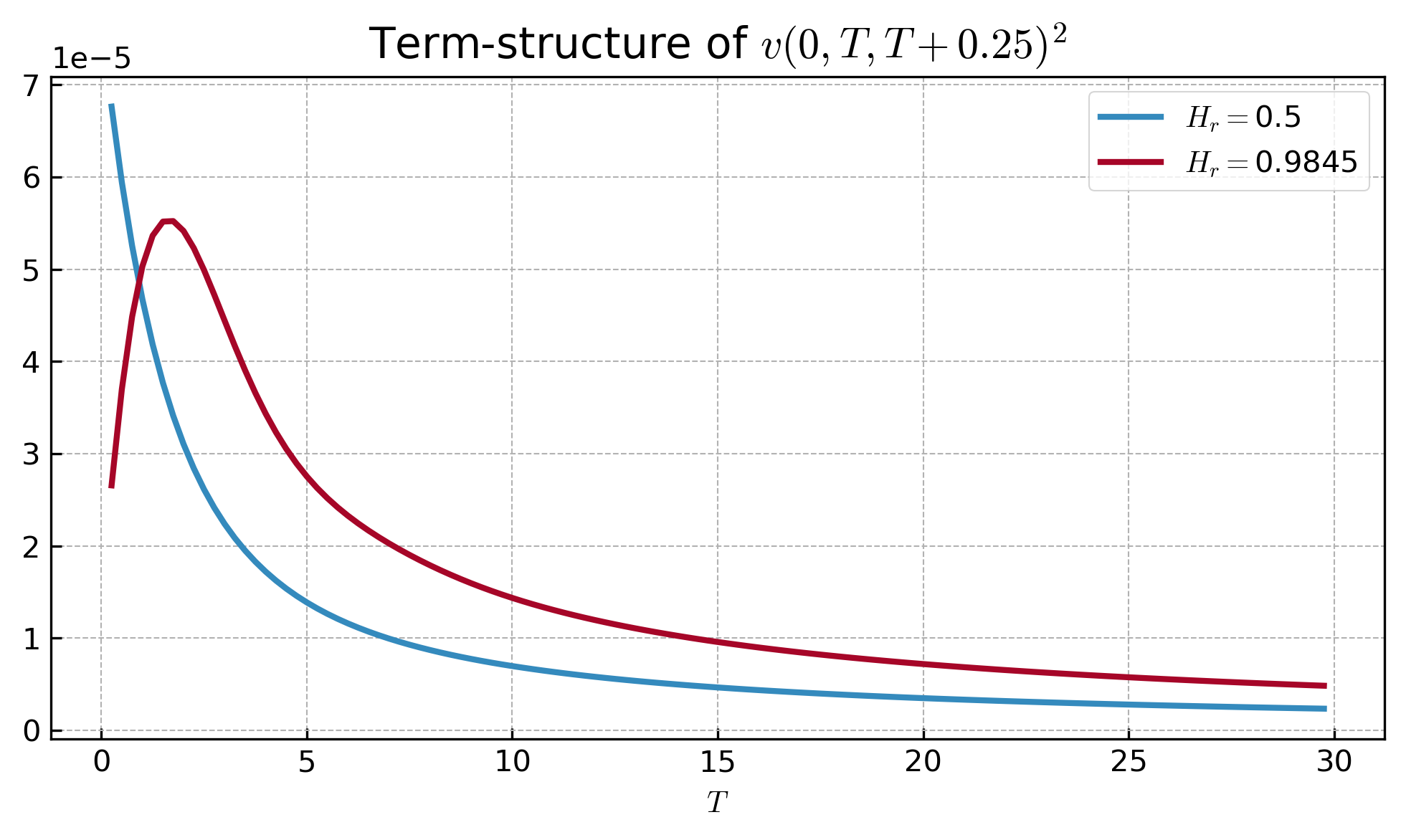}\includegraphics[width=0.5\textwidth,totalheight=0.22\textheight]{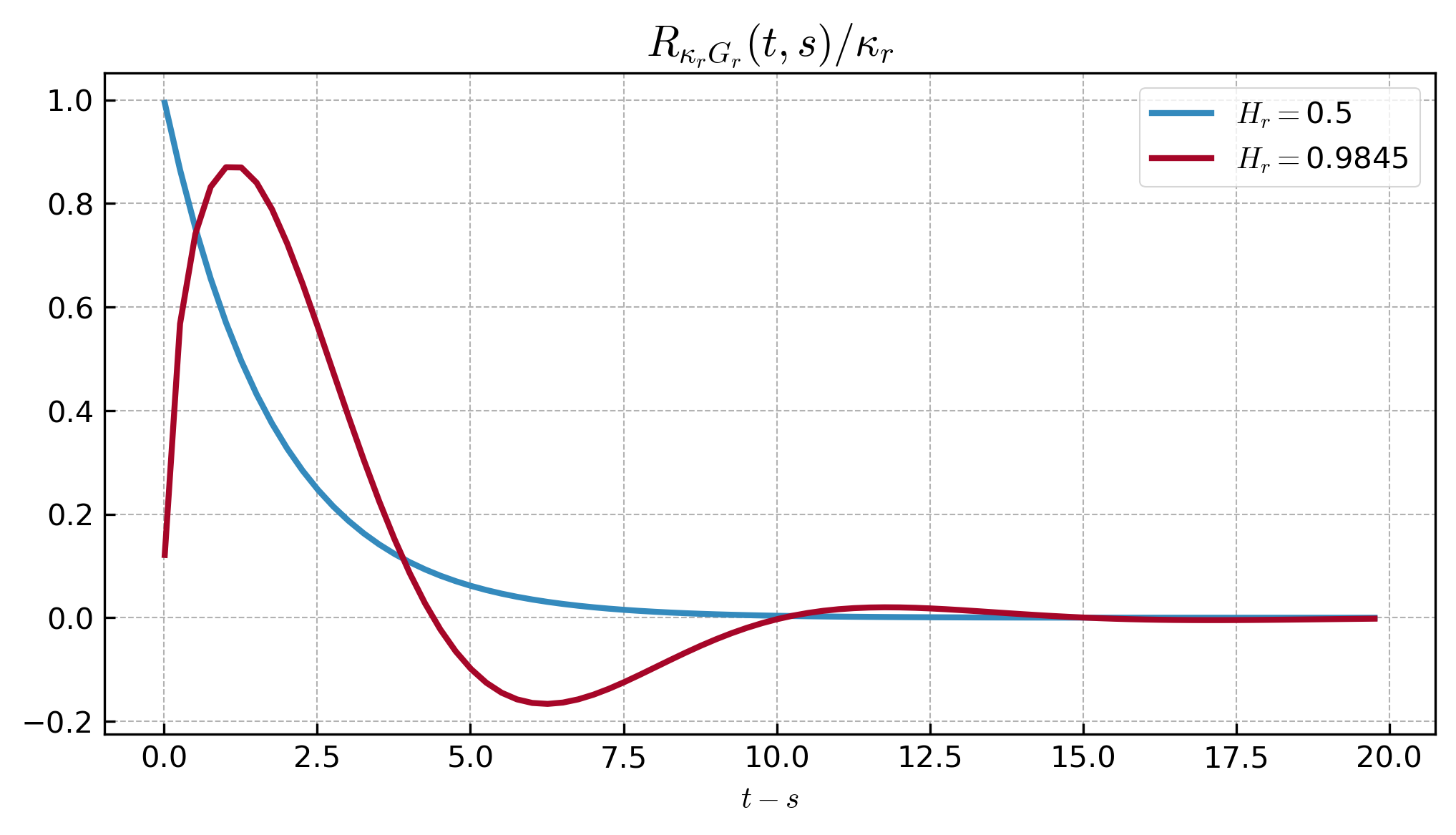}\caption{\protect\label{fig:USD-cap-implied}ATM cap implied volatility for
USD 3M Libor: calibrated fractional model vs market data of 25/08/2022.
RMSE: 0.3663\%. Calibrated parameters: $\hat{\kappa}_{r}=-0.5566,$
$\hat{\eta}_{r}=0.0377$ and $\hat{H}_{r}=0.9845$ (top); term-structure
of the ZC bond option pricing variance \eqref{eq:variance_ZC_bond_option_pricing}
generated with the calibrated parameters (bottom left); evolution
of the resolvent associated to fractional kernel (bottom right).}
\end{figure}

For the interest model, we consider the fractional kernel of the form
$G_{r}(t,s)=1_{s<t}\frac{(t-s)^{H_{r}-1/2}}{\Gamma(H_{r}+1/2)}.$ The data
we have for the calibration are the USD 3M Libor yield curve and the
ATM USD cap implied volatility for annual maturities ranging from
1 year to 30 years. Firstly, using \eqref{eq:link_r0_ZC} and the
USD 3M Libor yield curve (Figure \ref{fig:USD-OIS-curve}), the input
curve $r_{0}(t)$ can deduced such that $P(0,T)=P^{Market}(0,T).$
Then, since cap options are sum of zero-coupon bond options, we use
Proposition \ref{prop:pricing_ZC_options} and calibrate the interest
rate parameters by minimizing the root square error (RMSE) between
model and market ATM cap implied volatility. Results of this calibration
are displayed on Figure \ref{fig:USD-cap-implied} with the corresponding
calibrated parameters. The RMSE is 0.3663\% and we observe an excellent
fit of the data with $\hat{H}_{r}>1/2$ as the calibrated model perfectly
reproduces the humped shape of the market ATM cap implied volatility
term-structure with only three parameters. As revealed by Figure \ref{fig:USD-cap-implied},
this can be explained by the fact that the term-structure of the ZC
bond options pricing variance (\ref{eq:variance_ZC_bond_option_pricing}),
with $S=T+1/4,$ produces a humped shape\footnote{Due to humped behavior of the resolvent $\frac{R_{\kappa_{rG_{r}}}(t,s)}{\kappa_{r}}$
that drives $B_{G_{r}}$ and thus the pricing variance.} when $\kappa_{r}<0$ and $H_{r}>1/2$, which is totally not the case
when considering the Hull-White model (i.e. $H_{r}=1/2)$. Moreover,
as expected and revealed in Figure \ref{fig:Sample_auto_corr_r},
$H_{r}>1/2$ generates a smoother sample path than $H_{r}=1/2$, and
exhibits long-range dependencies. To check that long-term dependence makes sense, we consider historical data on the USD 3M Libor rate (daily close values) from 01/06/2021 to 30/09/2024, and examine the empirical auto-correlation structure. As we can see in Figure \ref{fig:Sample_auto_corr_r}, the USD rate exhibits strong persistence, which is perfectly in line with previous studies (see \cite{key-10,key-19}), and emphasizes that considering long-range dependence models seems appropriate for modeling the dynamics of short rates. Of course, other interest rate models can also capture the humped shape of the ATM cap implied volatility, see for instance \cite{levendorskii2005pseudodiffusions}; \cite[Section 3]{key-9}; \cite{ boyarchenko2007eigenfunction}. However, we want to stress that our model is very parsimonious since we perfectly capture the humped shape with only three parameters. Moreover, it is highly tractable, can account for long-range dependencies, and appears very suitable for a joint equity-rate calibration. Finally, although Libor data are used for illustrative purposes, the proposed framework is equally applicable to SOFR-based instruments.

\begin{figure}[H]
\begin{centering}
\includegraphics[width=0.55\textwidth]{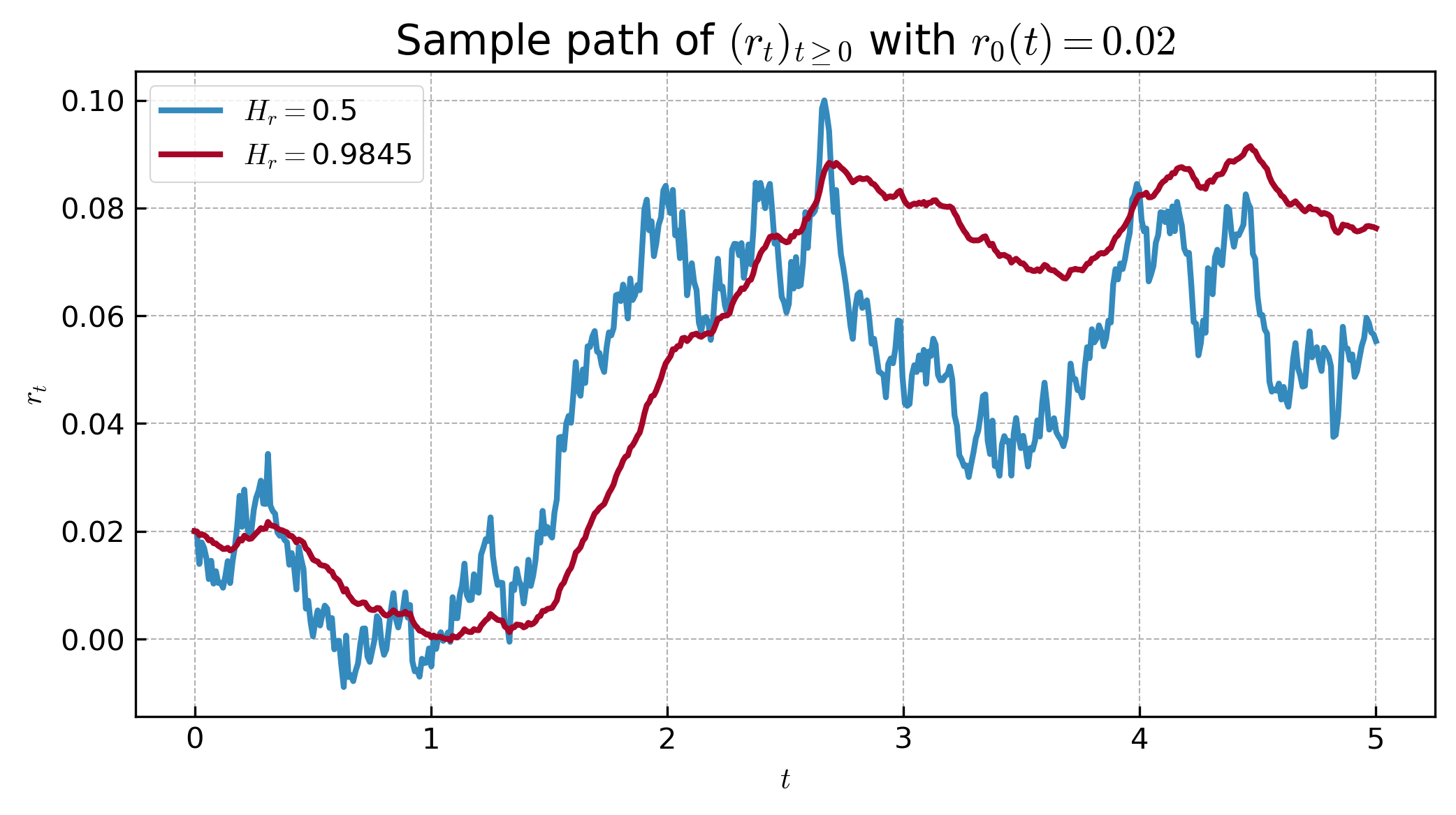}\includegraphics[width=0.55\textwidth]{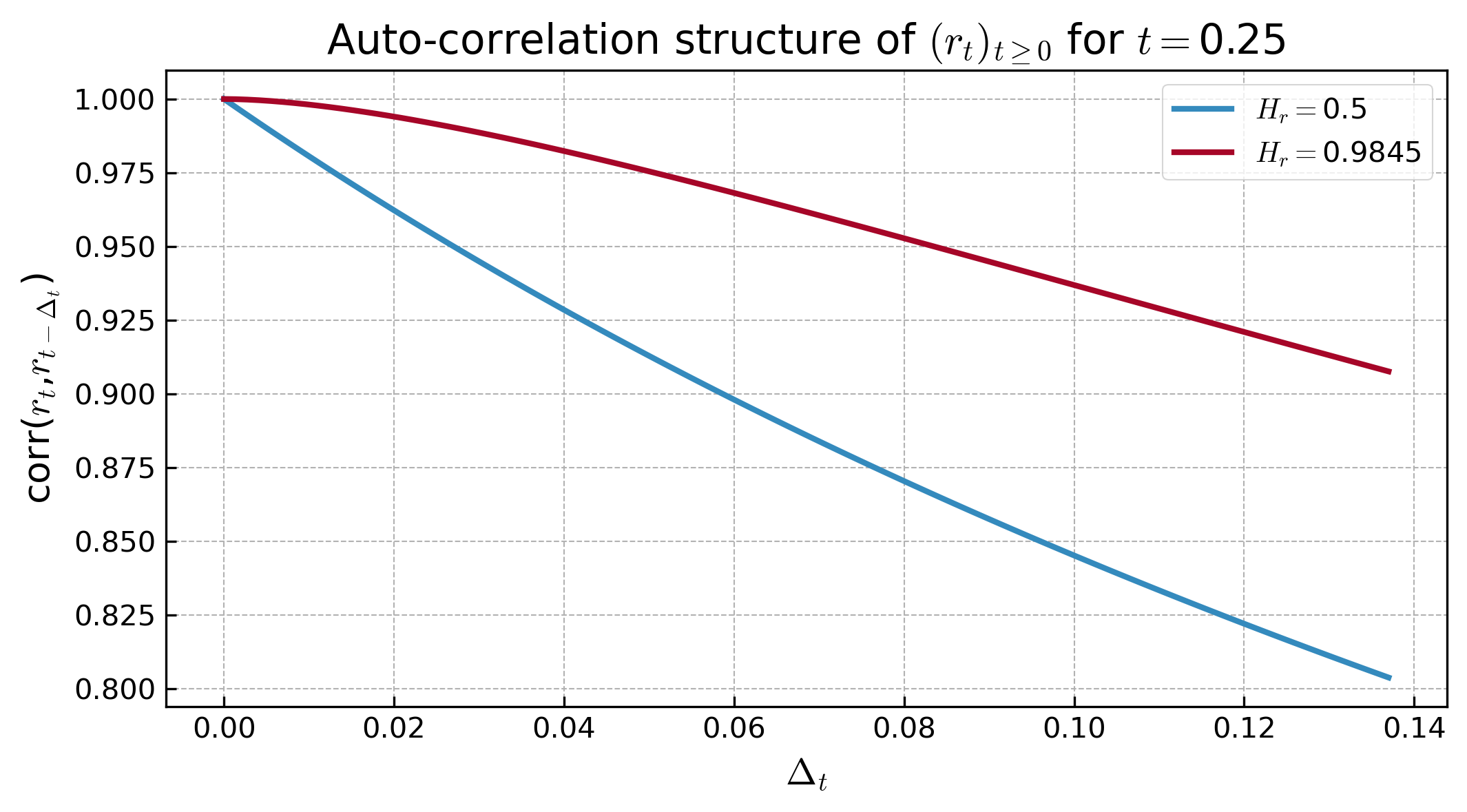}\\
\includegraphics[width=0.55\textwidth]{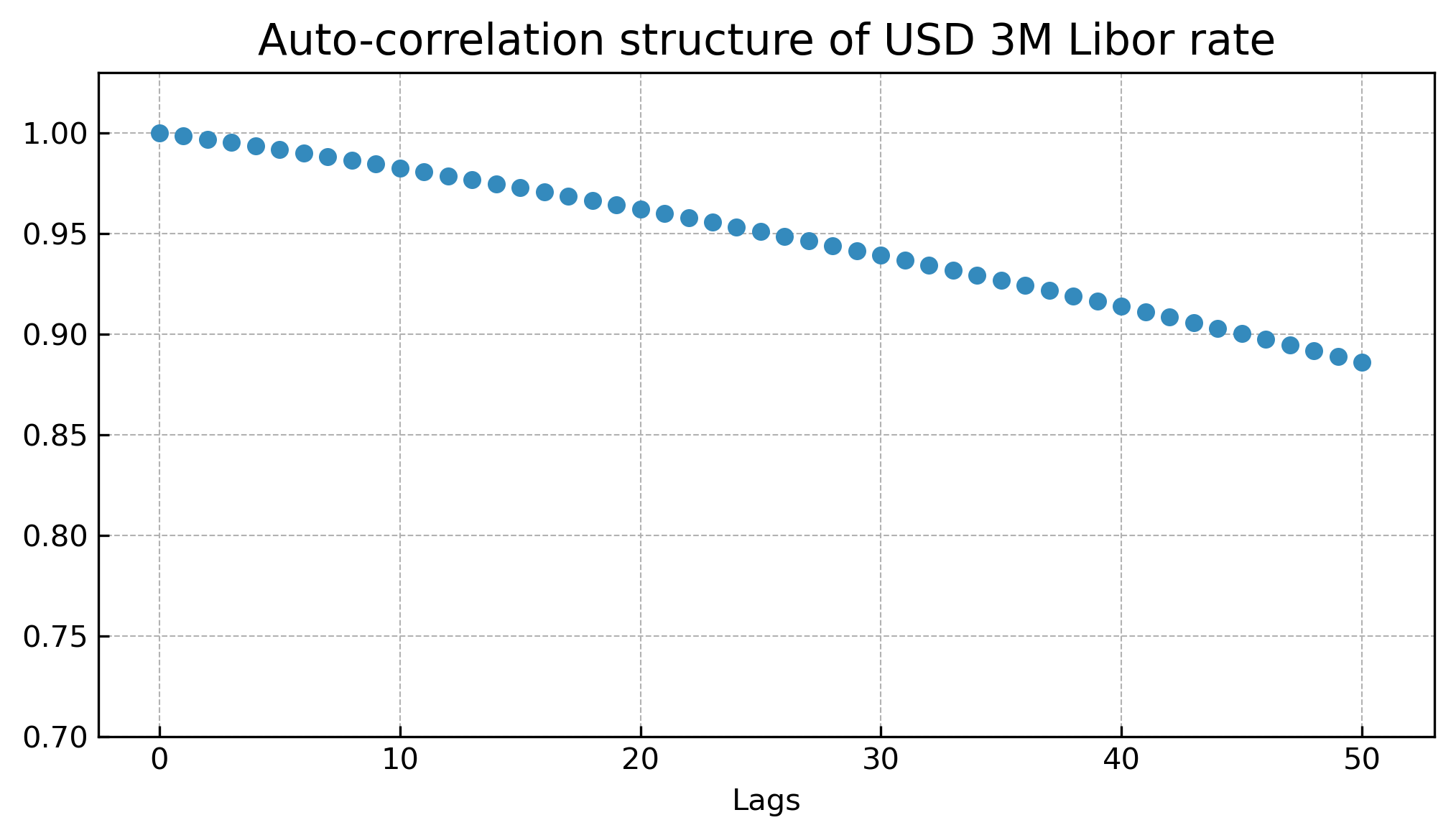}
\par\end{centering}
\caption{\protect\label{fig:Sample_auto_corr_r}Sample path of $(r_{t})_{t\geq0}$ generated with the calibrated parameters (upper left); Theoretical auto-correlation structure of $(r_{t})_{t\geq0}$ with the calibrated parameters (upper right); Empirical auto-correlation structure of USD 3M Libor rates (bottom). }
\end{figure}

We now consider the calibration of the volatility and correlation
parameters. To this end, we consider the fractional kernel of the
form $G_{\nu}(t,s)=1_{s<t}\frac{(t-s)^{H_{\nu}-1/2}}{\Gamma(H_\nu+1/2)}$ as
well as the shifted fractional kernel of the form $G_{\nu}(t,s)=1_{s<t}\frac{(t-s+\varepsilon)^{H_{\nu}-1/2}}{\Gamma(H_\nu+1/2)},$
with $\varepsilon=1/52.$ The data we have for the calibration are
the S\&P500 implied volatility data for different maturities ranging
from 1 week to 1.82 years and different strikes. Using the operator
discretization method with $N=40$ to approximate the characteristic
function of the log-forward index and a Fourier method, we can price call options efficiently. Therefore, we calibrate the volatility and
correlation parameters by minimizing the RMSE between model and S\&P500
implied volatility. To have a more parsimonious model, we decide to
fix $\kappa_{\nu}=0$ and $\rho_{r\nu}=0$. The impact of relaxing these restrictions on the calibration results is discussed in Appendix~\ref{sec:appendix_full_calib}.\\

Results of this
calibration are displayed on Figure \ref{fig:S=000026P500-implied-volatility_fractional}, \ref{fig:S=000026P500-implied-volatility_Path-dependent} and \ref{fig:ATM_IV_term_structure}, we
observe that, for both kernels, $\hat{H}_{\nu}<1/2$. The RMSE are
0.5321\% for the fractional kernel and 0.4602\% for the shifted fractional
kernel. Overall, the calibration is good for both kernels, especially
around the ATM, yet the shifted fractional kernel outperforms the fractional
kernel. Moreover, as we can observe on Figure \ref{fig:S=000026P500-ATM-skew},
the calibrated shifted fractional kernel reproduces the term structure of the ATM skew with a concave shape on the log-log scale for the chosen date, which is not the case for the fractional kernel. The power law decay for low maturities is not compatible with the concave shape of the ATM skew term structure, which is why the fractional kernel fails to reproduce the market term structure. We also observe that $\hat{H}_\nu^{\text{Shifted fractional}}<\hat{H}_\nu^{\text{Fractional}}$, indicating that the shifted kernel exhibits a relative faster decrease for longer maturities while $\varepsilon=1/52$ prevents it from having an exploding skew for shorter maturities, allowing it to better capture the concave shape of the skew (in log-log scale). Our results demonstrate that the rough model underperforms the non-rough model associated with the shifted fractional kernel in capturing the entire volatility surface, confirming the findings of \cite{ delemotte2023yet, guyon2023volatility,abi2024volatility}. To validate that the shapes of the implied volatilities generated with $N=40$ for the calibrated parameters are accurate, and to avoid drawing spurious conclusions, we compare them to the Monte Carlo confidence intervals. The results can be seen in Figures \ref{fig:IV_calib_vs_MC_fract} and \ref{fig:IV_calib_vs_MC_PD}, and as we can observe, the shapes obtained with $N=40$ are consistent with the confidence intervals.\\

We also observe that the calibrated correlations between interest rate and index processes
$\hat{\rho}_{Ir}$ are significant. To ensure that these estimated values make sense, we looked at the empirical 180 days rolling correlation between the S\&P500 index and the USD 3M Libor rate (daily close values). The results are displayed on Figure \ref{fig:rolling_correlation_SP500_USD_rate}.
We first observe that the correlation is not constant and almost never
zero. Secondly, we notice that, for the calibration date we have chosen,
the empirical rolling correlation is negative. This seems in line with the correlation values
obtained during calibration on the implied volatility surface of the
S\&P500, and thus assuming correlation between processes appears coherent
with market data. \\

Finally, we examine the effect of relaxing the restrictions $\kappa_\nu=0$ and $\rho_{\nu r}=0$ by comparing the parsimonious specification with the fully calibrated model, in which $\kappa_\nu$ and $\rho_{\nu r}$ are left unrestricted; see
Appendix~\ref{sec:appendix_full_calib} for further details. The results show that relaxing $\kappa_\nu=0$ leads to some changes in the calibrated parameters, notably in $H_\nu$, $\theta_\nu$, and $\eta_\nu$. In particular, allowing $\kappa_\nu$ to vary modifies the memory structure of the dynamic of $\nu$ and leads to a reallocation between the effects captured by $H_\nu$ and $\kappa_\nu$. Importantly, the main qualitative conclusions remain unchanged: the calibrated models remain in a regime with $H_\nu<0.5$, and the shifted fractional kernel continues to outperform the fractional kernel.

\begin{figure}[H]
\begin{centering}
\includegraphics[width=1\textwidth]{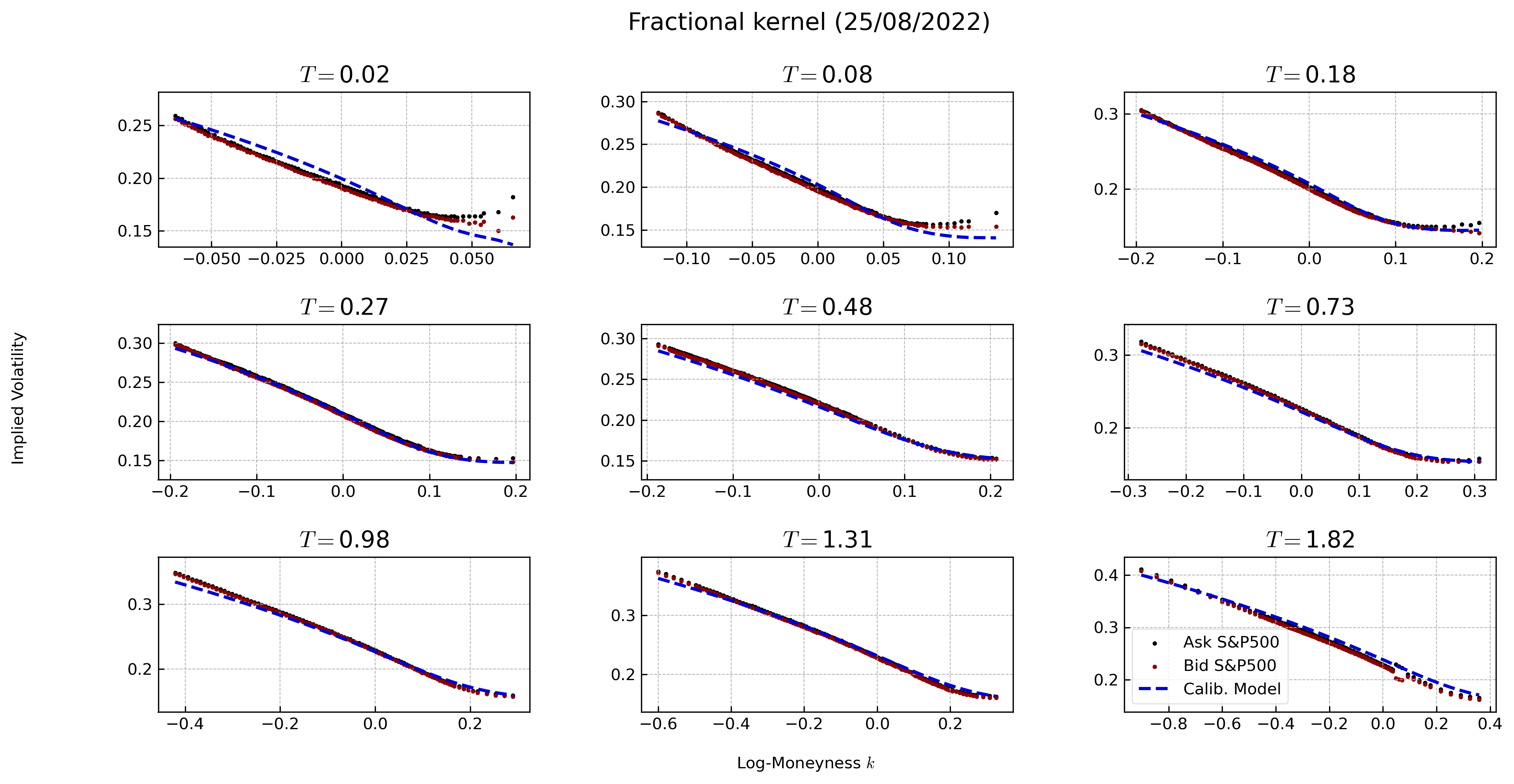}
\par\end{centering}
\caption{\protect\label{fig:S=000026P500-implied-volatility_fractional}S\&P500
implied volatility: calibrated fractional model vs market data of
25/08/2022. RMSE: 0.5321\%. Calibrated parameters: $\hat{\nu}_{0}=0.1964$,
$\hat{\theta}_{\nu}=-0.0248,$ $\hat{\eta}_{\nu}=0.2123,$ $\hat{\rho}_{I\nu}=-0.7981,$
$\hat{\rho}_{Ir}=-0.5971$ and $\hat{H}_{\nu}=0.2992$.}

\end{figure}

\begin{figure}[H]

\begin{centering}
\includegraphics[width=1\textwidth]{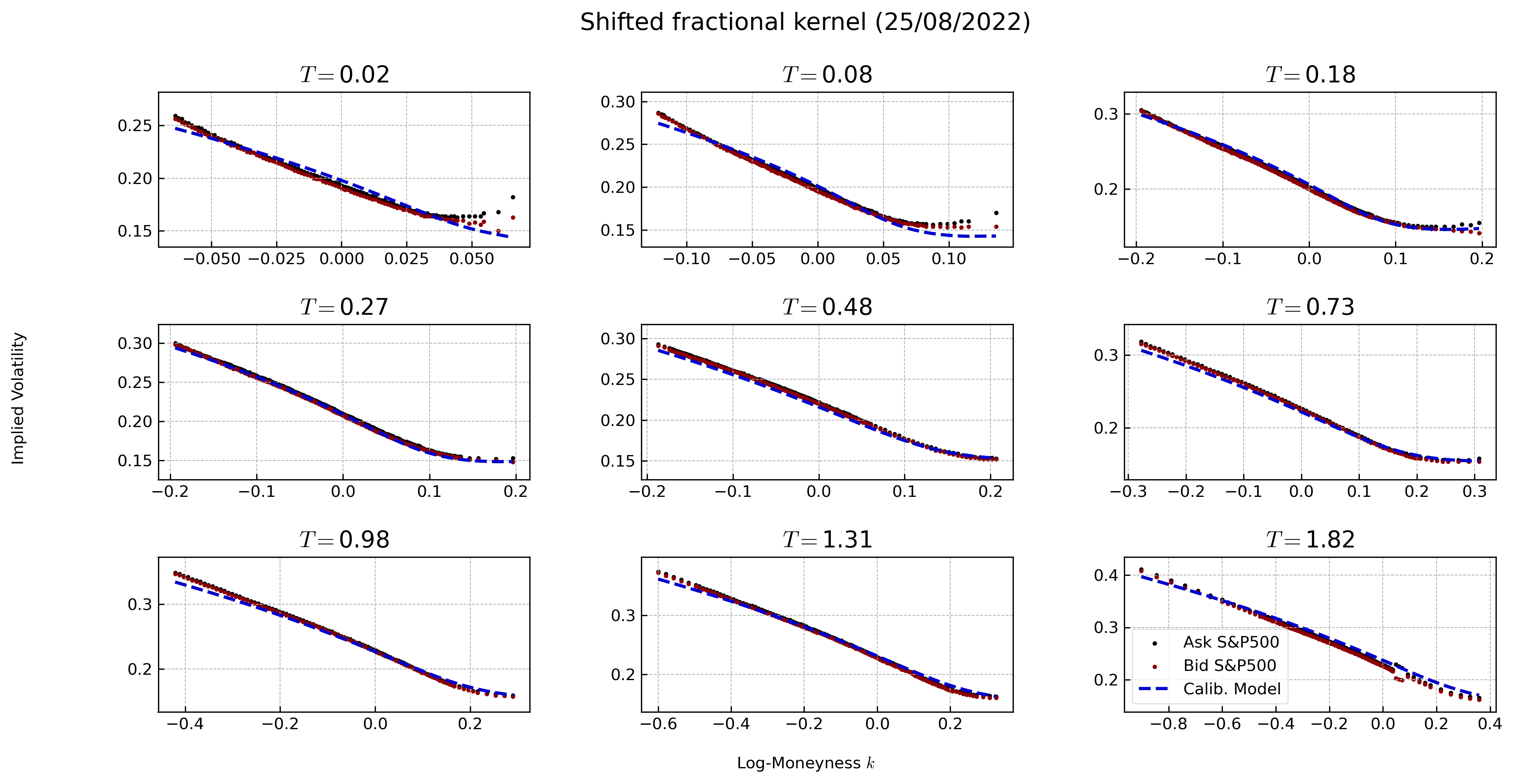}\caption{\protect\label{fig:S=000026P500-implied-volatility_Path-dependent}S\&P500
implied volatility: calibrated shifted fractional model vs market data
of 25/08/2022. RMSE: 0.4602\%. Calibrated parameters: $\hat{\nu}_{0}=0.1978$,
$\hat{\theta}_{\nu}=-0.0259,$ $\hat{\eta}_{\nu}= 0.2164,$ $\hat{\rho}_{I\nu}=-0.7868,$
$\hat{\rho}_{Ir}=-0.6107$ and $\hat{H}_{\nu}=0.2273$.}
\par\end{centering}
\end{figure}

\begin{figure}[H]
\begin{centering}
\includegraphics[width=0.8\textwidth]{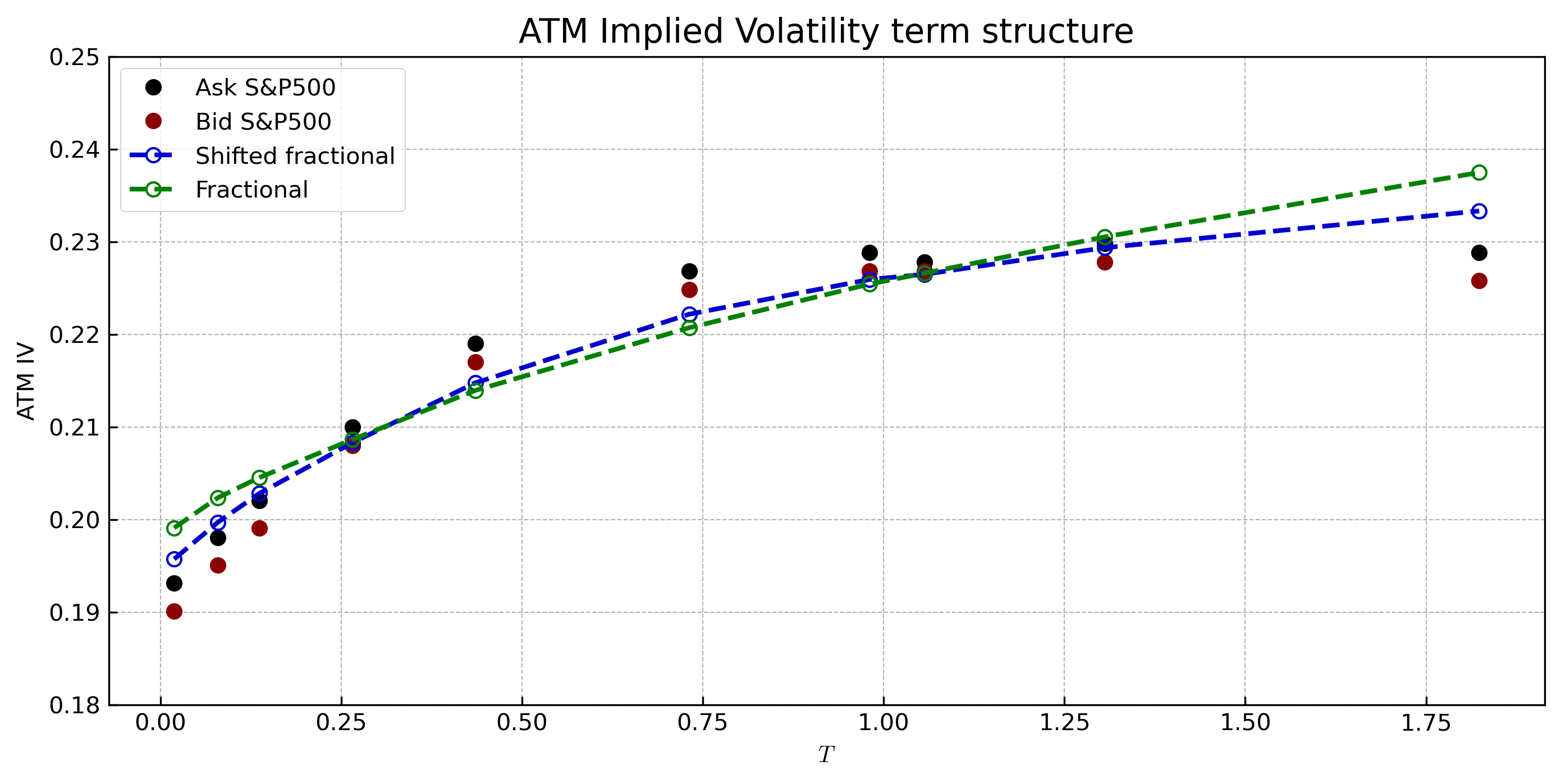}
\par\end{centering}
\centering{}\caption{\protect\label{fig:ATM_IV_term_structure} ATM implied volatility term structure:
calibrated models vs market data of 25/08/2022 with calibrated parameters.}
\end{figure}

\begin{figure}[H]
\begin{centering}
\includegraphics[width=0.75\textwidth]{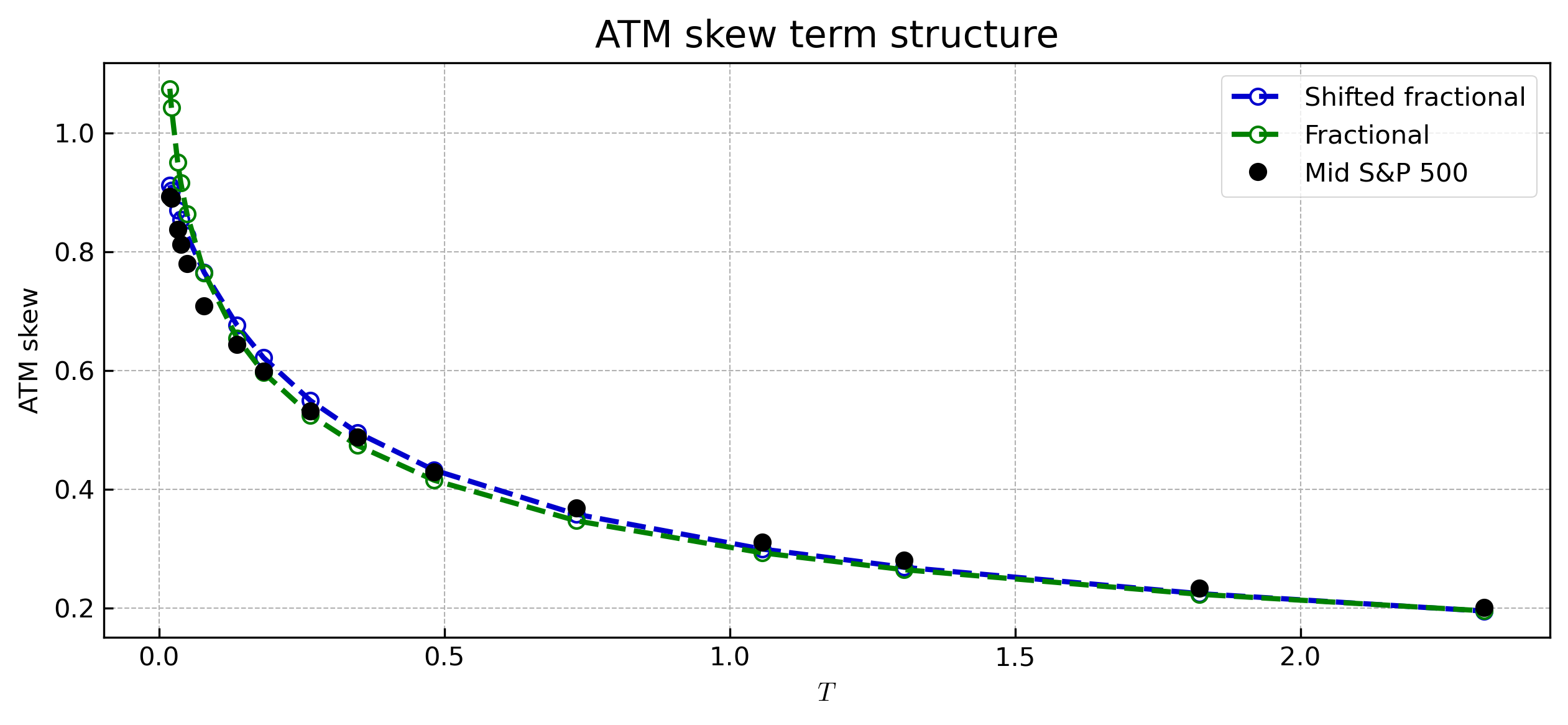}
\includegraphics[width=0.75\textwidth]{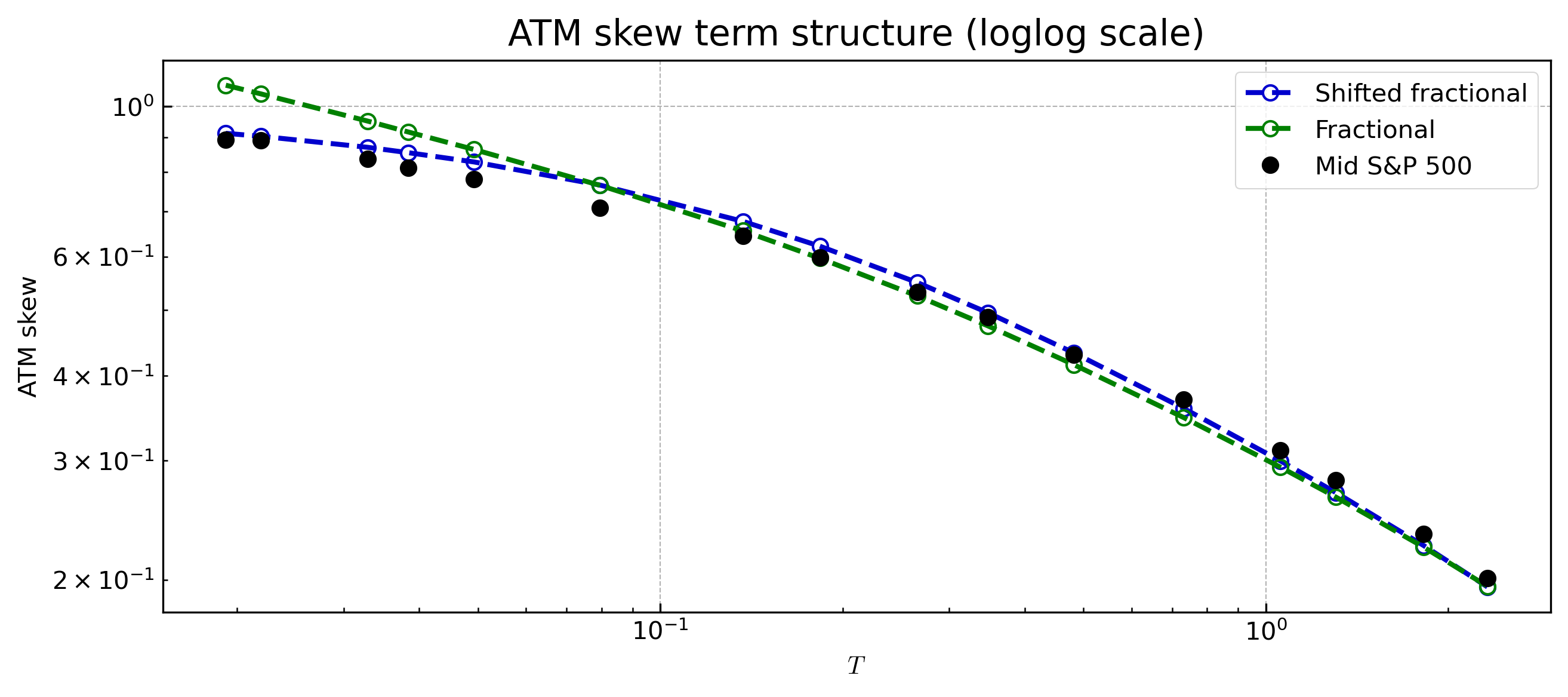}\\
\par\end{centering}
\centering{}\caption{\protect\label{fig:S=000026P500-ATM-skew} ATM skew term structure:
calibrated models vs market data of 25/08/2022 with calibrated parameters.}
\end{figure}

\begin{figure}[H]
\begin{centering}
\includegraphics[width=0.75\textwidth]{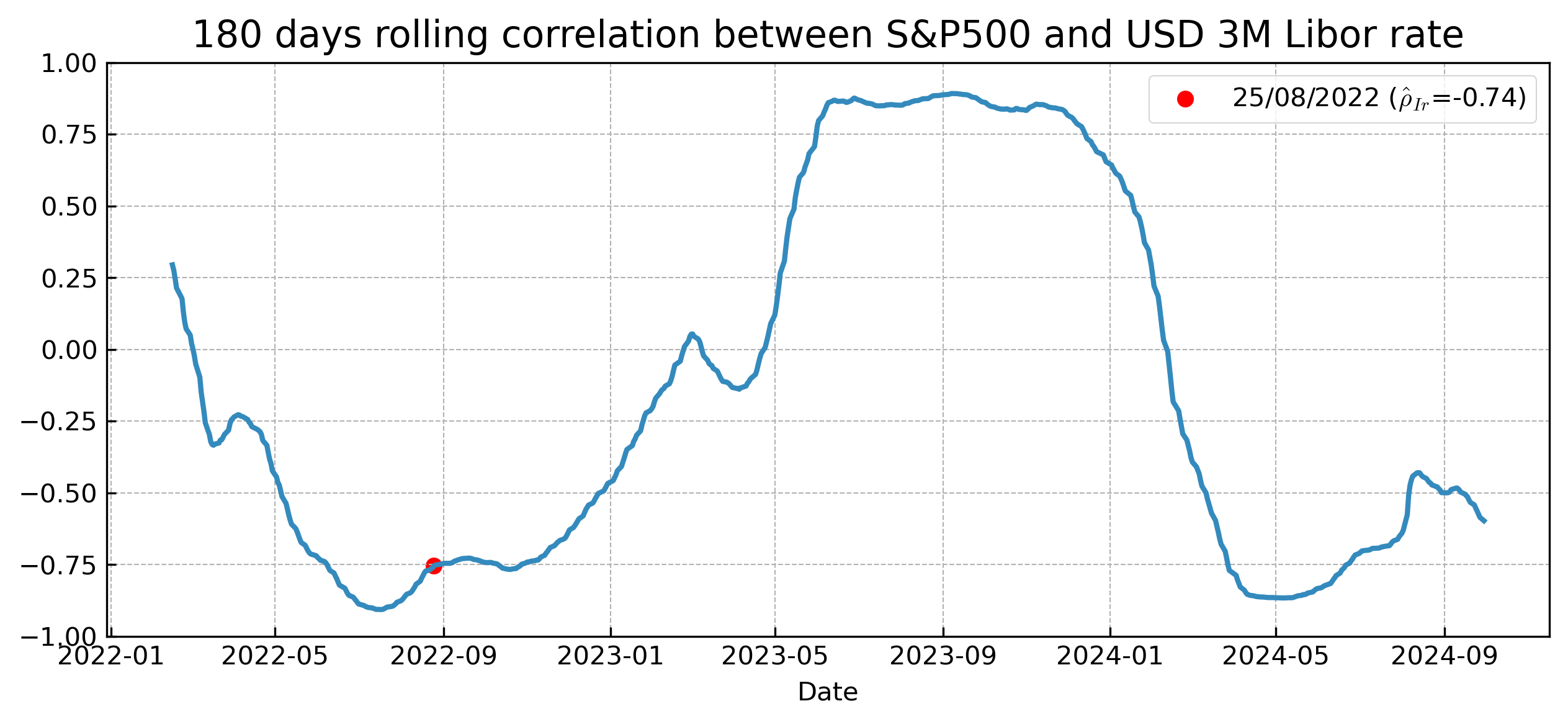}
\par\end{centering}
\centering{}\caption{\protect\label{fig:rolling_correlation_SP500_USD_rate} 180 days rolling correlation between S\&P500 and USD 3M Libor rate.}
\end{figure}

\newpage
\section{\protect\label{sec:Link-with-conventional_models}Link with conventional
quadratic linear models}

\subsection{Link with operator Riccati equations}

In this section, we provide the link between the analytic expression
of the characteristic function \eqref{eq:chf_explicit} and expression
that depends on Riccati equations. 
\begin{prop}
\label{prop:-chf_riccati_1}Let $g_{0}(.)$ be given by \eqref{eq:g_0_t},
$G_{\nu}(.)$ a Volterra kernel as in Definition \ref{def:L2_kernel-1}. 
Fix $u\in\mathbb{C}$ such that $0\leq\mathcal{\mathfrak{R}}(u)\leq1.$
Then, for all $t\leq T$,
\begin{equation}
\E^{\mathbb{Q}^{T}}\bigg[\exp\bigg(u\log\frac{I_{T}^{T}}{I_{t}^{T}}\bigg)\bigg|\mathcal{F}_{t}\bigg]=\exp\left(\phi_{t}^{u}+\chi_{t}^{u}+a^{u}\int_{t}^{T}h_{t}^{u}(s)^{2}ds+\int_{t}^{T}\int_{t}^{T}h_{t}^{u}(s)h_{t}^{u}(w)\:\bar{\psi}_{t}^{u}(s,w)dsdw\right),\label{eq:chf_explicit-1}
\end{equation}
 with $\chi_{t}^{u}$ as in Theorem \ref{thm:chf_general}, $\phi_{t}^{u}$
such that
\begin{align}
\dot{\phi}_{t}^{u} & =-a^{u}\eta_{\nu}^{2}\int_{t}^{T}G_{\nu}(s,t)^{2}ds-\eta_{\nu}^{2}\int_{t}^{T}\int_{t}^{T}G_{\nu}(s,t)G_{\nu}(w,t)\:\bar{\psi}_{t}^{u}(s,w)\,ds\,dw,\:t<T,\nonumber \\
\phi_{T}^{u} & =0,\label{eq:phi_t_ODE}
\end{align}
and $\bar{\psi}_{t}^{u}(s,w)$ that satisfies  a Riccati equation
of the form 
\begin{align*}
\dot{\bar{\psi}}_{t}^{u}(s,w)= & 2(a^{u})^{2}\eta_{\nu}^{2}\,G_{\nu}(s,t)G_{\nu}(w,t)+2a^{u}\eta_{\nu}^{2}\bigg(G_{\nu}(s,t)\,(\mathbf{G}_{\nu}^{*}\bar{\psi}_{t}^{u}(.,w))(t)+G_{\nu}(w,t)\,(\mathbf{G}_{\nu}^{*}\bar{\psi}_{t}^{u}(s,.))(t)\bigg)\\
 & +2\eta_{\nu}^{2}\,(\mathbf{G}_{\nu}^{*}\bar{\psi}_{t}^{u}(s,.))(t)\:(\mathbf{G}_{\nu}^{*}\bar{\psi}_{t}^{u}(.,w))(t),\:t<T,\:(s,w)\in(t,T]^{2}\:a.e.,\\
\bar{\psi}_{t}(t,s)= & \bar{\psi}_{t}(s,t)=b^{u}\bigg(a^{u}G_{\nu}(s,t)+(\mathbf{G}_{\nu}^{*}\bar{\psi}_{t}^{u}(s,.))(t)\bigg),\:t\leq s\leq T.
\end{align*}
\end{prop}
\begin{proof}
From Lemma \ref{lem:integral_operator}, we have that $\mathbf{\Psi}_{t}^{u}=a^{u}\text{id}+\bar{\mathbf{\Psi}}_{t}^{u}$,
where $\bar{\mathbf{\Psi}}_{t}^{u}$ is an integral operator induced
by a symmetric kernel $\bar{\psi}_{t}^{u}(s,w)$ such that 
\begin{align*}
\bar{\psi}_{t}^{u}(t,s)=\bar{\psi}_{t}^{u}(s,t) & =b^{u}\int_{t}^{T}G_{\nu}(w,t)(a^{u}\delta_{w=s}+\bar{\psi}_{t}^{u}(s,w))dw,\:t\leq s\leq T\\
 & =b^{u}\bigg(a^{u}G_{\nu}(s,t)+(\mathbf{G}_{\nu}^{*}\bar{\psi}_{t}^{u}(s,.))(t)\bigg),\:t\leq s\leq T.
\end{align*}
In this case, we have that 
\[
\langle h_{t}^{u},\mathbf{\Psi}_{t}^{u}h_{t}^{u}\rangle_{L^{2}}=a^{u}\int_{t}^{T}h_{t}^{u}(s)^{2}\:ds+\int_{t}^{T}\int_{t}^{T}h_{t}^{u}(s)h_{t}^{u}(w)\:\bar{\psi}_{t}^{u}(s,w)\:dsdw.
\]
Moreover, we also know from Lemma B.1. in \cite{key-4} that 
\begin{align*}
\dot{\mathbf{\Psi}}_{t}^{u} & =2\mathbf{\Psi}_{t}^{u}\dot{\mathbf{\Sigma}}_{t}\mathbf{\Psi}_{t}^{u},\:t<T,\\
\mathbf{\Psi}_{T}^{u} & =a^{u}(\text{id}-b^{u}\mathbf{G}_{\nu}^{*})^{-1}(\text{id}-b^{u}\mathbf{G}_{\nu})^{-1},
\end{align*}
and by the definition, we also have that 
\begin{align*}
\dot{\bar{\mathbf{\Psi}}}_{t}^{u} & =2\mathbf{\Psi}_{t}^{u}\dot{\mathbf{\Sigma}}_{t}\mathbf{\Psi}_{t}^{u},\:t<T,\\
 & =2(a^{u}\text{id}+\bar{\mathbf{\Psi}}_{t}^{u})\dot{\mathbf{\Sigma}}_{t}(a^{u}\text{id}+\bar{\mathbf{\Psi}}_{t}^{u}),\:t<T,\\
\bar{\mathbf{\Psi}}_{T}^{u} & =a^{u}(\text{id}-b^{u}\mathbf{G}_{\nu}^{*})^{-1}(\text{id}-b^{u}\mathbf{G}_{\nu})^{-1}-a^{u}\text{id},
\end{align*}
As $\dot{\bar{\mathbf{\Psi}}}_{t}^{u}$ is composed of integral operators,
it is also an integral operator induced by the following kernel 
\[
2((a^{u}\delta+\bar{\mathbf{\Psi}}_{t}^{u})\star\dot{\mathbf{\Sigma}}_{t}\star(a^{u}\delta+\bar{\mathbf{\Psi}}_{t}^{u}))(s,w),
\]
with $\delta$ the kernel induced by the identity operator such that
$(\text{id}f)(s)=\int_{t}^{T}\delta_{s=w}(ds,dw)f(w)=f(s).$ Thus, using
the dominated convergence theorem, we obtain that $t\to\bar{\psi}_{t}^{u}(s,w)$
solves a Riccati equation given, for $t<s$ and $w\leq T,$ by
\[
\dot{\bar{\psi}}_{t}^{u}(s,w)=2((a^{u}\,\delta+\bar{\mathbf{\Psi}}_{t}^{u})\star\dot{\mathbf{\Sigma}}_{t}\star(a^{u}\,\delta+\bar{\mathbf{\Psi}}_{t}^{u}))(s,w).
\]
More explicitly, since $G(z,t)=0$ for $z\leq t,$ we also have, for
$(s,w)\in(t,T]^{2}\:a.e.$, 
\begin{align*}
\dot{\bar{\psi}}_{t}^{u}(s,w)= & 2(a^{u})^{2}\eta_{\nu}^{2}\,G_{\nu}(s,t)G_{\nu}(w,t)+2a^{u}\eta_{\nu}^{2}\,G_{\nu}(s,t)\int_{t}^{T}G_{\nu}(z,t)\,\bar{\psi}_{t}^{u}(z,w)\,dz\\
 & +2a^{u}\eta_{\nu}^{2}\,G_{\nu}(w,t)\int_{t}^{T}G_{\nu}(z,t)\,\bar{\psi}_{t}^{u}(s,z)\,dz\\
 & +2\eta_{\nu}^{2}\,\int_{t}^{T}G_{\nu}(z,t)\,\bar{\psi}_{t}^{u}(s,z)\,dz\int_{t}^{T}G_{\nu}(z^{'},t)\,\bar{\psi}_{t}^{u}(z^{'},w)\,dz^{'},\:t<T.
\end{align*}
Thus, using the integral operator $\mathbf{G}_{\nu}^{*},$ we obtain
that, for $(s,w)\in(t,T]^{2}\:a.e.$, 
\begin{align*}
\dot{\bar{\psi}}_{t}^{u}(s,w)= & 2(a^{u})^{2}\eta_{\nu}^{2}\,G_{\nu}(s,t)G_{\nu}(w,t)+2a^{u}\eta_{\nu}^{2}\bigg(G_{\nu}(s,t)\,(\mathbf{G}_{\nu}^{*}\bar{\psi}_{t}^{u}(.,w))(t)+G_{\nu}(w,t)\,(\mathbf{G}_{\nu}^{*}\bar{\psi}_{t}^{u}(s,.))(t)\bigg)\\
 & +2\eta_{\nu}^{2}\,(\mathbf{G}_{\nu}^{*}\bar{\psi}_{t}^{u}(s,.))(t)\:(\mathbf{G}_{\nu}^{*}\bar{\psi}_{t}^{u}(.,w))(t),\:t<T.
\end{align*}
In addition, from Theorem \ref{thm:chf_general}, we know that $t\to\phi_{t}^{u}$
satisfies the following ODE 
\begin{align*}
\dot{\phi}_{t}^{u} & =\text{Tr}(\mathbf{\Psi}_{t}^{u}\dot{\mathbf{\Sigma}_{t}}),\:t<T\\
 & =\text{Tr}((a^{u}\text{id}+\bar{\mathbf{\Psi}}_{t}^{u})\dot{\mathbf{\Sigma}_{t}}),\:t<T\\
 & =a^{u}\text{Tr}(\dot{\mathbf{\Sigma}_{t}})+\text{Tr}(\bar{\mathbf{\Psi}}_{t}^{u}\dot{\mathbf{\Sigma}_{t}}),\:t<T\\
 & =-a^{u}\eta_{\nu}^{2}\int_{t}^{T}G_{\nu}(s,t)^{2}ds-\eta_{\nu}^{2}\int_{t}^{T}\int_{t}^{T}G_{\nu}(s,t)G_{\nu}(w,t)\:\bar{\psi}_{t}^{u}(s,w)\,ds\,dw,\:t<T,
\end{align*}
with $\phi_{T}^{u}=0.$ 
\end{proof}
We go now a step further and deduce more explicit Riccati equations by
considering completely monotone Volterra kernels. 
\begin{defn}
A kernel $G:[0,T]^{2}\to\mathbb{R}$ is a completely monotone Volterra
kernel if it satisfies Definition \ref{def:L2_kernel-1} and admits
a Laplace representation of the form 
\begin{equation}
G(t,s)=1_{s<t}\int_{\mathbb{R}^{+}}e^{-(t-s)x}\lambda(dx),\label{eq:completely_monotone_kernel}
\end{equation}
where $\lambda(.)$ is a positive measure. 
\end{defn}
\begin{prop}
\label{prop:riccati_ODE_last}Assume that the kernel function $G_{\nu}(t,s)$
is completely monotone and can be represented as \eqref{eq:completely_monotone_kernel}.
Moreover, assume that $g_{0}(.)$ is given by \eqref{eq:g_0_t} and that the kernel satisfies Definition \ref{def:L2_kernel-1}. Fix
$u\in\mathbb{C}$ such that $0\leq\mathcal{\mathfrak{R}}(u)\leq1.$
Then, for all $t\leq T$,
\begin{equation}
\E^{\mathbb{Q}^{T}}\bigg[\exp\bigg(u\log\frac{I_{T}^{T}}{I_{t}^{T}}\bigg)\bigg|\mathcal{F}_{t}\bigg]=\exp\bigg(\Theta_{t}^{u}+2\int_{\mathbb{R}_{+}}\Lambda_{t}^{u}(x)Y_{t}(x)\:\lambda(dx)+\int_{\mathbb{R}_{+}^{2}}\Gamma_{t}^{u}(x,y)Y_{t}(x)Y_{t}(y)\:\lambda(dx)\:\lambda(dy)\bigg),\label{eq:chf_explicit-1-1-1}
\end{equation}
where $t\to(\Theta_{t}^{u},\Lambda_{t}^{u},\Gamma_{t}^{u})$ solve
Riccati equations of the form 
\begin{align}
\dot{\Theta}_{t}^{u}= & \dot{\chi}_{t}^{u}-a^{u}\bigg(g_{0}^{T}(t)+\rho_{Ir}\eta_{r}B_{G_{r}}(t,T)\bigg)^{2}-\eta_{\nu}^{2}\bigg(2\bigg(\int_{\mathbb{R}_{+}}\Lambda_{t}^{u}(x)\:\lambda(dx))\bigg)^{2}+\int_{\mathbb{R}_{+}}\int_{\mathbb{R}_{+}}\Gamma_{t}^{u}(x,y)\:\lambda(dx)\:\lambda(dy)\bigg)\nonumber \\
 & -2\int_{\mathbb{R}_{+}}\bigg(b^{u}g_{0}^{T}(t)+u\eta_{\nu}\rho_{\nu r}\eta_{r}B_{G_{r}}(t,T)\bigg)\:\Lambda_{t}^{u}(x^{'})\:\lambda(dx'),\:t<T,\label{eq:theta_t}
\end{align}
\begin{align}
\dot{\Lambda}_{t}^{u}(x)= & -a^{u}\bigg(g_{0}^{T}(t)+\rho_{Ir}\eta_{r}B_{G_{r}}(t,T)\bigg)+x\Lambda_{t}^{u}(x)-2\eta_{\nu}^{2}\bigg(\int_{\mathbb{R}_{+}}\Gamma_{t}^{u}(x,x')\:\lambda(dx')\bigg)\bigg(\int_{\mathbb{R}_{+}}\Lambda_{t}^{u}(y)\:\lambda(dy)\bigg)\nonumber \\
 & -\int_{\mathbb{R}_{+}}\bigg(b^{u}g_{0}^{T}(t)+u\eta_{\nu}\rho_{\nu r}\eta_{r}B_{G_{r}}(t,T)\,\bigg)\Gamma_{t}^{u}(x,x^{'})\:\lambda(dx')\nonumber \\
 & -b^{u}\int_{\mathbb{R}_{+}}\Lambda_{t}^{u}(x^{'})\:\lambda(dx'),\:t<T,\:x\in\mathbb{R}_{+},\label{eq:lambda_t}
\end{align}
 and 
\begin{align}
\dot{\Gamma}_{t}^{u}(x,y)= & (x+y)\Gamma_{t}^{u}(x,y)-a^{u}-2\,\eta_{\nu}^{2}\bigg(\int_{\mathbb{R}_{+}}\Gamma_{t}^{u}(x,x')\:\lambda(dx')\bigg)\bigg(\int_{\mathbb{R}_{+}}\Gamma_{t}^{u}(y,y')\:\lambda(dy')\bigg)\nonumber \\
 & -b^{u}\bigg(\int_{\mathbb{R}_{+}}\Gamma_{t}^{u}(x,x^{'})\:\lambda(dx')+\int_{\mathbb{R}_{+}}\Gamma_{t}^{u}(y^{'},y)\:\lambda(dy')\bigg),\:t<T,\:(x,y)\in\mathbb{R}_{+}^{2},\label{eq:gamma_t}
\end{align}
with $\Theta_{T}^{u}=\Lambda_{T}^{u}=\Gamma_{T}^{u}=0.$
\end{prop}
\begin{proof}
The proof is given in Section \ref{sec:Proofs}.
\end{proof}
\begin{rem}
A general result about the existence and the uniqueness of solutions
to Riccati equations \eqref{eq:theta_t}-\eqref{eq:lambda_t}-\eqref{eq:gamma_t}
is provided in \cite{key-3}. \QEDA
\end{rem}

\subsection{\protect\label{subsec:Riccati-equations-for-Markov-vol}Riccati equations
for Markovian volatility models }

In this section, we consider some particular completely monotone kernels
associated to Markovian volatility models. For those models, we derive
a simplified version of the characteristic function based on Riccati
equations. In particular, for the indicator $G_{\nu}(t,s)=1_{s<t}$
and exponential $G_{\nu}(t,s)=1_{s<t}\alpha\exp(-\beta(t-s))$ kernels,
we obtain closed form expressions similar to the characteristic functions
deduced in \cite{key-25}. 
\begin{prop}
\label{prop:chf_multi_factors_ODE}Suppose that $G_{\nu}(t,s)=1_{s<t}\sum_{i=1}^{N}w_{i}e^{-x_{i}(t-s)}$
and $g_{0}(t)=\nu_{0}+\theta_{\nu}\sum_{i=1}^{N}\int_{0}^{t}w_{i}\exp(-x_{i}(t-s))ds$
with $\nu_{0}>0$ and $\theta_{\nu}\in\mathbb{R}$, then, for $u\in\mathbb{C}$
such that $0\leq\mathcal{\mathfrak{R}}(u)\leq1,$
\begin{equation}
\E^{\mathbb{Q}^{T}}\bigg[\exp\bigg(u\log\frac{I_{T}^{T}}{I_{t}^{T}}\bigg)\bigg|\mathcal{F}_{t}\bigg]=\exp\left(A_{t}^{u;\,N}+2\sum_{i=1}^{N}B_{t}^{u;\,i}\:\nu_{t}^{i}+\sum_{i=1}^{N}\sum_{j=1}^{N}C_{t}^{u;\,ij}\:\nu_{t}^{i}\nu_{t}^{j}\right),\label{eq:chf_markov_vol_n_factors}
\end{equation}
 with $\left((\nu_{t}^{i})_{i=1,...,N}\right)_{0\leq t\leq T}$ solution
of the following SDEs 
\begin{align*}
d\nu_{t}^{i} & =(-x_{i}\nu_{t}^{i}+\theta(t)+\kappa_{\nu}\:\nu_{t})\:dt+\eta_{\nu}\:dW_{\nu}^{\mathbb{Q}^{T}}(t),\:t>0,\:i=1,...,N,\\
\nu_{0}^{i} & =0,
\end{align*}
on $(\Omega,\mathcal{F},(\mathcal{F}_{t})_{t\leq T},\mathbb{Q}^{T})$
with $\theta(t):=\theta_{\nu}-\eta_{\nu}\sum_{i=1}^{n}\eta_{i}\rho_{i\nu}B_{i}(t,T)$,
such that, almost surely, 
\[
\nu_{t}=\nu_{0}+\sum_{i=1}^{N}w_{i}\nu_{t}^{i},\:t\leq T,
\]
and $A_{t}^{u;\,N}$, $(B_{t}^{u;\,i})_{i=1,...,N}$ and $(C_{t}^{u;\,ij})_{i,j=1,...,N}$
time-dependent functions satisfying Riccati equations of the form
\begin{align}
\dot{A}_{t}^{u;\,N} & =-a^{u}\bigg(\nu_{0}^{2}+\eta_{r}^{2}B_{G_{r}}(t,T)^{2}-2\nu_{0}\rho_{Ir}\eta_{r}B_{G_{r}}(t,T)\bigg)\nonumber \\
 & -2\sum_{i=1}^{N}B_{t}^{u;\,i}\bigg(\theta(t)+u\eta_{\nu}\rho_{\nu r}\eta_{r}B_{G_{r}}(t,T)+b^{u}\nu_{0}\bigg)-\eta_{\nu}^{2}\sum_{i,j=1}^{N}(2B_{t}^{u;\,i}B_{t}^{u;\,j}+C_{t}^{u;\,ij}),\:t<T.\label{eq:A_n_ODE}
\end{align}
 
\begin{align}
\dot{B}_{t}^{u;\,i}= & x_{i}B_{t}^{u;\,i}-a^{u}w_{i}\bigg(\nu_{0}+\rho_{Ir}\eta_{r}B_{G_{r}}(t,T)\bigg)-2\eta_{\nu}^{2}\sum_{j,k=1}^{N}C_{t}^{u;\,ij}B_{t}^{u;\,k}\nonumber \\
 & -w_{i}b^{u}\sum_{j=1}^{N}B_{t}^{u;\,j}-\sum_{j=1}^{N}C_{t}^{u;\,ij}\bigg(\theta(t)+u\eta_{\nu}\rho_{\nu r}\eta_{r}B_{G_{r}}(t,T)+b^{u}\nu_{0}\bigg),\:t<T,\:i=1,...,N,\label{eq:B_n_ODE}
\end{align}
 
\begin{equation}
\dot{C}_{t}^{u;\,ij}=(x_{i}+x_{j})C_{t}^{u;\,ij}-w_{i}w_{j}a^{u}-2\eta_{\nu}^{2}\sum_{k=1}^{N}\sum_{l=1}^{N}C_{t}^{u;\,ik}C_{t}^{u;\,jl}-b^{u}\sum_{k=1}^{N}\bigg(w_{j}C_{t}^{u;\,ik}+w_{i}C_{t}^{u;\,kj}\bigg)\label{eq:C_n_ODE}
\end{equation}
with $A_{T}^{u;\,N}=B_{T}^{u;\,i}=C_{T}^{u;\,ij}=0,$ for $i,j=1,...,N.$
\end{prop}
\begin{proof}
The proof is given in Section \ref{sec:Proofs}. 
\end{proof}
\begin{cor}
\label{corr:chf_stein-stein}Suppose that $G_{\nu}(t,s)=1_{s<t}\:\alpha\exp(-\beta(t-s))$
and $g_{0}(t)=\nu_{0}+\theta_{\nu}\int_{0}^{t}\alpha\exp(-\beta(t-s))ds,$
with $\nu_{0}>0$, $\alpha,\beta\in\mathbb{R}_{+}$ and $\theta_{\nu}\in\mathbb{R}$.
Then, for $u\in\mathbb{C}$ such that $0\leq\mathcal{\mathfrak{R}}(u)\leq1,$
\begin{equation}
\E^{\mathbb{Q}^{T}}\bigg[\exp\bigg(u\log\frac{I_{T}^{T}}{I_{t}^{T}}\bigg)\bigg|\mathcal{F}_{t}\bigg]=\exp(A_{t}^{u}+2B_{t}^{u}\:\nu_{t}+C_{t}^{u}\:\nu_{t}^{2}),\label{eq:chf_markov_vol}
\end{equation}
 with $A_{t}^{u}$, $B_{t}^{u}$ and $C_{t}^{u}$ time-dependent functions
given by 
\begin{align}
A_{t}^{u} & :=-\frac{1}{2}u(1-u)V_{r}(t,T)+\int_{t}^{T}\bigg[\bigg(\alpha\theta_{\nu}+(u-1)\rho_{\nu r}\alpha\eta_{\nu}\eta_{r}B_{G_{r}}(t,T)\bigg)B_{s}^{u}\nonumber \\
 & \;\;\;\;\;\;\;\;\;\;\;\;\;\;\;\;\;\;+\frac{1}{2}\eta_{\nu}^{2}\bigg((B_{s}^{u})^{2}+C_{s}^{u}\bigg)\bigg]\:ds,\label{eq:A_stein}
\end{align}
\begin{align}
B_{t}^{u} & :=-\frac{1}{2}\frac{u(1-u)}{\gamma_{1}+\gamma_{2}e^{-2\gamma(T-t)}}\bigg[\gamma_{0}(1+e^{-2\gamma(T-t)})\nonumber \\
 & \;\;\;\;\;+(\gamma_{3}-\gamma_{4}e^{-2\gamma(T-t)}-(\gamma_{5}e^{-\kappa_{r}(T-t)}-\gamma_{6}e^{-(2\gamma+\kappa_{r})(T-t)})-\gamma_{7}e^{-\gamma(T-t)}\bigg],\label{eq:B_stein}
\end{align}

\begin{equation}
C_{t}^{u}:=-\frac{1}{2}u(1-u)\frac{1-e^{-2\gamma(T-t)}}{\gamma_{1}+\gamma_{2}e^{-2\gamma(T-t)}},\label{eq:C_stein}
\end{equation}
where $V_{r}(t,T)$ is the integrated variance of the interest rate
process \eqref{eq:IR_HW_multi_factor} and 
\begin{align*}
\gamma &: =\sqrt{(-(\alpha\kappa_{\nu}-\beta)-\rho_{I\nu}\alpha\eta_{\nu}u)^{2}-\alpha^{2}\eta_{\nu}^{2}(u^{2}-u)},\:\:\:\:\:\:\gamma_{0}:=\frac{\alpha\theta_{\nu}}{\gamma},\\
\gamma_{1} & :=\gamma-(\alpha\kappa_{\nu}-\beta+\rho_{I\nu}\alpha\eta_{\nu}u),\:\:\:\:\:\:\:\:\:\:\:\:\:\:\:\:\:\:\:\:\:\:\:\:\:\:\:\:\:\:\:\:\:\:\:\:\:\:\:\:\:\:\:\gamma_{2}:=\gamma((\alpha\kappa_{\nu}-\beta+\rho_{I\nu}\alpha\eta_{\nu}u),\\
\gamma_{3} & :=\frac{\rho_{Ii}\eta_{r}\gamma_{1}+\rho_{\nu i}\eta_{r}\alpha\eta_{\nu}(u-1)}{\kappa_{r}\gamma},\:\:\:\:\:\:\:\:\:\:\:\:\:\:\:\:\:\:\:\:\:\:\:\:\:\:\:\:\:\:\:\:\:\:\:\:\:\:\:\:\gamma_{4}:=\frac{\rho_{Ir}\eta_{r}\gamma_{2}-\rho_{\nu r}\eta_{r}\alpha\eta_{\nu}(u-1)}{\kappa_{r}\gamma},\\
\gamma_{5} &: =\frac{\rho_{Ir}\eta_{r}\gamma_{1}+\rho_{\nu r}\eta_{r}\alpha\eta_{\nu}(u-1)}{\kappa_{r}(\gamma-\kappa_{r})},\:\:\:\:\:\:\:\:\:\:\:\:\:\:\:\:\:\:\:\:\:\:\:\:\:\:\:\:\:\:\:\:\:\:\:\:\:\:\:\gamma_{6}:=\frac{\rho_{Ir}\eta_{r}\gamma_{2}-\rho_{\nu r}\eta_{r}\alpha\eta_{\nu}(u-1)}{\kappa_{r}(\gamma+\kappa_{r})},\\
\gamma_{7} &: =(\gamma_{3}-\gamma_{4})-(\gamma_{5}-\gamma_{6}).
\end{align*}
\end{cor}
\begin{proof}
Based on Proposition \ref{prop:chf_multi_factors_ODE} with $N=1$,
$w_{1}=1$ and $x_{1}=\alpha$, we deduce the Riccati equations satisfied
by $A_{t}^{u},$$B_{t}^{u},$ $C_{t}^{u}$. Thus, using a laborious
development as in \cite[Proposition 2.1.]{key-25} \cite[Proposition 4.1]{key-26}, we can obtain the explicit
forms given by \eqref{eq:A_stein}-\eqref{eq:B_stein}-\eqref{eq:C_stein}.
\end{proof}

\subsection{\protect\label{subsec:Multi-factor-approximation}Multi-factor approximation
for completely monotone Volterra kernels}

In this section, we propose another alternative to approximate
the characteristic function when considering completely monotone Volterra
kernels in the sense of Definition \ref{def:L2_kernel-1}. The approximate
method is based on a Laplace representation of completely monotone
kernels and expression of the characteristic function in terms of
Riccati equations. Unlike the approach proposed in Section \ref{subsec:numerical_implementation_general_kernels},
a convergence result is deduced for this multi-factor approximate
approach. Using a Laplace transform representation, we have that completely
monotone kernels can be rewritten such that 
\begin{equation}
G_{\nu}(t,s)=1_{s<t}\int_{\mathbb{R}^{+}}e^{-(t-s)x}\lambda(dx),\label{eq:completely_monotone_kernel-1}
\end{equation}
with $\lambda(.)$ a positive measure. Hence a natural approximation
of the kernel is given by 
\begin{equation}
\hat{G}_{\nu}(t,s):=1_{s<t}\sum_{i=1}^{N}\hat{w}_{i}e^{-(t-s)\hat{x}_{i}},\label{eq:approx_kernel}
\end{equation}
where $(\hat{w}_{i})_{i=1,...,N}$ are the weights and $(\hat{x}_{i})_{i=1,...,N}$
the mean reversion terms that should be appropriately defined. Under
suitable choice, $\hat{G}_{\nu}$ converges in $L^{2}([0,T]^{2},\mathbb{R})$
to the completely monotone kernels. Therefore, based on Proposition
\ref{prop:chf_multi_factors_ODE}, we can propose another approximate
solution of the characteristic function associated to completely monotone
Volterra kernels. 
\begin{lem}
\label{lem:conve_kernel}Suppose that for all $N\geq1$, $(\hat{w}_{i})_{i=1,...,N}$
and $(\hat{x}_{i})_{i=1,...,N}$ are such that 
\[
\hat{w}_{i}:=\int_{k_{i-1}}^{k_{i}}\lambda(dk),\:\hat{x}_{i}:=\frac{1}{\hat{w}_{i}}\int_{k_{i-1}}^{k_{i}}k\;\lambda(dk),
\]
such that $k_{0}=0<k_{1}<...<k_{n}$ and 
\[
k_{n}\to\infty,\:\sum_{i=1}^{N}\int_{k_{i-1}}^{k_{i}}(k_{i}-k)^{2}\;\lambda(dk)\to0.
\]
\emph{as $N$ goes to infinity.} For $N\geq1$ fixed, there exists
a function $f_{N}^{(2)}((k_{i})_{i=0,...,N})$ such that the following
inequality holds
\begin{equation}
||\hat{G}_{\nu}-G_{\nu}||_{L^{2}([0,T]^{2},\mathbb{R})}\leq f_{N}^{(2)}((k_{i})_{i=1,...,N}),\label{eq:upper_bound_L2}
\end{equation}
and $\hat{G}_{\nu}$ converges in $L^{2}([0,T]^{2},\mathbb{R})$ to
$G_{\nu}$ when $N$ goes to infinity i.e. 
\[
||\hat{G}_{\nu}-G_{\nu}||_{L^{2}([0,T]^{2},\mathbb{R})}\to0,
\]
as $N$ goes to infinity. 
\end{lem}
\begin{proof}
We refer to the proof of \cite[Proposition 3.3.]{key-1}. 
\end{proof}
\begin{rem}
As stated in \cite{key-1}, there are different possible choices
of the auxiliary terms $(k_{i})_{i=0,...,N}$. In this paper, for fixed
$N\geq1,$ we take auxiliary terms $(k_{i})_{i=0,...,N}$ that minimize
the upper bound \eqref{eq:upper_bound_L2}. Thus, we choose $(k_{i})_{i=1,...,N}$
solutions of
\[
\inf_{(k_{i})_{i=1,..,N}\in\mathcal{E}_{N}}f_{N}^{(2)}((k_{i})_{i=1,...,N}),
\]
with $\mathcal{E}_{N}:=\{(k_{i})_{i=1,..,N}:\;k_{0}=0<k_{1}<...<k_{n}\}.$ 
\end{rem}
\begin{prop}
\label{prop:second_approx_fractional_kernel}Let $(I_{t}^{T})_{0\leq t\leq T}$
be the solution of \eqref{eq:dyn_forward_inflation} with the completely
monotone kernel \eqref{eq:completely_monotone_kernel-1} and $(\tilde{I}_{t}^{T})_{0\leq t\leq T}$
solution of \eqref{eq:dyn_forward_inflation} with the approximate
kernel \eqref{eq:approx_kernel}. Suppose that the assumptions of
Lemma \ref{lem:conve_kernel} are satisfied and that, for $u\in\mathbb{C}$
such that $0\leq\mathcal{\mathfrak{R}}(u)\leq1,$ the sequence $(\exp(u\log\tilde{I}_{T}^{T}))_{N\in\mathbb{N}}$
is uniformly integrable, then
\[
\lim_{N\to\infty}\exp(A_{0}^{u;\,N})=\E^{\mathbb{Q}^{T}}\bigg[\exp\bigg(u\log\frac{I_{T}^{T}}{I_{0}^{T}}\bigg)\bigg|\mathcal{F}_{0}\bigg],
\]
where $A_{0}^{u;\,N}$ is the solution of Riccati equation \eqref{eq:A_n_ODE}
with $w_{i}=\hat{w}_{i}$ and $x_{i}=\hat{x}_{i}$, for $i=1,...,N.$ 
\end{prop}
\begin{proof}
Using arguments similar to those of \cite{key-20}, we can show that,
for $t\in[0,T],$
\[
\tilde{I}_{t}^{T}\xrightarrow{\mathcal{L}}I_{t}^{T},
\]
as $N$ goes to infinity, where $\xrightarrow{\mathcal{L}}$ stands
for weak convergence. Therefore, using the uniform integrability of
$\exp(u\log\tilde{I}_{T}^{T})$ and the fact that $I_{0}^{T}=\tilde{I}_{0}^{T},$
we obtain that 
\[
\lim_{N\to\infty}\E^{\mathbb{Q}^{T}}\bigg[\exp\bigg(u\log\frac{\tilde{I}_{T}^{T}}{\tilde{I}_{0}^{T}}\bigg)\bigg|\mathcal{F}_{0}\bigg]=\E^{\mathbb{Q}^{T}}\bigg[\exp\bigg(u\log\frac{I_{T}^{T}}{I_{0}^{T}}\bigg)\bigg|\mathcal{F}_{0}\bigg],
\]
and thus, we conclude the proof using Corollary \ref{prop:chf_multi_factors_ODE}
and the fact that $\nu_{0}^{i}=0$ for $i=1,...,N.$
\end{proof}
Based on Proposition \ref{prop:second_approx_fractional_kernel},
we can deduce a second approximate solution of the characteristic
function for which a theoretical convergence result is established.
However, it is a semi-closed solution in the sense that it requires
solving Riccati equations. \\
\\
Finally, we decide to compare this ``multi-factor'' approximate
method with the ``operator discretization'' method proposed
in Section \ref{subsec:numerical_implementation_general_kernels}.
For this purpose, we consider the same kernels and model parameters
as in Section \ref{subsec:numerical_implementation_general_kernels}.
Figure \ref{fig:Implied-volatility-dynamics_H_03_multi_factors} presents
the implied volatility generated for different $N$ by the multi-factor
method where the system of Riccati $(N^{2}+N+1)$ equations \eqref{eq:A_n_ODE}-\eqref{eq:B_n_ODE}-\eqref{eq:C_n_ODE}
is solved numerically using an implicit method. We observe a convergence
as the number of factors increases. However,
the speed of convergence differs between the methods. As revealed by Figure
\ref{fig:Implied-volatility-ATM_H_01_}, the operator discretization
approach converges much faster than the multi-factor approach, either
with respect to $N$ or in terms of computation time. Nevertheless, it
should be noted that the speed of convergence of the ``multi-factor''
method depends strongly on the choice of multi-factor parameters $(\hat{w}_{i},\hat{x}_{i})_{i=1,...,N}$
, as well as the numerical method used to solve the system of Riccati $(N^{2}+N+1)$ equations and other choices than those used in this paper could
give a faster speed of convergence. The aim of this comparison is
above all to illustrate that the ``operator discretization'' method
converges quickly and is fairly simple to implement compared with
the ``multi-factor'' method.  Note that all numerical experiments were conducted using Python 3.13 on a standard laptop.

\begin{figure}[H]
\centering{}\includegraphics[width=0.5\textwidth]{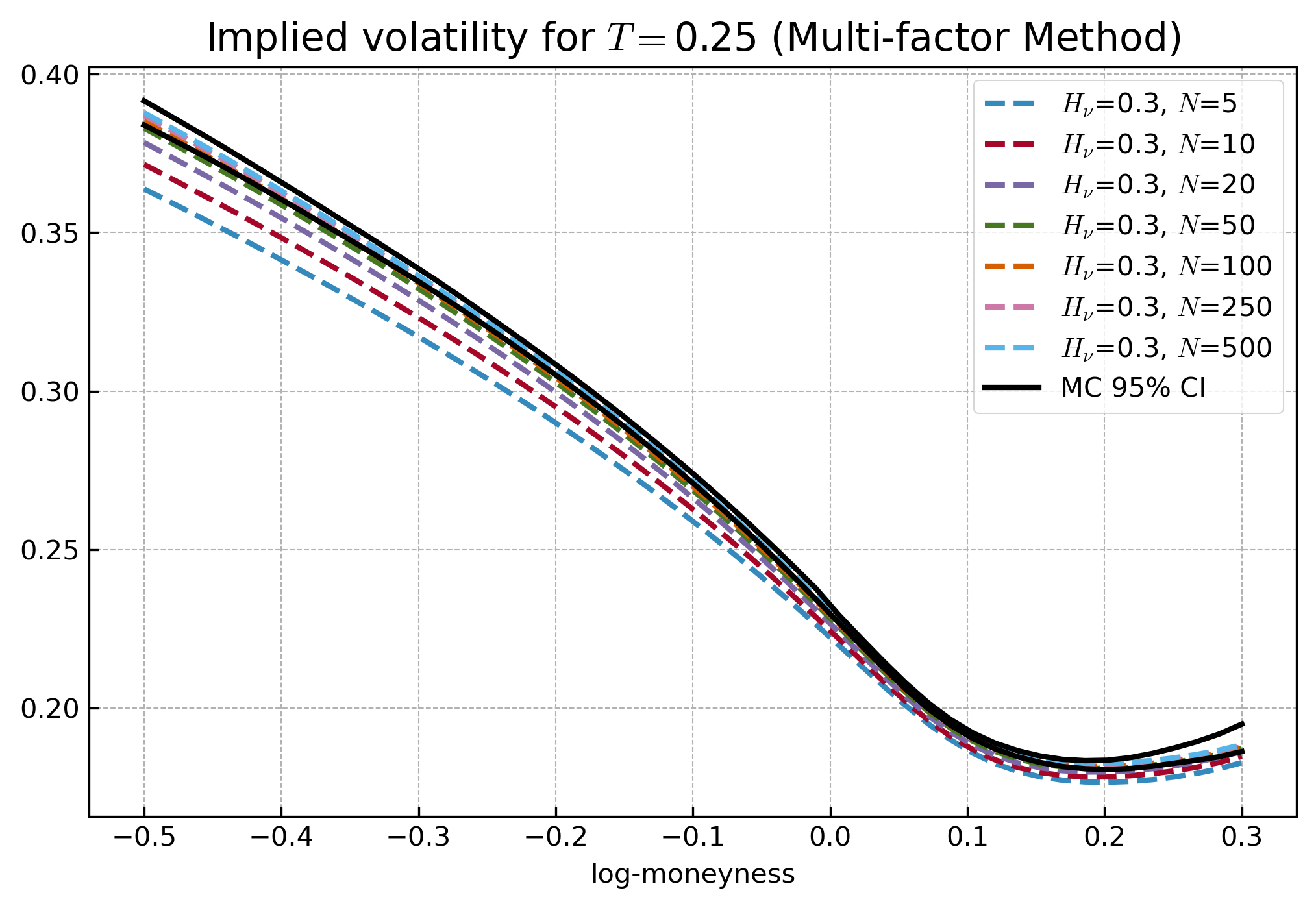}\caption{\protect\label{fig:Implied-volatility-dynamics_H_03_multi_factors}
Implied volatility dynamics generated by the multi-factor approximation
method with $H_{\nu}=0.3$ (right). Monte Carlo confidence intervals
are generated with $200,000$ simulations of risk processes using
a Euler scheme with with $500$ as discretization steps.}
\end{figure}
 
\begin{figure}[H]
\centering{}\includegraphics[width=0.5\textwidth]{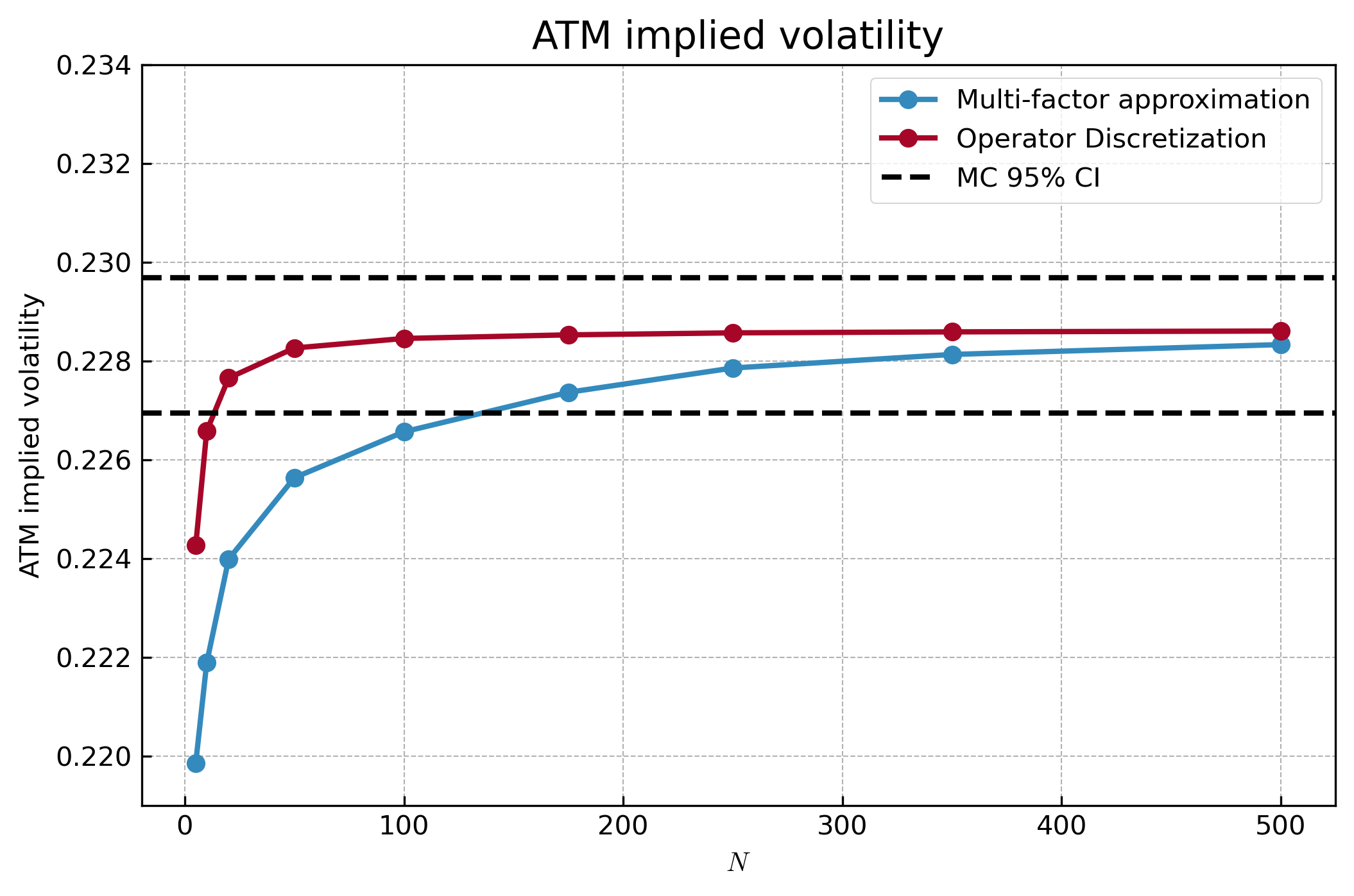}\includegraphics[width=0.49\textwidth]{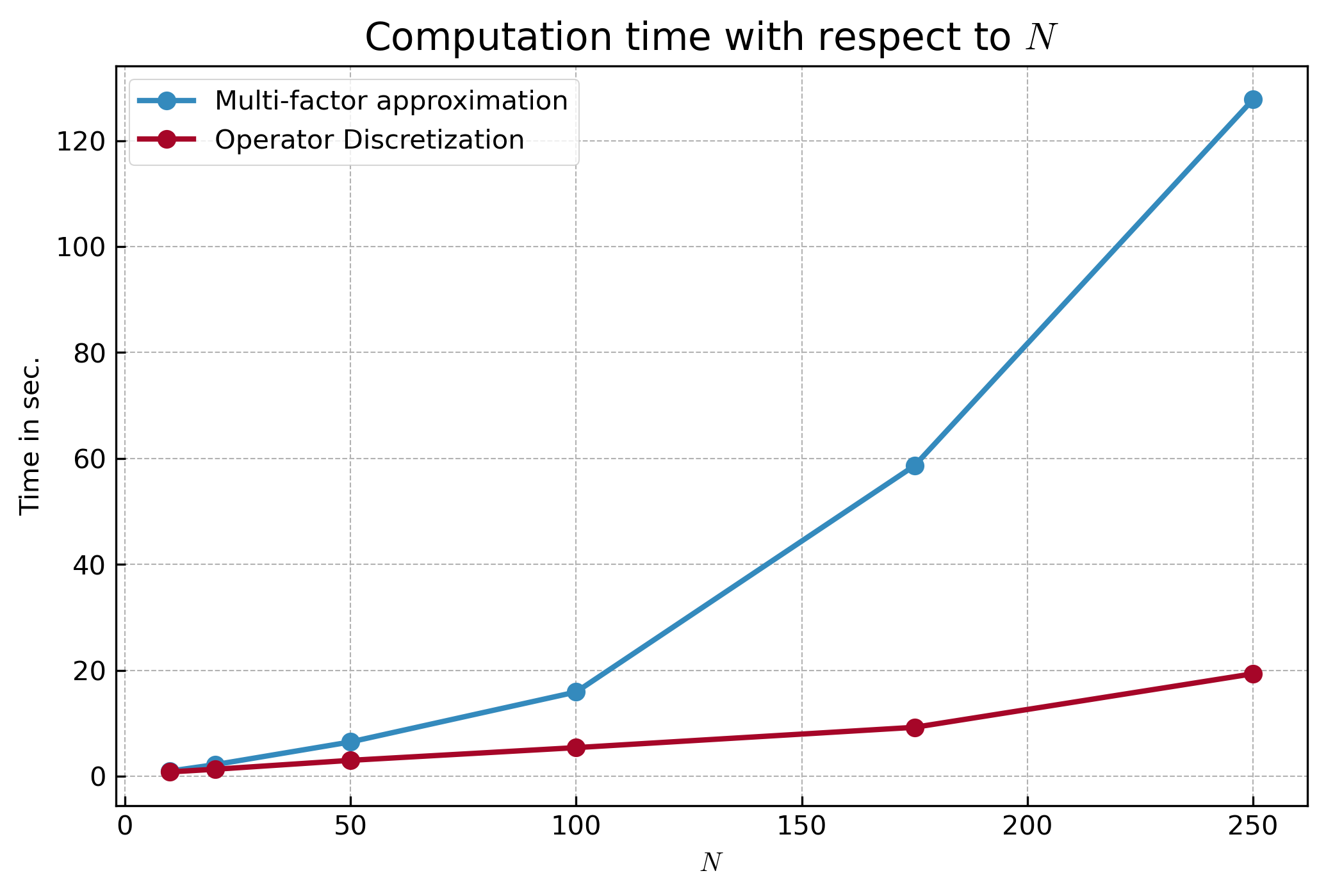}\caption{\protect\label{fig:Implied-volatility-ATM_H_01_}ATM forward implied
volatility generated by the operator discretization and multi-factor
approximation methods with $H_{\nu}=0.3$ (left) and the associated
computation time (right).}
\end{figure}

\section{Proofs\protect\label{sec:Proofs}}

\subsection{Proof of Theorem \ref{thm:chf_general}}

We first consider two useful lemmas before proving Theorem \ref{thm:chf_general}. 
\begin{lem}
\label{lem:integral_operator}For $u\in\mathbb{C}$ such that $0\leq\mathcal{\mathfrak{R}}(u)\leq1,$
there exists an integral operator $\bar{\mathbf{\Psi}}_{t}^{u}$ induced
by a symmetric kernel $\bar{\psi}_{t}^{u}(s,w)$ such that 
\begin{equation}
\mathbf{\Psi}_{t}^{u}=a^{u}\text{id}+\bar{\mathbf{\Psi}}_{t}^{u},\label{eq:operator_wrt_integral_op}
\end{equation}
and for any $f\in L^{2}([0,T],\mathbb{C}),$
\begin{equation}
(\mathbf{\Psi}_{t}^{u}f1_{t})(t)=(a^{u}\text{id}+b^{u}\mathbf{G}_{\nu}^{*}\mathbf{\Psi}_{t}^{u})(f1_{t})(t),\label{eq:operator_f_t}
\end{equation}
with $1_{t}:s\to1_{t\leq s}.$ Moreover, the kernel $\bar{\psi}_{t}(s,w)$
satisfies 
\begin{equation}
\bar{\psi}_{t}^{u}(t,s)=\bar{\psi}_{t}^{u}(s,t)=b^{u}\int_{t}^{T}G_{\nu}(w,t)(a\delta_{w=s}+\bar{\psi}_{t}^{u}(s,w))dw,\:t\leq s\leq T,\:a.e.\label{eq:cond_psi_t}
\end{equation}
\end{lem}
\begin{proof}
From \cite[Lemma B.1.]{key-4}, we know that in our setting there
exists an integral operator $\bar{\mathbf{\Psi}}_{t}^{u}$ induced
by a symmetric kernel $\bar{\psi}_{t}(s,u)$ such that \eqref{eq:operator_wrt_integral_op}
and \eqref{eq:operator_f_t} are satisfied. Moreover, from \eqref{eq:operator_f_t},
since $\mathbf{G}_{\nu}^{*}$ is an integral operator and $\bar{\psi}_{t}^{u}(s,w)$
is symmetric, we deduce that, for any $f\in L^{2}([0,T],\mathbb{C})$,
\begin{align*}
(\bar{\mathbf{\Psi}}_{t}^{u}f1_{t})(t) & =(b^{u}\mathbf{G}_{\nu}^{*}\mathbf{\Psi}_{t}^{u})(f1_{t})(t)\\
 & =b^{u}\int_{t}^{T}G_{\nu}^{*}(t,s)\:(\mathbf{\Psi}_{t}^{u}f1_{t})(s)\:ds\\
 & =b^{u}\int_{t}^{T}f(s)\int_{t}^{T}G_{\nu}(w,t)(a^{u}\delta_{w=s}+\bar{\psi}_{t}^{u}(s,w))dw\:ds
\end{align*}
Moreover, since $\bar{\mathbf{\Psi}}_{t}$ is a symmetric integral
operator, we also have that 
\begin{align*}
(\bar{\mathbf{\Psi}}_{t}^{u}f1_{t})(t) & =\int_{t}^{T}\bar{\psi}_{t}^{u}(t,s)\:f(s)\:ds.\\
 & =\int_{t}^{T}\bar{\psi}_{t}^{u}(s,t)\:f(s)\:ds
\end{align*}
We conclude that, for any $f\in L^{2}([0,T],\mathbb{C})$, 
\begin{align*}
0 & =\int_{t}^{T}f(s)\bigg[\bar{\psi}_{t}^{u}(s,t)-b^{u}\int_{t}^{T}G_{\nu}(w,t)(a^{u}\delta_{w=s}+\bar{\psi}_{t}^{u}(s,w))dw\bigg]\:ds\\
 & =\int_{t}^{T}f(s)\bigg[\bar{\psi}_{t}^{u}(t,s)-b^{u}\int_{t}^{T}G_{\nu}(w,t)(a^{u}\delta_{w=s}+\bar{\psi}_{t}^{u}(s,w))dw\bigg]\:ds,
\end{align*}
and thus, for $t\leq s\leq T,$
\[
\bar{\psi}_{t}^{u}(t,s)=\bar{\psi}_{t}^{u}(s,t)=b^{u}\int_{t}^{T}G_{\nu}(w,t)(a^{u}\delta_{w=s}+\bar{\psi}_{t}^{u}(s,w))dw,
\]
almost everywhere. 
\end{proof}

\begin{lem} 
\label{lem:h_Psi_t_dynamic}For $u\in\mathbb{C}$ such that $0\leq\mathcal{\mathfrak{R}}(u)\leq1$
and $t\in[0,T],$ we define $h_{t}^{u}(.)$ such that 
\[
h_{t}^{u}(s):=g_{t}(s)+1_{t\leq s}\bigg(\rho_{Ir}\eta_{r}B_{G_{r}}(s,T)-\int_{t}^{s}G_{\nu}(s,w)\:(b^{u}\rho_{Ir}-u\eta_{\nu}\rho_{\nu r})\eta_{r}B_{G_{r}}(w,T)\:dw\bigg),\;s,t\leq T.
\]
Then the dynamics of $t\to\langle h_{t}^{u},\mathbf{\Psi}_{t}^{u}h_{t}^{u}\rangle_{L^{2}}$
is given by 
\begin{align}
d\langle h_{t}^{u},\mathbf{\Psi}_{t}^{u}h_{t}^{u}\rangle_{L^{2}}= & \bigg(-a^{u}\bigg(\nu_{t}+\rho_{Ir}\eta_{r}B_{G_{r}}(t,T)\bigg)^{2}\nonumber \\
 & -2u\eta_{\nu}\left(\rho_{I\nu}\nu_{t}+\rho_{\nu r}\eta_{r}B_{G_{r}}(t,T)\right)(\mathbf{G}_{\nu}^{*}\mathbf{\Psi}_{t}^{u}h_{t})(t)\nonumber \\
 & -Tr(\mathbf{\Psi}_{t}^{u}\dot{\mathbf{\Sigma}}_{t})+\langle h_{t}^{u},\dot{\mathbf{\Psi}}_{t}^{u}h_{t}^{u}\rangle_{L^{2}}\bigg)\:dt\nonumber \\
 & +2\eta_{\nu}(\mathbf{G}_{\nu}^{*}\mathbf{\Psi}_{t}^{u}h_{t}^{u})(t)\:dW_{\nu}^{\mathbb{Q}^{T}}(t).\label{eq:dyn_quad_function}
\end{align}
\end{lem}
\begin{proof}
The proof is obtained directly by adapting the proof of \cite[Lemma B.2.]{key-4} with $h_t^u$ instead of $g_t$.

\end{proof}
Based now of Lemma \ref{lem:integral_operator} and \ref{lem:h_Psi_t_dynamic},
we are now ready to prove Theorem \ref{thm:chf_general}. 
\begin{proof}
As explained in \cite{key-4}, it is sufficient to make the proof
for $u\in\mathbb{R}$ such $0\leq u\leq1.$ Fix $u\in[0,1]$ and consider
the processes $(U_{t})_{0\leq t\leq T}$ and $(M_{t})_{0\leq t\leq T}$
defined, for $t\in[0,T],$ by 
\begin{equation}
U_{t}=u\log I_{t}^{T}+\phi_{t}^{u}+\chi_{t}^{u}+\langle h_{t}^{u},\mathbf{\Psi}_{t}^{u}h_{t}^{u}\rangle_{L^{2}},\label{eq:dyn_u}
\end{equation}
and 
\[
M_{t}=\exp(U_{t}).
\]
 If we prove that $(M_{t})_{0\leq t\leq T}$ is a martingale under
$\mathbb{Q}^{T}$, then the proof is complete since, in this case,
we have that 
\[
\E^{\mathbb{Q}^{T}}(M_{T}|\mathcal{F}_{t})=M_{t},
\]
and thus 
\[
\E^{\mathbb{Q}_{}^{T}}\bigg[\exp\bigg(u\log\frac{I_{T}^{T}}{I_{t}^{T}}\bigg)\bigg|\mathcal{F}_{t}\bigg]=\exp(\phi_{t}^{u}+\chi_{t}^{u}+\langle h_{t}^{u},\mathbf{\Psi}_{t}^{u}h_{t}^{u}\rangle_{L^{2}}).
\]
Let us prove that $(M_{t})_{0\leq t\leq T}$ is a (true) martingale under $\mathbb{Q}^{T}$. First, we can prove that $(M_{t})_{0\leq t\leq T}$ is a local martingale by showing that the drift of the dynamic of $(M_{t})_{0\leq t\leq T}$ is null. Using the dynamic of $(I_{t}^{T})_{0\leq t\leq T}$
given by \eqref{eq:dyn_forward_inflation}, the dynamic of $(\langle h_{t}^{u},\mathbf{\Psi}_{t}^{u}h_{t}^{u}\rangle_{L^{2}})_{0\leq t\leq T}$
given by \eqref{eq:dyn_quad_function} and the fact that $a^{u}=\frac{1}{2}(u^{2}-u)$,
we observe that 
\begin{align*}
dU_{t}= & u\:d(\log I_{t}^{T})+(\dot{\phi}_{t}^{u}+\dot{\chi}_{t}^{u})dt+d\langle h_{t}^{u},\mathbf{\Psi}_{t}^{u}h_{t}^{u}\rangle_{L^{2}}\\
= & \bigg(-\frac{u^{2}}{2}\bigg(\nu_{t}^{2}+\eta_{r}^{2}B_{G_{r}}(t,T)^{2}+2\nu_{t}\eta_{r}B_{G_{r}}(t,T)\rho_{Ir}\bigg)\\
 & +\langle h_{t}^{u},\dot{\mathbf{\Psi}}_{t}^{u}h_{t}^{u}\rangle_{L^{2}}-2u\eta_{\nu}(\rho_{I\nu}\nu_{t}+\rho_{\nu r}\eta_{r}B_{G_{r}}(t,T))(\mathbf{G}_{\nu}^{*}\mathbf{\Psi}_{t}^{u}h_{t}^{u})(t)\:\bigg)\:dt\\
 & +u\nu_{t}\:dW_{I}^{\mathbb{Q}^{T}}(t)+u\eta_{r}B_{G_{r}}(t,T)dW_{r}^{\mathbb{Q}^{T}}(t)+2\eta_{\nu}(\mathbf{G}_{\nu}^{*}\mathbf{\Psi}_{t}^{u}h_{t}^{u})(t))\:dW_{\nu}^{\mathbb{Q}^{T}}(t).
\end{align*}
Moreover, the quadratic variation of $(U_{t})_{0\leq t\leq T}$ satisfies
\begin{align*}
d\langle U\rangle_{t}= & \bigg((u^{2}\nu_{t}^{2}+4\eta_{\nu}^{2}(\mathbf{G}_{\nu}^{*}\mathbf{\Psi}_{t}^{u}h_{t}^{u})(t))^{2}+u^{2}\eta_{r}^{2}B_{G_{r}}(t,T)^{2}\\
 & +2u^{2}\nu_{t}\rho_{Ir}\eta_{r}B_{G_{r}}(t,T)\\
 & +4u\eta_{\nu}(\mathbf{G}_{\nu}^{*}\mathbf{\Psi}_{t}^{u}h_{t}^u)(t))\bigg(\rho_{I\nu}\nu_{t}+\rho_{r\nu}\eta_{r}B_{G_{r}}(t,T)\bigg)\:\bigg)dt.
\end{align*}
As the dynamic of $(M_{t})_{0\leq t\leq T}$ can be written such that
\[
dM_{t}=M_{t}\,\bigg(\frac{1}{2}d\langle U\rangle_{t}+dU_{t}\bigg),
\]
we easily obtain that 
\begin{align*}
dM_{t}= & M_{t}\,\bigg(2\eta_{\nu}^{2}((\mathbf{G}_{\nu}^{*}\mathbf{\Psi}_{t}^{u}h_{t}^{u})(t))^{2}+\langle h_{t}^{u},\dot{\mathbf{\Psi}}_{t}^{u}h_{t}^{u}\rangle_{L^{2}}\bigg)\:dt\\
 & +M_{t}\,\bigg(u\nu_{t}\:dW_{I}^{\mathbb{Q}^{T}}(t)+u\eta_{r}B_{G_{r}}(t,T)dW_{r}^{\mathbb{Q}_{}^{T}}(t)+2\eta_{\nu}(\mathbf{G}_{\nu}^{*}\mathbf{\Psi}_{t}^{u}h_{t}^{u})(t))\:dW_{\nu}^{\mathbb{Q}_{}^{T}}(t)\bigg),
\end{align*}
but as 
\[
2\eta_{\nu}^{2}((\mathbf{G}_{\nu}^{*}\mathbf{\Psi}_{t}^{u}h_{t}^{u})(t))^{2}=-2\langle h_{t}^{u},\mathbf{\Psi}_{t}^{u}\dot{\mathbf{\Sigma}}_{t}\mathbf{\Psi}_{t}^{u}h_{t}^{u}\rangle_{L^{2}},
\]
we have that 
\begin{align*}
dM_{t}= & M_{t}\,\bigg(\langle h_{t}^{u},(\dot{\mathbf{\Psi}}_{t}^{u}-2\mathbf{\Psi}_{t}^{u}\dot{\mathbf{\Sigma}}_{t}\mathbf{\Psi}_{t}^{u})h_{t}^{u}\rangle_{L^{2}}\bigg)\:dt\\
 & +M_{t}\,\bigg(u\nu_{t}\:dW_{I}^{\mathbb{Q}^{T}}(t)+u\eta_{r}B_{G_{r}}(t,T)dW_{r}^{\mathbb{Q}^{T}}(t)+2\eta_{\nu}(\mathbf{G}_{\nu}^{*}\mathbf{\Psi}_{t}^{u}h_{t}^{u})(t))\:dW_{I}^{\mathbb{Q}^{T}}(t)\bigg)\\
= & M_{t}\,\bigg(u\nu_{t}\:dW_{I}^{\mathbb{Q}^{T}}(t)+u\eta_{r}B_{G_{r}}(t,T)dW_{r}^{\mathbb{Q}^{T}}(t)+2\eta_{\nu}(\mathbf{G}_{\nu}^{*}\mathbf{\Psi}_{t}^{u}h_{t}^{u})(t))\:dW_{I}^{\mathbb{Q}^{T}}(t)\bigg),
\end{align*}
and since, from \cite[Lemma B.1.]{key-4},
\[
\dot{\mathbf{\Psi}}_{t}^{u}-2\mathbf{\Psi}_{t}^{u}\dot{\mathbf{\Sigma}}_{t}\mathbf{\Psi}_{t}^{u}=0,
\]
we have that $(M_{t})_{0\leq t\leq T}$ is a local martingale. It
remains to show that this is a true martingale. Since $\langle h_{t}^{u},\dot{\mathbf{\Psi}}_{t}^{u}h_{t}^{u}\rangle_{L^{2}}\leq0$
and $\phi_{t}^{u}=-\int_{t}^{T}Tr(\mathbf{\Psi}_{s}^{u}\dot{\mathbf{\Sigma}}_{s})\:ds\leq0$,
from \eqref{eq:dyn_u}, it follows that 

\[
U_{t}\leq u\log I_{t}^{T}+\chi_{t}^{u}.
\]
Furthermore, we observe that
\begin{align*}
U_{t}\leq & u\log I_{0}^{T}-\frac{u}{2}^{2}\int_{0}^{t}\bigg(\,\bigg(\nu_{s}+\rho_{Ir}\eta_{r}B_{G_{r}}(s,T)\bigg)^{2}+(1-\rho_{Ir}^{2})\eta_{r}^{2}B_{G_{r}}(s,T)^{2}\bigg)\:ds\\
 & +\frac{1}{2}(u^{2}-u)\bigg(\int_{0}^{T}(1-\rho_{Ir}^{2})\eta_{r}^{2}B_{G_{r}}(s,T)^{2}\:ds\bigg)\\
 & +\int_{0}^{t}u\nu_{s}\:dW_{I}^{\mathbb{Q}^{T}}(s)+\int_{0}^{t}u\eta_{r}B_{G_{r}}(s,T)dW_{r}^{\mathbb{Q}^{T}}(s),
\end{align*}
and, as $u\in[0,1],$ we deduce that 
\begin{align*}
U_{t}\leq & u\log I_{0}^{T}-\frac{u}{2}^{2}\int_{0}^{t}\bigg(\,\bigg(\nu_{s}+\rho_{Ir}\eta_{r}B_{G_{r}}(s,T)\bigg)^{2}+(1-\rho_{Ir}^{2})\eta_{r}^{2}B_{G_{r}}(s,T)^{2}\bigg)\:ds\\
 & +\int_{0}^{t}u\nu_{s}\:dW_{I}^{\mathbb{Q}^{T}}(s)+\int_{0}^{t}u\eta_{r}B_{r}(s,T)dW_{r}^{\mathbb{Q}^{T}}(s).
\end{align*}
Then, we define the process $(N_{t})_{0\leq t\leq T}$ such that 
\begin{align*}
N_{t}:= & \bigg(I_{0}^{T}\bigg)^{u}\exp\bigg(-\frac{u}{2}^{2}\int_{0}^{t}\bigg(\,\bigg(\nu_{s}+\rho_{Ir}\eta_{r}B_{G_{r}}(s,T)\bigg)^{2}+(1-\rho_{Ir}^{2})\eta_{r}^{2}B_{G_{r}}(s,T)^{2}\bigg)\:ds\\
 & +\int_{0}^{t}u\nu_{s}\:dW_{I}^{\mathbb{Q}^{T}}(s)+\int_{0}^{t}u\eta_{r}B_{G_{r}}(s,T)dW_{r}^{\mathbb{Q}^{T}}(s)\bigg),
\end{align*}
This process $(N_{t})_{0\leq t\leq T}$ is a true martingale (see
arguments in \cite[in Lemma 7.3]{key-3}). Therefore, we finally have
that 
\[
|M_{t}|\leq\exp(U_{t})\leq N_{t}
\]
and as $(M_{t})_{0\leq t\leq T}$ is a local martingale upper bounded
by a true martingale, this a true martingale and it completes the
proof. 
\end{proof}

\subsection{Proof of Proposition \ref{prop:riccati_ODE_last}}

We first consider a lemma before proving Proposition \ref{prop:riccati_ODE_last}. 
\begin{lem}
\label{lem:laplace_representation}Assume that the kernel function
$G_{\nu}(t,s)$ is completely monotone and can be represented as \eqref{eq:completely_monotone_kernel}.
Then, 
\[
\nu_{t}=g_{0}^{T}(t)+\int_{\mathbb{R}_{+}}Y_{t}(x)\:\lambda(dx),\:t\leq T,
\]
\[
g_{t}(s)=1_{t\leq s}\bigg(g_{0}^{T}(s)+\int_{\mathbb{R}_{+}}e^{-x(s-t)}Y_{t}(x)\:\lambda(dx)\bigg),
\]
with 
\begin{equation}
Y_{t}(x)=\int_{0}^{t}e^{-(t-w)x}\kappa_{\nu}\:\nu_{w}\:dw+\int_{0}^{t}e^{-(t-w)x}\eta_{\nu}\:dW_{\nu}^{\mathbb{Q}^{T}}(w),\:t\leq T,\:x\in\mathbb{R}_{+}.\label{eq:Y_t_x}
\end{equation}
Moreover, for $u\in\mathbb{C}$ such that $0\leq\mathcal{\mathfrak{R}}(u)\leq1,$
\[
h_{t}^{u}(s)=g_{t}(s)+1_{t\leq s}\bigg(\rho_{Ir}\eta_{r}B_{G_{t}}(s,T)-\int_{\mathbb{R}_{+}}b_{t}^{u}(s,x)\:\lambda(dx)\bigg),
\]
with
\begin{equation}
b_{t}^{u}(s,x):=\int_{t}^{s}e^{-(s-w)x}\bigg((b^{u}\rho_{Ir}-u\eta_{\nu}\rho_{\nu r})\eta_{r}B_{G_{r}}(w,T)\bigg)\:dw,\:t\leq s\leq T,\:x\in\mathbb{R}_{+}.\label{eq:b_t_u}
\end{equation}
\end{lem}

\begin{proof}
For $t\leq T$, applying Fubini's theorem, we can rewrite $\nu_{t}$
such that 
\begin{align*}
\nu_{t} & =g_{0}^{T}(s)+\int_{0}^{t}G_{\nu}(t,w)\kappa_{\nu}\:\nu_{u}dw+\int_{0}^{t}G_{\nu}(t,w)\eta_{\nu}\:dW_{\nu}^{\mathbb{Q}^{T}}(w)\\
 & =g_{0}^{T}(s)+\int_{0}^{t}\int_{\mathbb{R}_{+}}e^{-(t-w)x}\lambda(dx)\kappa_{\nu}\:\nu_{w}dw+\int_{0}^{t}\int_{\mathbb{R}_{+}}e^{-(t-w)x}\lambda(dx)\eta_{\nu}\:dW_{\nu}^{\mathbb{Q}^{T}}(w)\\
 & =g_{0}^{T}(s)+\int_{\mathbb{R}_{+}}\bigg(\int_{0}^{t}e^{-(t-w)x}\kappa_{\nu}\:\nu_{w}dw+\int_{0}^{t}e^{-(t-w)x}\eta_{\nu}\:dW_{\nu}^{\mathbb{Q}^{T}}(w)\bigg)\:\lambda(dx)\\
 & =g_{0}^{T}(s)+\int_{\mathbb{R}_{+}}Y_{t}(x)\:\lambda(dx),
\end{align*}
with $Y_{t}(x)$ given by \eqref{eq:Y_t_x}. Using the same arguments,
we obtain that, 
\[
g_{t}(s)=1_{t\leq s}\bigg(g_{0}^{T}(s)+\int_{\mathbb{R}_{+}}e^{-x(s-t)}Y_{t}(x)\:\lambda(dx)\bigg).
\]
 Finally, we have that, 
\begin{align*}
h_{t}^{u}(s)= & g_{t}(s)+1_{t\leq s}\bigg(\rho_{Ir}\eta_{r}B_{G_{r}}(s,T)-\int_{t}^{s}G_{\nu}(s,w)\bigg((b^{u}\rho_{Ir}-u\eta_{\nu}\rho_{\nu r})\eta_{r}B_{G_{r}}(w,T)\bigg)dw\bigg)\\
= & g_{t}(s)+1_{t\leq s}\bigg(\rho_{Ir}\eta_{r}B_{G_{r}}(s,T)-\int_{\mathbb{R}_{+}}\bigg(\int_{t}^{s}e^{-(s-w)x}\bigg((b^{u}\rho_{Ir}-u\eta_{\nu}\rho_{\nu r})\eta_{r}B_{G_{r}}(w,T)\bigg)\:dw\bigg)\:\lambda(dx)\bigg)\\
= & g_{t}(s)+1_{t\leq s}\bigg(\rho_{Ir}\eta_{r}B_{G_{r}}(s,T)-\int_{\mathbb{R}_{+}}b_{t}^{u}(s,x)\:\lambda(dx)\bigg),
\end{align*}
with $b_{t}^{u}(s,x)$ given by \eqref{eq:b_t_u}.
\end{proof}

\begin{proof}
We divide the proof into two steps. The first step is to show that
\[
\E^{\mathbb{Q}^{T}}\bigg[\exp\bigg(u\log\frac{I_{T}^{T}}{I_{t}^{T}}\bigg)\bigg|\mathcal{F}_{t}\bigg]=\exp\bigg(\Theta_{t}^{u}+2\int_{\mathbb{R}_{+}}\Lambda_{t}^{u}(x)Y_{t}(x)\:\lambda(dx)+\int_{\mathbb{R}_{+}^{2}}\Gamma_{t}^{u}(x,y)Y_{t}(x)Y_{t}(y)\:\lambda(dx)\:\lambda(dy)\bigg),
\]
with $t\to(\Theta_{t}^{u},\Lambda_{t}^{u},\Gamma_{t}^{u})$ such that
\begin{equation}
\Theta_{t}^{u}:=\phi_{t}^{u}+\chi_{t}^{u}+\int_{t}^{T}a^{u}\bar{h}_{t}^{u}(s)^{2}ds+\int_{t}^{T}\int_{t}^{T}\bar{h}_{t}^{u}(s)\bar{h}_{t}(w)\bar{\psi}_{t}^{u}(s,w)dsdw,\label{eq:capital_theta}
\end{equation}
\begin{equation}
\Lambda_{t}(x):=\int_{t}^{T}a^{u}\bar{h}_{t}^{u}(s)e^{-x(s-t)}\,ds+\int_{t}^{T}\int_{t}^{T}\bar{h}_{t}^{u}(w)e^{-x(s-t)}\bar{\psi}_{t}^{u}(s,w)\,dsdw,\:x\in\mathbb{R}_{+},\label{eq:capital_lambda}
\end{equation}
and 
\begin{equation}
\Gamma_{t}(x,y):=\int_{t}^{T}\int_{t}^{T}e^{-x(s-t)}e^{-y(w-t)}\:(a^{u}\delta_{s=w}+\bar{\psi}_{t}^{u}(s,w))\:ds\,dw,\:(x,y)\in\mathbb{R}_{+}^{2},\label{eq:capital_gamma}
\end{equation}
with $\bar{h}_{t}^{u}(s)$ a time-dependent function given by 
\begin{equation}
\bar{h}_{t}^{u}(s):=1_{t\leq s}\bigg(g_{0}^{T}(s)+\rho_{Ir}\eta_{r}B_{G_{r}}(s,T)-\int_{\mathbb{R}_{+}}b_{t}^{u}(s,x)\:\lambda(dx)\bigg).\label{eq:h_t_u_bar}
\end{equation}
Then, the second step is to deduce the equations satisfied by $(\Theta_{t}^{u},\Lambda_{t}^{u},\Gamma_{t}^{u})$. 
\subsubsection*{Step 1}

Using Proposition \ref{prop:-chf_riccati_1}, we have that 
\[
\E^{\mathbb{Q}^{T}}\bigg[\exp\bigg(u\log\frac{I_{T}^{T}}{I_{t}^{T}}\bigg)\bigg|\mathcal{F}_{t}\bigg]=\exp(\phi_{t}^{u}+\chi_{t}^{u}+a^{u}\int_{t}^{T}h_{t}^{u}(s)^{2}ds+\int_{t}^{T}\int_{t}^{T}h_{t}^{u}(s)h_{t}^{u}(w)\:\bar{\psi}_{t}^{u}(s,w)dsdw).
\]
Moreover, from Lemma \ref{lem:laplace_representation}, we have that
\[
h_{t}^{u}(s)=\bar{h}_{t}^{u}(s)+1_{t\leq s}\int_{\mathbb{R}_{+}}e^{-x(s-t)}Y_{t}(x)\:\lambda(dx),
\]
with $\bar{h}_{t}^{u}(s)$ given by \eqref{eq:h_t_u_bar}. \\
Therefore, using once again Fubini's theorem, we have that 
\begin{align*}
\int_{t}^{T}h_{t}^{u}(s)^{2}ds= & \int_{t}^{T}\bigg(\bar{h}_{t}^{u}(s)+\int_{\mathbb{R}_{+}}e^{-x(s-t)}Y_{t}(x)\:\lambda(dx)\bigg)\:\bigg(\bar{h}_{t}^{u}(s)+\int_{\mathbb{R}_{+}}e^{-y(s-t)}Y_{t}(y)\:\lambda(dy)\bigg)ds\\
= & \int_{t}^{T}\bigg(\bar{h}_{t}^{u}(s)^{2}+2\bar{h}_{t}^{u}(s)\int_{\mathbb{R}_{+}}e^{-x(s-t)}Y_{t}(x)\:\lambda(dx)+\int_{\mathbb{R}_{+}^{2}}e^{-x(s-t)}e^{-y(s-t)}Y_{t}(x)\:Y_{t}(y)\:\lambda(dx)\:\lambda(dy)\bigg)ds\\
= & \int_{t}^{T}\bar{h}_{t}^{u}(s)^{2}ds+\int_{\mathbb{R}_{+}}\bigg(\int_{t}^{T}2\bar{h}_{t}^{u}(s)e^{-x(s-t)}\,ds\bigg)\:Y_{t}(x)\:\lambda(dx)\\
 & +\int_{\mathbb{R}_{+}^{2}}\bigg(\int_{t}^{T}e^{-x(s-t)}e^{-y(s-t)}ds\bigg)Y_{t}(x)\:Y_{t}(y)\:\lambda(dx)\:\lambda(dy),
\end{align*}
and 
\begin{align*}
\int_{t}^{T}\int_{t}^{T}h_{t}^{u}(s)h_{t}^{u}(w)\:\bar{\psi}_{t}^{u}(s,w)dsdw= & \int_{t}^{T}\int_{t}^{T}\bar{h}_{t}^{u}(s)\bar{h}_{t}^{u}(w)\bar{\psi}_{t}^{u}(s,w)dsdw\\
 & +\int_{\mathbb{R}_{+}}\bigg(\int_{t}^{T}\int_{t}^{T}2\bar{h}_{t}^{u}(w)e^{-x(s-t)}\bar{\psi}_{t}^{u}(s,w)\,dsdw\bigg)\:Y_{t}(x)\:\lambda(dx)\\
 & +\int_{\mathbb{R}_{+}^{2}}\bigg(\int_{t}^{T}\int_{t}^{T}e^{-x(s-t)}e^{-y(u-t)}\bar{\psi}_{t}^{u}(s,w)ds\,dw\bigg)Y_{t}(x)\:Y_{t}(y)\:\lambda(dx)\:\lambda(dy).
\end{align*}
 Thus, we have that 
\begin{align*}
a^{u}\int_{t}^{T}h_{t}^{u}(s)^{2}ds+\int_{t}^{T}\int_{t}^{T}h_{t}^{u}(s)h_{t}^{u}(w)\:\bar{\psi}_{t}^{u}(s,w)dsdw= & \int_{t}^{T}a^{u}\bar{h}_{t}^{u}(s)^{2}ds+\int_{t}^{T}\int_{t}^{T}\bar{h}_{t}^{u}(s)\bar{h}_{t}^{u}(w)\bar{\psi}_{t}^{u}(s,w)dsdw\\
 & +2\int_{\mathbb{R}_{+}}\bigg(\int_{t}^{T}a^{u}\bar{h}_{t}^{u}(s)e^{-x(s-t)}\,ds\\
 & \;\;\;\;\;\;\;+\int_{t}^{T}\int_{t}^{T}\bar{h}_{t}^{u}(w)e^{-x(s-t)}\bar{\psi}_{t}^{u}(s,w)\,dsdw\bigg)\:Y_{t}(x)\:\lambda(dx)\\
 & +\int_{\mathbb{R}_{+}^{2}}\bigg(\int_{t}^{T}\int_{t}^{T}e^{-x(s-t)}e^{-y(w-t)}\:(a^{u}\delta_{s=w}\\
 & \;\;\;\;\;\;\;\;\;\;\;\;\;\;+\bar{\psi}_{t}^{u}(s,w))\:ds\,dw\bigg)Y_{t}(x)\:Y_{t}(y)\:\lambda(dx)\:\lambda(dy).
\end{align*}
By defining $t\to(\Theta_{t}^{u},\Lambda_{t}^{u},\Gamma_{t}^{u})$
by \eqref{eq:capital_theta}-\eqref{eq:capital_lambda}-\eqref{eq:capital_gamma},
we obtain that 
\[
\E^{\mathbb{Q}^{T}}\bigg[\exp\bigg(u\log\frac{I_{T}^{T}}{I_{t}^{T}}\bigg)\bigg|\mathcal{F}_{t}\bigg]=\exp\bigg(\Theta_{t}^{u}+2\int_{\mathbb{R}_{+}}\Lambda_{t}^{u}(x)Y_{t}(x)\:\lambda(dx)+\int_{\mathbb{R}_{+}^{2}}\Gamma_{t}^{u}(x,y)Y_{t}(x)Y_{t}(y)\:\lambda(dx)\:\lambda(dy)\bigg).
\]

\subsubsection*{Step 2}

We now deduce the equations satisfied by $(\Theta_{t}^{u},\Lambda_{t}^{u},\Gamma_{t}^{u})$.
First, we have that 
\[
\Gamma_{t}^{u}(x,y)=\int_{t}^{T}\int_{t}^{T}e^{-x(s-t)}e^{-y(w-t)}\bigg(a^{u}\delta_{s=w}+\bar{\psi}_{t}^{u}(s,w)\bigg)ds\,dw.
\]
 A direct differentiation of $t\to\Gamma_{t}^{u}$ with respect to
$t$ leads to 
\begin{align*}
\dot{\Gamma}_{t}^{u}(x,y)= & (x+y)\Gamma_{t}^{u}(x,y)-a^{u}-\int_{t}^{T}\bigg(e^{-x(s-t)}+e^{-y(s-t)}\bigg)\:\bar{\psi}_{t}^{u}(s,t)ds\\
 & +\int_{t}^{T}\int_{t}^{T}e^{-x(s-t)}e^{-y(w-t)}\dot{\bar{\psi}}_{t}^{u}(s,w)ds\,dw.
\end{align*}
Using the form of $\dot{\bar{\psi}}_{t}^{u}(s,w),$ we obtain that
\[
\int_{t}^{T}\int_{t}^{T}e^{-x(s-t)}e^{-y(w-t)}\dot{\bar{\psi}}_{t}^{u}(s,w)ds\,dw=-2\,\eta_{\nu}^{2}\bigg(\int_{\mathbb{R}_{+}}\Gamma_{t}^{u}(x,x')\:\lambda(dx')\bigg)\bigg(\int_{\mathbb{R}_{+}}\Gamma_{t}^{u}(y,y')\:\lambda(dy')\bigg).
\]
Moreover, using \eqref{eq:cond_psi_t} in Lemma \ref{lem:integral_operator}
and once again the Fubini's theorem, we have that
\begin{align*}
\int_{t}^{T}e^{-x(s-t)}\:(\bar{\psi}_{t}^{u}(s,t))ds & =b^{u}\int_{t}^{T}e^{-x(s-t)}\int_{t}^{T}G_{\nu}(w,t)(a^{u}\delta_{w=s}+\bar{\psi}_{t}^{u}(s,w))dwds\\
 & =b^{u}\bigg(\int_{t}^{T}e^{-x(s-t)}\int_{t}^{T}\bigg(\int_{\mathbb{R}_{+}}e^{-x^{'}(w-t)}\:\lambda(dx')\bigg)\:(a^{u}\delta_{w=s}+\bar{\psi}_{t}^{u}(s,w)\:dwds\\
 & =b^{u}\int_{\mathbb{R}_{+}}\bigg(\int_{t}^{T}\int_{t}^{T}e^{-x^{'}(w-t)}e^{-x(s-t)}\bigg(a^{u}\delta_{w=s}+\bar{\psi}_{t}^{u}(s,w)\bigg)\:dwds\bigg)\:\lambda(dx')\\
 & =b^{u}\int_{\mathbb{R}_{+}}\Gamma_{t}(x,x^{'})\:\lambda(dx').
\end{align*}
Using the same arguments, we also obtain that 
\[
\int_{t}^{T}e^{-y(s-t)}\:(a^{u}\delta_{s=w}+\bar{\psi}_{t}^{u}(s,t))ds=b^{u}\int_{\mathbb{R}_{+}}\Gamma_{t}^{u}(y^{'},y)\:\lambda(dy').
\]
Thus, we deduce that 
\begin{align*}
\dot{\Gamma}_{t}^{u}(x,y)= & (x+y)\Gamma_{t}^{u}(x,y)-a^{u}-2\,\eta_{\nu}^{2}\bigg(\int_{\mathbb{R}_{+}}\Gamma_{t}^{u}(x,x')\:\lambda(dx')\bigg)\bigg(\int_{\mathbb{R}_{+}}\Gamma_{t}^{u}(y,y')\:\lambda(dy')\bigg)\\
 & -b^{u}\bigg(\int_{\mathbb{R}_{+}}\Gamma_{t}^{u}(x,x^{'})\:\lambda(dx')+\int_{\mathbb{R}_{+}}\Gamma_{t}^{u}(y^{'},y)\:\lambda(dy')\bigg).
\end{align*}
We consider now the ODE satisfied by $t\to\Lambda_{t}^{u}$. We
have that 
\begin{align*}
\Lambda_{t}^{u}(x) & =\int_{t}^{T}a^{u}\bar{h}_{t}^{u}(s)e^{-x(s-t)}\,ds+\int_{t}^{T}\int_{t}^{T}\bar{h}_{t}^{u}(w)e^{-x(s-t)}\bar{\psi}_{t}^{u}(s,w)\,dsdw\\
 & =\int_{t}^{T}\int_{t}^{T}\bar{h}_{t}^{u}(w)e^{-x(s-t)}\bigg(a^{u}\delta_{s=w}+\bar{\psi}_{t}^{u}(s,w)\bigg)\,dsdw.
\end{align*}
Using the same argument than for $\Gamma_{t}^{u},$ we have that 
\begin{align*}
\dot{\Lambda}_{t}^{u}(x)= & -a^{u}\bar{h}_{t}^{u}(t)+x\Lambda_{t}^{u}(x)+\int_{t}^{T}\int_{t}^{T}\bar{h}_{t}^{u}(w)e^{-x(s-t)}\dot{\bar{\psi}}_{t}^{u}(s,w)\,dsdw\\
 & +\int_{t}^{T}\int_{t}^{T}\dot{\bar{h}}_{t}^{u}(w)e^{-x(s-t)}\bigg(a^{u}\delta_{s=w}+\bar{\psi}_{t}^{u}(s,w)\bigg)\,dsdw\\
 & -\int_{t}^{T}(\bar{h}_{t}^{u}(t)e^{-x(s-t)}+\bar{h}_{t}^{u}(s))\:\bar{\psi}_{t}^{u}(s,t)\,ds.
\end{align*}
Using the form of $\dot{\bar{\psi}}_{t}^{u}(s,w),$ we obtain that
\[
\int_{t}^{T}\int_{t}^{T}\bar{h}_{t}^{u}(w)e^{-x(s-t)}\dot{\bar{\psi}}_{t}^{u}(s,w)\,dsdw=-2\eta_{\nu}^{2}\bigg(\int_{\mathbb{R}_{+}}\Gamma_{t}^{u}(x,x')\:\lambda(dx')\bigg)\bigg(\int_{\mathbb{R}_{+}}\Lambda_{t}^{u}(y)\:\lambda(dy)\bigg),
\]
and, using similar arguments than previously, we also have that 
\begin{align*}
\int_{t}^{T}(\bar{h}_{t}^{u}(t)e^{-x(s-t)}+\bar{h}_{t}^{u}(s))\:\bar{\psi}_{t}^{u}(s,t)\,ds & =b^{u}\int_{t}^{T}\int_{t}^{T}(\bar{h}_{t}^{u}(t)e^{-x(s-t)}+\bar{h}_{t}^{u}(s))\:G(w,t)(a^{u}\delta_{w=s}+\bar{\psi}_{t}^{u}(s,w))\:dw\,ds\\
 & =b^{u}\int_{\mathbb{R}_{+}}\bigg(\int_{t}^{T}\int_{t}^{T}e^{-x^{'}(w-t)}(\bar{h}_{t}^{u}(t)e^{-x(s-t)}+\bar{h}_{t}^{u}(s))\:(a^{u}\delta_{w=s}\\
 & \;\;\;\;\;\;\;\;\;\;\;\;\;\;\;\;\;\;\;\;\;\;\;\;\;\;\;\;+\bar{\psi}_{t}^{u}(s,w))\:dwds\bigg)\:\lambda(dx')\\
 & =b^{u}\int_{\mathbb{R}_{+}}\bigg(\bar{h}_{t}^{u}(t)\Gamma_{t}^{u}(x,x^{'})+\Lambda_{t}^{u}(x^{'})\bigg)\:\lambda(dx').
\end{align*}
Moreover, as 
\[
\dot{\bar{h}}_{t}^{u}(s)=(b^{u}\rho_{Ir}-u\eta_{\nu}\rho_{\nu r})\eta_{r}B_{G_{r}}(t,T)\int_{\mathbb{R}_{+}}e^{-(s-t)x^{'}}\,\:\lambda(dx^{'}),
\]
we have that 
\begin{align*}
 & \int_{t}^{T}\int_{t}^{T}\dot{\bar{h}}_{t}^{u}(w)e^{-x(s-t)}\bigg(a^{u}\delta_{s=w}+\bar{\psi}_{t}^{u}(s,w)\bigg)\,dsdw\\
 & =(b^{u}\rho_{Ir}-u\eta_{\nu}\rho_{\nu r})\eta_{r}B_{G_{r}}(t,T)\int_{\mathbb{R}_{+}}\Gamma_{t}^{u}(x,x^{'})\,\:\lambda(dx^{'}).
\end{align*}
Therefore, as $\bar{h}_{t}^{u}(t)=g_{0}^{T}(t)+\rho_{Ir}\eta_{r}B_{G_{r}}(t,T)$,
we obtain that 
\begin{align*}
\dot{\Lambda}_{t}^{u}(x)= & -a^{u}\bar{h}_{t}^{u}(t)+x\Lambda_{t}^{u}(x)-2\eta_{\nu}^{2}\bigg(\int_{\mathbb{R}_{+}}\Gamma_{t}^{u}(x,x')\:\lambda(dx')\bigg)\bigg(\int_{\mathbb{R}_{+}}\Lambda_{t}^{u}(y)\:\lambda(dy)\bigg)\\
 & -\int_{\mathbb{R}_{+}}\bigg(b^{u}\bar{h}_{t}^{u}(t)-(b^{u}\rho_{Ir}-u\eta_{\nu}\rho_{\nu r})\eta_{r}B_{G_{r}}(t,T)\bigg)\Gamma_{t}^{u}(x,x^{'})\:\lambda(dx')\\
 & -b^{u}\int_{\mathbb{R}_{+}}\Lambda_{t}^{u}(x^{'})\:\lambda(dx')\\
= & -a^{u}\bigg(g_{0}^{T}(t)+\rho_{Ir}\eta_{r}B_{G_{r}}(t,T)\bigg)+x\Lambda_{t}^{u}(x)-2\eta_{\nu}^{2}\bigg(\int_{\mathbb{R}_{+}}\Gamma_{t}^{u}(x,x')\:\lambda(dx')\bigg)\bigg(\int_{\mathbb{R}_{+}}\Lambda_{t}^{u}(y)\:\lambda(dy)\bigg)\\
 & -\int_{\mathbb{R}_{+}}\bigg(b^{u}g_{0}^{T}(t)+u\eta_{\nu}\rho_{\nu r}\eta_{r}B_{G_{r}}(t,T)\bigg)\Gamma_{t}^{u}(x,x^{'})\:\lambda(dx')\\
 & -b^{u}\int_{\mathbb{R}_{+}}\Lambda_{t}^{u}(x^{'})\:\lambda(dx').
\end{align*}
Finally, we consider the equation satisfied by $t\to\Theta_{t}^{u}$.
We know that 
\begin{align*}
\Theta_{t}^{u} & =\phi_{t}^{u}+\chi_{t}^{u}+\int_{t}^{T}a^{u}\bar{h}_{t}^{u}(s)^{2}ds+\int_{t}^{T}\int_{t}^{T}\bar{h}_{t}^{u}(s)\bar{h}_{t}^{u}(w)\bar{\psi}_{t}^{u}(s,w)dsdw\\
 & =\phi_{t}^{u}+\chi_{t}^{u}+\int_{t}^{T}\int_{t}^{T}\bar{h}_{t}^{u}(s)\bar{h}_{t}^{u}(w)(a^{u}\delta_{s=w}+\bar{\psi}_{t}^{u}(s,w))dsdw.
\end{align*}
Therefore, we have that 
\begin{align*}
\dot{\Theta}_{t}^{u}= & \dot{\phi}_{t}^{u}+\dot{\chi}_{t}^{u}-a^{u}\bar{h}_{t}^{u}(t)^{2}+\int_{t}^{T}\int_{t}^{T}\bar{h}_{t}^{u}(s)\bar{h}_{t}^u(w)\dot{\bar{\psi}}_{t}^{u}(s,w)\,dsdw\\
 & -2\bar{h}_{t}^{u}(t)\int_{t}^{T}\bar{h}_{t}^{u}(s)\:\bar{\psi}_{t}^{u}(s,t)\,ds\\
 & +2\int_{t}^{T}\int_{t}^{T}\dot{\bar{h}}_{t}^{u}(s)\bar{h}_{t}^{u}(w)(a\delta_{s=w}+\bar{\psi}_{t}^{u}(s,w))\,dsdw.
\end{align*}
 Using once again the form of $\dot{\bar{\psi}}_{t}^{u}(s,w),$ we
deduce that 
\[
\int_{t}^{T}\int_{t}^{T}\bar{h}_{t}^{u}(s)\bar{h}_{t}^u(w)\dot{\bar{\psi}}_{t}^{u}(s,w)\,dsdw=-2\eta_{\nu}^{2}\bigg(\int_{\mathbb{R}_{+}}\Lambda_{t}^{u}(x)\:\lambda(dx))\bigg)\bigg(\int_{\mathbb{R}_{+}}\Lambda_{t}^{u}(y)\:\lambda(dy)\bigg).
\]
 Also, we have that 
\[
-2\bar{h}_{t}^{u}(t)\int_{t}^{T}\bar{h}_{t}^{u}(s)\:\bar{\psi}_{t}^{u}(s,t)\,ds=-2b^{u}\bar{h}_{t}^{u}(t)\int_{\mathbb{R}_{+}}\Lambda_{t}^{u}(x^{'})\:\lambda(dx'),
\]
and
\begin{align*}
\dot{\phi}_{t}^{u} & =-\eta_{\nu}^{2}\int_{t}^{T}\int_{t}^{T}G_\nu(s,t)G_\nu(w,t)\bigg(a^{u}\delta_{s=w}+\bar{\psi}_{t}^{u}(s,w)\bigg)ds\,dw\\
 & =-\eta_{\nu}^{2}\int_{\mathbb{R}_{+}}\int_{\mathbb{R}_{+}}\Gamma_{t}^{u}(x,y)\:\lambda(dx)\:\lambda(dy),
\end{align*}
Moreover, as 
\begin{align*}
2\int_{t}^{T}\int_{t}^{T}\dot{\bar{h}}_{t}^{u}(s)\bar{h}_{t}^{u}(w)(a^{u}\delta_{s=w}+\bar{\psi}_{t}^{u}(s,w))\,dsdw= & 2(b^{u}\rho_{Ir}-u\eta_{\nu}\rho_{\nu r})\eta_{r}B_{G_{r}}(t,T)\\
 & \int_{\mathbb{R}_{+}}\bigg(\int_{t}^{T}e^{-(s-t)x^{'}}\bar{h}_{t}(u)(a\delta_{s=u}+\bar{\psi}_{t}(s,u))\,dsdu\bigg)\:\lambda(dx^{'})\\
= & 2(b^{u}\rho_{Ir}-u\eta_{\nu}\rho_{\nu r})\eta_{r}B_{G_{r}}(t,T)\int_{\mathbb{R}_{+}}\Lambda_{t}(x^{'})\:\lambda(dx^{'}),
\end{align*}
and $\bar{h}_{t}^{u}(t)=g_{0}^{T}(t)+\rho_{Ir}\eta_{r}B_{G_{r}}(t,T)$,
we finally obtain that 
\begin{align*}
\dot{\Theta}_{t}^{u}= & \dot{\chi}_{t}^{u}-a^{u}\bigg(g_{0}^{T}(t)+\rho_{Ir}\eta_{r}B_{G_{r}}(t,T)\bigg)^{2}-\eta_{\nu}^{2}\bigg(2\bigg(\int_{\mathbb{R}_{+}}\Lambda_{t}^{u}(x)\:\lambda(dx))\bigg)^{2}+\int_{\mathbb{R}_{+}}\int_{\mathbb{R}_{+}}\Gamma_{t}^{u}(x,y)\:\lambda(dx)\:\lambda(dy)\bigg)\\
 & -2\int_{\mathbb{R}_{+}}\bigg(b^{u}g_{0}^{T}(t)+u\eta_{\nu}\rho_{\nu r}\eta_{r}B_{G_{r}}(t,T)\bigg)\:\Lambda_{t}^{u}(x^{'})\:\lambda(dx'),
\end{align*}
and that concludes the proof. 
\end{proof}

\subsection{Proof of Proposition \ref{prop:chf_multi_factors_ODE}}
\begin{proof}
Assuming that $G_{\nu}(t,s)=1_{s<t}\sum_{i=1}^{N}w_{i}e^{-x_{i}(t-s)}$
and $g_{0}^{T}(t)=\nu_{0}+\sum_{i=1}^{N}\int_{0}^{t}\theta(s)w_{i}\exp(-x_{i}(t-s))ds$,
with $\theta(t):=\theta_{\nu}-\eta_{\nu}\eta_{r}\rho_{r\nu}B_{G_{r}}(t,T),$
the volatility process $(\nu_{t})_{0\leq t\leq T}$ reduces to 
\[
\nu_{t}=\nu_{0}+\int_{0}^{t}\sum_{i=1}^{N}w_{i}\exp(-x_{i}(t-s))(\theta(s)+\kappa_{\nu}\:\nu_{s})ds+\eta_{\nu}\int_{0}^{t}\sum_{i=1}^{N}w_{i}\exp(-x_{i}(t-s))\:dW_{\nu}^{\mathbb{Q}_{}^{T}}(t).
\]
We introduce some auxiliary processes $\left((\nu_{t}^{i})_{i=1,...,N}\right)_{0\leq t\leq T}$
solution of the following SDEs 
\begin{align*}
d\nu_{t}^{i} & =(-x_{i}\nu_{t}^{i}+\theta(t)+\kappa_{\nu}\:\nu_{t})\:dt+\eta_{\nu}\:dW_{\nu}^{\mathbb{Q}^{T}}(t),\:t>0,\:i=1,...,N,\\
\nu_{0}^{i} & =0.
\end{align*}
Then, we can deduce that, almost surely, 
\[
\nu_{t}=\nu_{0}+\sum_{i=1}^{N}w_{i}\nu_{t}^{i},\:t\leq T.
\]
Also, we observe that 
\[
g_{t}(s)=1_{t\leq s}\bigg(\nu_{0}+\sum_{i=1}^{N}w_{i}e^{-x_{i}(s-t)}\nu_{t}^{i}\bigg),\;s,t\leq T,
\]
and 
\[
G_{\nu}(t,s)=1_{s<t}\sum_{i=1}^{N}w_{i}e^{-x_{i}(t-s)}=1_{s<t}\int_{0}^{+\infty}e^{-(t-s)x}\lambda(dx),
\]
with
\[
\lambda(dx)=\sum_{i=1}^{N}w_{i}\delta_{x=x_{i}}dx.
\]
We also have that, for $i=1,...,N,$
\[
Y_{t}(x_{i})=\nu_{t}^{i}-\theta_{x_{i}}(t).
\]
with $\theta_{x}(t):=\int_{0}^{t}e^{-x(t-s)}\theta(s)\:ds.$ Therefore,
using Proposition \ref{prop:riccati_ODE_last}, we have that 

\begin{align*}
\E^{\mathbb{Q}^{T}}\bigg[\exp\bigg(u\log\frac{I_{T}^{T}}{I_{t}^{T}}\bigg)\bigg|\mathcal{F}_{t}\bigg]
= & \exp\bigg(\Theta_{t}^{u}+2\sum_{i=1}^{N}w_{i}\Lambda_{t}^{u}(x_{i})Y_{t}(x_{i})+\sum_{i,j=1}^{N}w_{i}w_{j}\Gamma_{t}^{u}(x_{i},x_{j})Y_{t}(x_{i})Y_{t}(x_{j})\bigg)\\
= & \exp\bigg(\Theta_{t}^{u}-2\sum_{i=1}^{N}w_{i}\Lambda_{t}^{u}(x_{i})\theta_{x_{i}}(t)+\sum_{i=1}^{N}\sum_{j=1}^{N}w_{i}w_{j}\Gamma_{t}^{u}(x_{i},x_{j})\theta_{x_{i}}(t)\theta_{x_{j}}(t)\\
 & +2\sum_{i=1}^{N}\nu_{t}^{i}w_{i}\bigg(\Lambda_{t}^{u}(x_{i})-\sum_{j=1}^{N}w_{j}\Gamma_{t}^{u}(x_{i},x_{j})\theta_{x_{j}}(t)\bigg)+\sum_{i,j=1}^{N}w_{i}w_{j}\Gamma_{t}^{u}(x_{i},x_{j})\nu_{t}^{i}\nu_{t}^{j}\bigg)\\
:= & \exp(A_{t}^{u;\,N}+2\sum_{i=1}^{N}B_{t}^{u;\,i}\:\nu_{t}^{i}+\sum_{i=1}^{N}\sum_{j=1}^{N}C_{t}^{u;\,ij}\:\nu_{t}^{i}\nu_{t}^{j}),
\end{align*}
with 
\begin{align*}
A_{t}^{u;\,N} & :=\Theta_{t}^{u}-2\sum_{i=1}^{N}w_{i}\Lambda_{t}^{u}(x_{i})\theta_{x_{i}}(t)+\sum_{i=1}^{N}\sum_{j=1}^{N}w_{i}w_{j}\Gamma_{t}^{u}(x_{i},x_{j})\theta_{x_{i}}(t)\theta_{x_{j}}(t),\\
B_{t}^{u;\,i} & :=w_{i}\bigg(\Lambda_{t}^{u}(x_{i})-\sum_{j=1}^{N}w_{j}\Gamma_{t}^{u}(x_{i},x_{j})\theta_{x_{j}}(t)\bigg),\:i=1,...,N,\\
C_{t}^{u;\,ij} & :=w_{i}w_{j}\Gamma_{t}^{u}(x_{i},x_{j}),\:i,j=1,...,N.
\end{align*}
Using Proposition \ref{prop:riccati_ODE_last}, we have that 
\begin{align*}
\dot{C}_{t}^{u;\,ij}= & w_{i}w_{j}\dot{\Gamma}_{t}^{u}(x_{i},x_{j})\\
= & w_{i}w_{j}\bigg((x_{i}+x_{j})\Gamma_{t}^{u}(x_{i},x_{j})-a^{u}-2\eta_{\nu}^{2}\sum_{k=1}^{N}\sum_{l=1}^{N}w_{k}w_{l}\Gamma_{t}^{u}(x_{i},x_{k})\Gamma_{t}^{u}(x_{j},x_{l})\bigg)\\
 & -b^{u}\bigg(\sum_{k=1}^{N}w_{k}\Gamma_{t}^{u}(x_{i},x_{k})+\sum_{k=1}^{N}w_{k}\Gamma_{t}^{u}(x_{k},x_{j})\bigg)\\
= & (x_{i}+x_{j})C_{t}^{u;\,ij}-w_{i}w_{j}a^{u}-2\eta_{\nu}^{2}\sum_{k=1}^{N}\sum_{l=1}^{N}C_{t}^{u;\,ik}C_{t}^{u;\,jl}-b^{u}\sum_{k=1}^{N}\bigg(w_{j}C_{t}^{u;\,ik}+w_{i}C_{t}^{u;\,kj}\bigg).
\end{align*}
Moreover, after fastidious calculus, we end up with
\begin{align*}
\dot{B}_{t}^{u;\,i}= & x_{i}B_{t}^{u;\,i}-a^{u}w_{i}\bigg(\nu_{0}+\rho_{Ir}\eta_{r}B_{G_{r}}(t,T)\bigg)-2\eta_{\nu}^{2}\sum_{j,k=1}^{N}C_{t}^{u;\,ij}B_{t}^{u;\,k}\\
 & -b^{u}\bigg(\sum_{j=1}^{N}\bigg(w_{i}B_{t}^{u;\,j}+\nu_{0}C_{t}^{u;\,ij}\bigg)-\sum_{j=1}^{N}C_{t}^{u;\,ij}\bigg(\theta(t)+u\eta_{\nu}\rho_{\nu r}\eta_{r}B_{G_{r}}(t,T)\bigg),\:t<T,\:i=1,...,N,
\end{align*}
and 
\begin{align*}
\dot{A}_{t}^{u;\,N} & =-a^{u}\bigg(\nu_{0}^{2}+\eta_{r}^{2}B_{G_{r}}(t,T)^{2}-2\nu_{0}\rho_{Ir}\eta_{r}B_{G_{r}}(t,T)\bigg)\\
 & -2\sum_{i=1}^{N}B_{t}^{u;\,i}\bigg(\theta(t)+u\eta_{\nu}\rho_{\nu r}\eta_{r}B_{G_{r}}(t,T)+b^{u}\nu_{0}\bigg)-\eta_{\nu}^{2}\sum_{i,j=1}^{N}(2B_{t}^{u;\,i}B_{t}^{u;\,j}+C_{t}^{u;\,ij}).
\end{align*}
\end{proof}

\section{Appendix}

\subsection{Full calibration}\label{sec:appendix_full_calib}

\begin{table}[ht]
\centering
\begin{tabular}{lccc}
\hline
Setting & $\kappa_\nu=\rho_{\nu r}=0$ & $\rho_{\nu r}=0$& Full calibration \\
\hline
$\nu_0$          & $0.1964$ & $0.1983$ & $0.1982$ \\
$\theta_\nu$   & $-0.0248$ & $0.0982$ & $0.1003$ \\
$\kappa_\nu$   & 0         & $-0.7361$ & $-0.7331$ \\
$\eta_\nu$     & $0.2123$ & $0.3158$ & $0.3206$ \\
$H_\nu$        & $0.2992$ & $0.4316$ & $ 0.4375$ \\
$\rho_{Iv}$    & $-0.7981$ & $-0.7857$ & $-0.7897$ \\
$\rho_{In}$    & $-0.5971$ & $-0.5998$ & $-0.7515$ \\
$\rho_{\nu r}$ & 0        & 0 & $0.2625$ \\
RMSE           & $0.5321\%$ & $0.5067\%$   &  $0.4977\%$   \\
\hline
\end{tabular}
\caption{Calibrated parameters under different specifications with $G_{\nu}(t,s)=1_{s<t}\frac{(t-s)^{H_{\nu}-1/2}}{\Gamma(H_\nu+1/2)}$ and market data of
25/08/2022.}
\label{tab:calibration_comparison_fract}
\end{table}

\begin{table}[ht]
\centering
\begin{tabular}{lccc}
\hline
Setting & $\kappa_\nu=\rho_{\nu r}=0$ & $\rho_{\nu r}=0$& Full calibration \\
\hline
$\nu_0$          & $0.1978$  & $0.1960$ & $0.1959$ \\
$\theta_\nu$   & $-0.0259$ & $0.0621$ & $0.0611$ \\
$\kappa_\nu$   & 0         & $-0.5167$   & $-0.4761$ \\
$\eta_\nu$     & $0.2164$  & $0.2855$ & $0.2867$ \\
$H_\nu$        & $0.2273$  & $0.3662$ & $0.3737$ \\
$\rho_{Iv}$    & $-0.7868$ & $-0.7860$ & $-0.7906$ \\
$\rho_{In}$    & $-0.6107$ & $-0.6174$ & $-0.8252$ \\
$\rho_{\nu r}$ & 0         & 0 & $0.3019$ \\
RMSE           & $0.4602\%$ & $0.4496\%$  & $0.4437\%$    \\
\hline
\end{tabular}
\caption{Calibrated parameters under different specifications with $G_{\nu}(t,s)=1_{s<t}\frac{(t-s+\varepsilon)^{H_{\nu}-1/2}}{\Gamma(H_\nu+1/2)}$ with $\varepsilon=1/52$ and market data of
25/08/2022.}
\label{tab:calibration_comparison_shited_fract}
\end{table}

In this appendix, we analyze the effect of relaxing the restrictions $\kappa_\nu=0$ and $\rho_{\nu r}=0$ on the calibration results for both the fractional and shifted fractional kernels. By analyzing Tables \ref{tab:calibration_comparison_fract} and
\ref{tab:calibration_comparison_shited_fract}, we observe several
interesting features. First, relaxing the restriction $\kappa_\nu=0$ leads to changes in the calibrated values of
$\theta_\nu$, $\eta_\nu$, and $H_\nu$ for both the fractional and
shifted fractional kernels. This behavior is consistent with the role played by $\kappa_\nu$ in the volatility dynamics. When
$\kappa_\nu=0$, the process is given by
\begin{equation}\label{eq: dyn_nu_kappa_null}
    \nu_{t}  =\nu_0+\theta_\nu \int_{0}^{t}G_{\nu}(t,s) ds+\eta_{\nu}\int_{0}^{t}G_{\nu}(t,s)\:dW_{\nu}^{\mathbb{Q}}(s), \quad t\leq T. 
\end{equation}
so that the temporal dependence of the process is directly governed
by the kernel $G_\nu$. When $\kappa_\nu\neq0$, the volatility process
satisfies
\begin{equation*}
    \nu_{t}  =\nu_0+ \int_{0}^{t}\left(\theta_\nu+\kappa_n \nu_s\right)G_{\nu}(t,s) ds+\eta_{\nu}\int_{0}^{t}G_{\nu}(t,s)\:dW_{\nu}^{\mathbb{Q}}(s), \quad t\leq T,
\end{equation*}
and can be represented using the corresponding resolvent as
\begin{equation}\label{eq: dyn_nu_kappa_non_zero}
\nu_{t}=g_{0}(t)+\int_{0}^{t}R_{\kappa_{\nu}G_{\nu}}(t,s)g_{0}(s)ds+\eta_{\nu}\int_{0}^{t} \frac{R_{\kappa_{\nu}G_{\nu}}(t,s)}{\kappa_\nu}\;dW_{\nu}^{\mathbb{Q}}(s),\quad t\leq T,
\end{equation}
with $g_{0}(t)=\nu_0+\theta_\nu \int_{0}^{t}G_{\nu}(t,s) ds$. Comparing \eqref{eq: dyn_nu_kappa_null} and
\eqref{eq: dyn_nu_kappa_non_zero}, we observe that the volatility
process retains the same general structure in both cases, namely a
deterministic component and a stochastic convolution with respect to
the Brownian motion, both of which are governed by a memory kernel.
However, when $\kappa_\nu\neq0$, the feedback term involving
$\nu$ modifies the effective kernel through the associated
resolvent. Therefore, we observe that $\kappa_\nu$ affects the volatility dynamics through two channels. First, when $\kappa_\nu<0$, it introduces a mean-reverting feedback, with $\theta_\nu/(-\kappa_\nu)$ corresponding to the associated long-run level. Second, $\kappa_\nu$ modifies the effective memory structure through the resolvent kernel. This additional flexibility naturally affects the calibrated values of $\theta_\nu$, $\eta_\nu$ and $H_\nu$. When $\rho_{\nu r}$ is also left unrestricted, the changes in the calibrated parameters are comparatively limited. The main adjustment concerns the correlation structure, while the improvement in calibration error remains marginal. The initial volatility level $\nu_0$ and the spot-volatility correlation $\rho_{I\nu}$ remain stable across the different specifications. Overall, the main qualitative conclusions are robust: all calibrated models remain in the regime $H_\nu<0.5$, and the shifted fractional kernel consistently achieves a lower RMSE than the fractional kernel.

\subsection{Additional figures}

\begin{figure}[H]
\centering{}\includegraphics[width=0.65\textwidth]{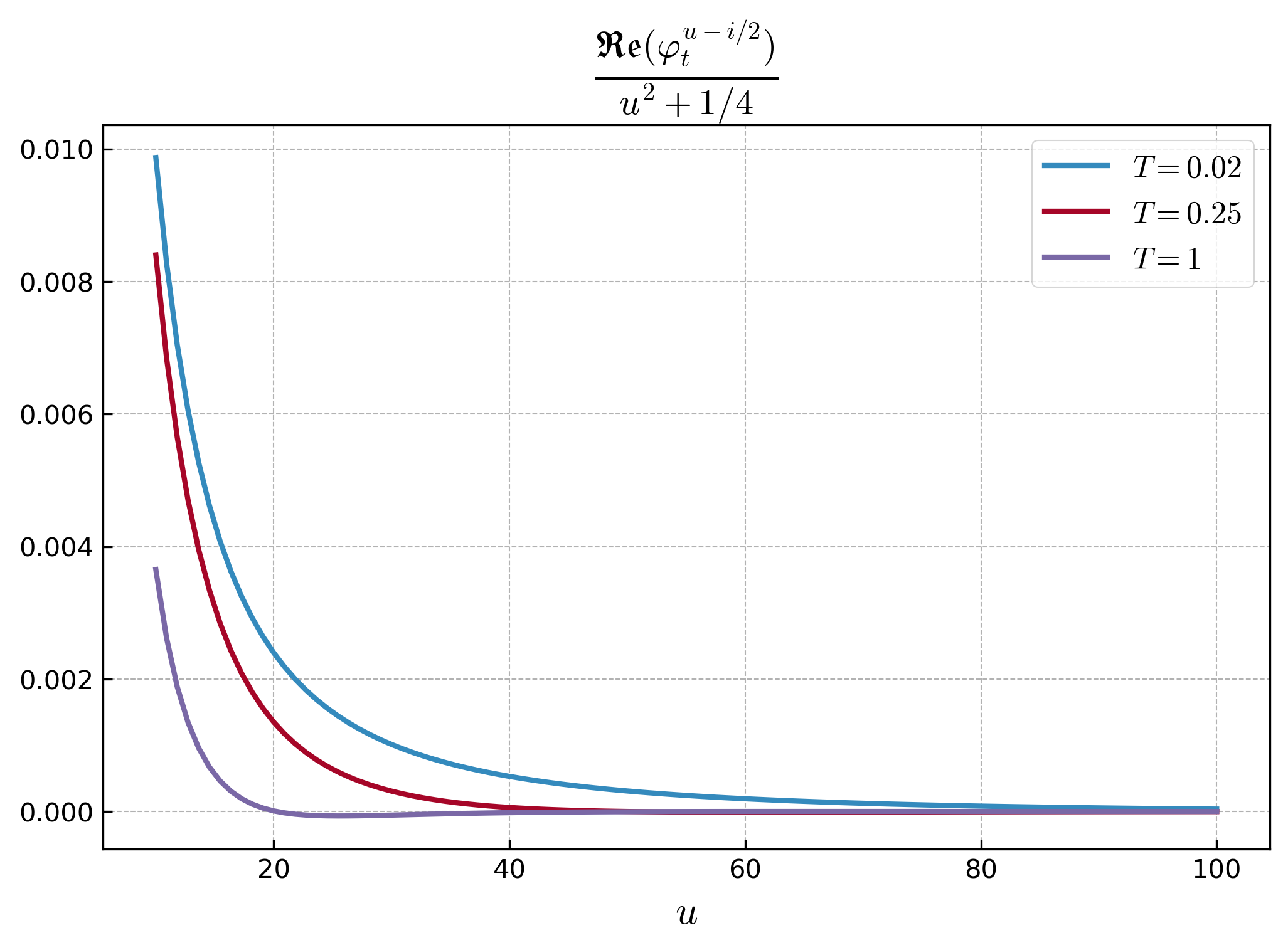}\caption{\protect\label{fig:integrand_decreasing}
Integrand of Lewis pricing formula \eqref{eq:lewis_price_call} for $k_t=0$. Hull-White and Stein-Stein model with parameters: $\kappa_r=-0.03$, $\eta_r=0.01$, $I_0^T=100$, ${\nu}_{0}=0.1$,
${\theta}_{\nu}=0.1,$ ${\eta}_{\nu}= 0.125$, $\kappa_\nu = 0$, ${\rho}_{I\nu}=-0.7,$ ${\rho}_{r\nu}=-0.25,$
${\rho}_{Ir}=-0.25$.}
\end{figure}

\begin{figure}[H]
\centering{}\includegraphics[width=0.4\textwidth]{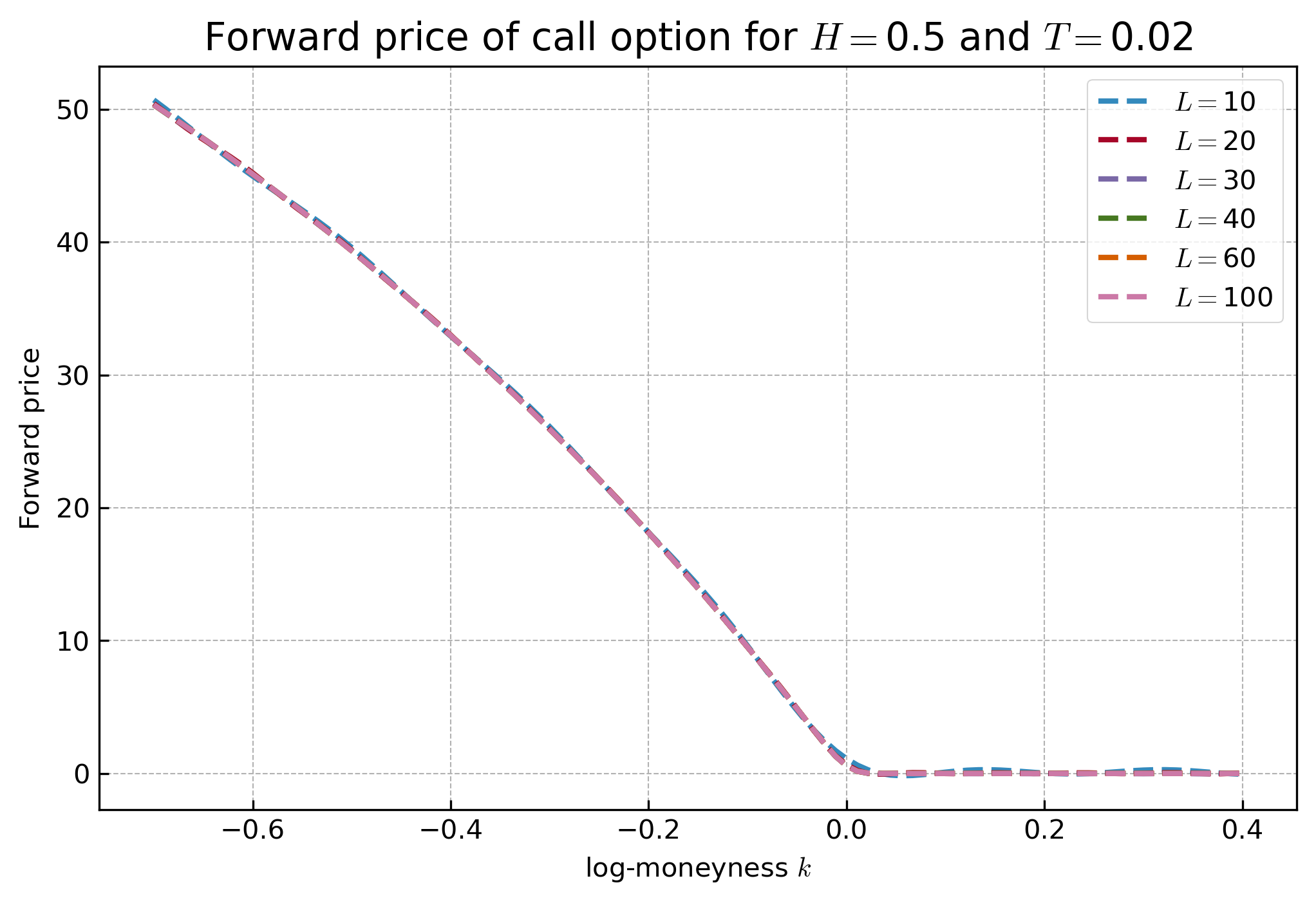} \\
\includegraphics[width=0.4\textwidth]{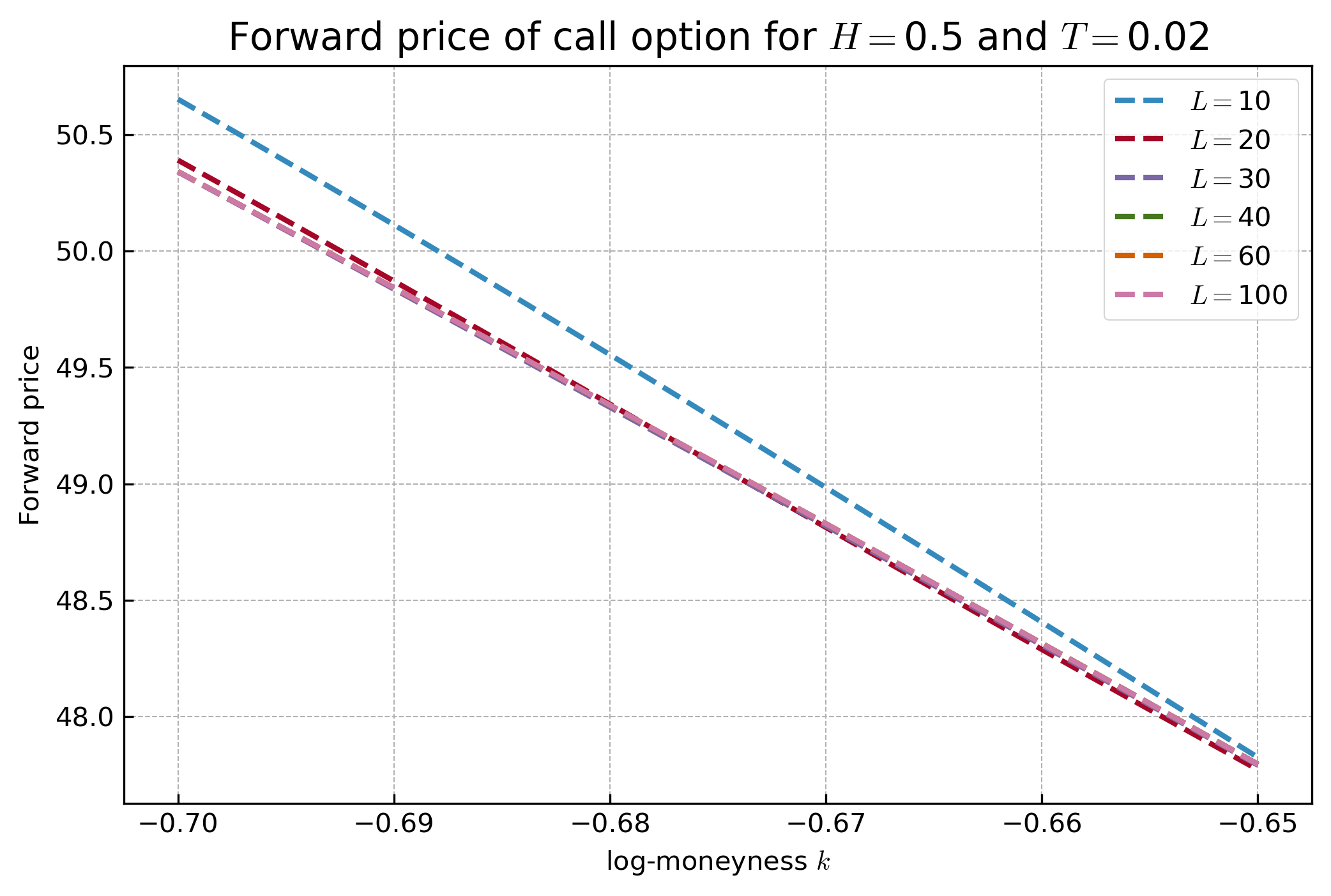} \includegraphics[width=0.4\textwidth]{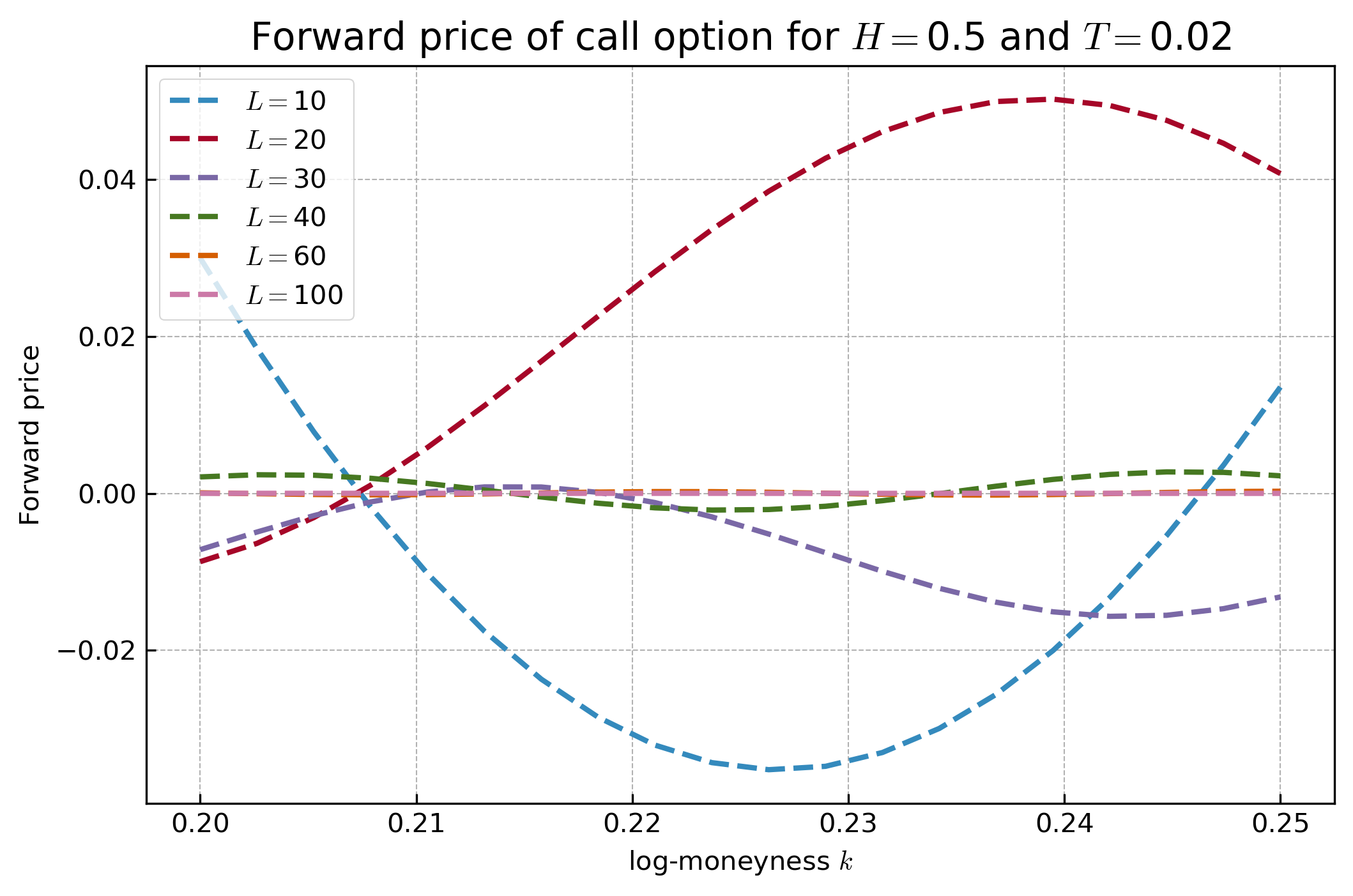}\caption{\protect\label{fig:GL_pricing_T_002}
Forward price of call option for $T=0.02$ using \cite{key-18} pricing formula with Gauss-Laguerre quadrature numerical integration by varying $L$ (number of weights). Hull-White and Stein-Stein model with parameters: $\kappa_r=-0.03$, $\eta_r=0.01$, $I_0^T=100$, ${\nu}_{0}=0.1$,
${\theta}_{\nu}=0.1,$ ${\eta}_{\nu}= 0.125$, $\kappa_\nu = 0$, ${\rho}_{I\nu}=-0.7,$ ${\rho}_{r\nu}=-0.25,$
${\rho}_{Ir}=-0.25$.}
\end{figure}

\begin{figure}[H]
\centering{}\includegraphics[width=0.4\textwidth]{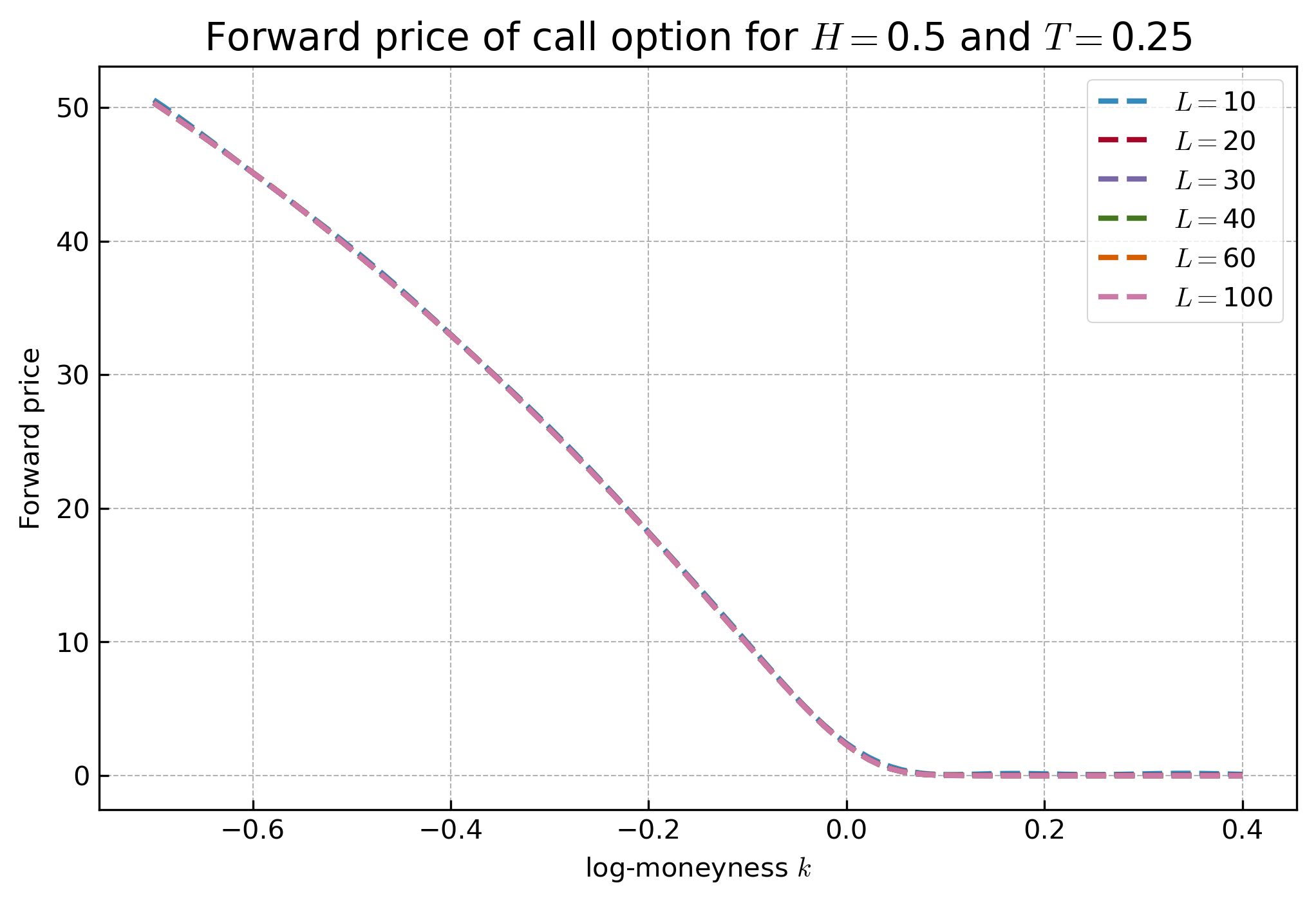} \\
\includegraphics[width=0.4\textwidth]{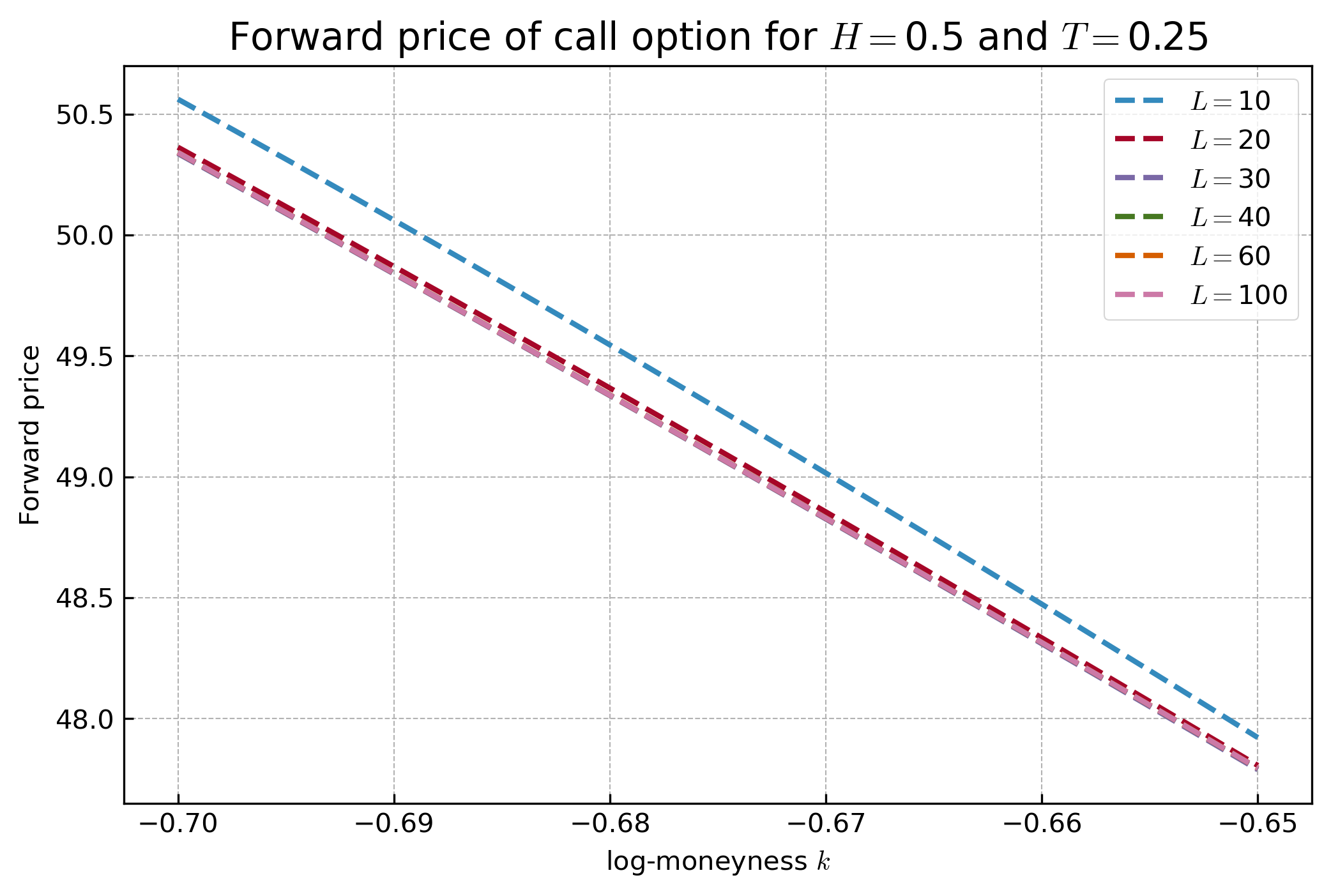} \includegraphics[width=0.4\textwidth]{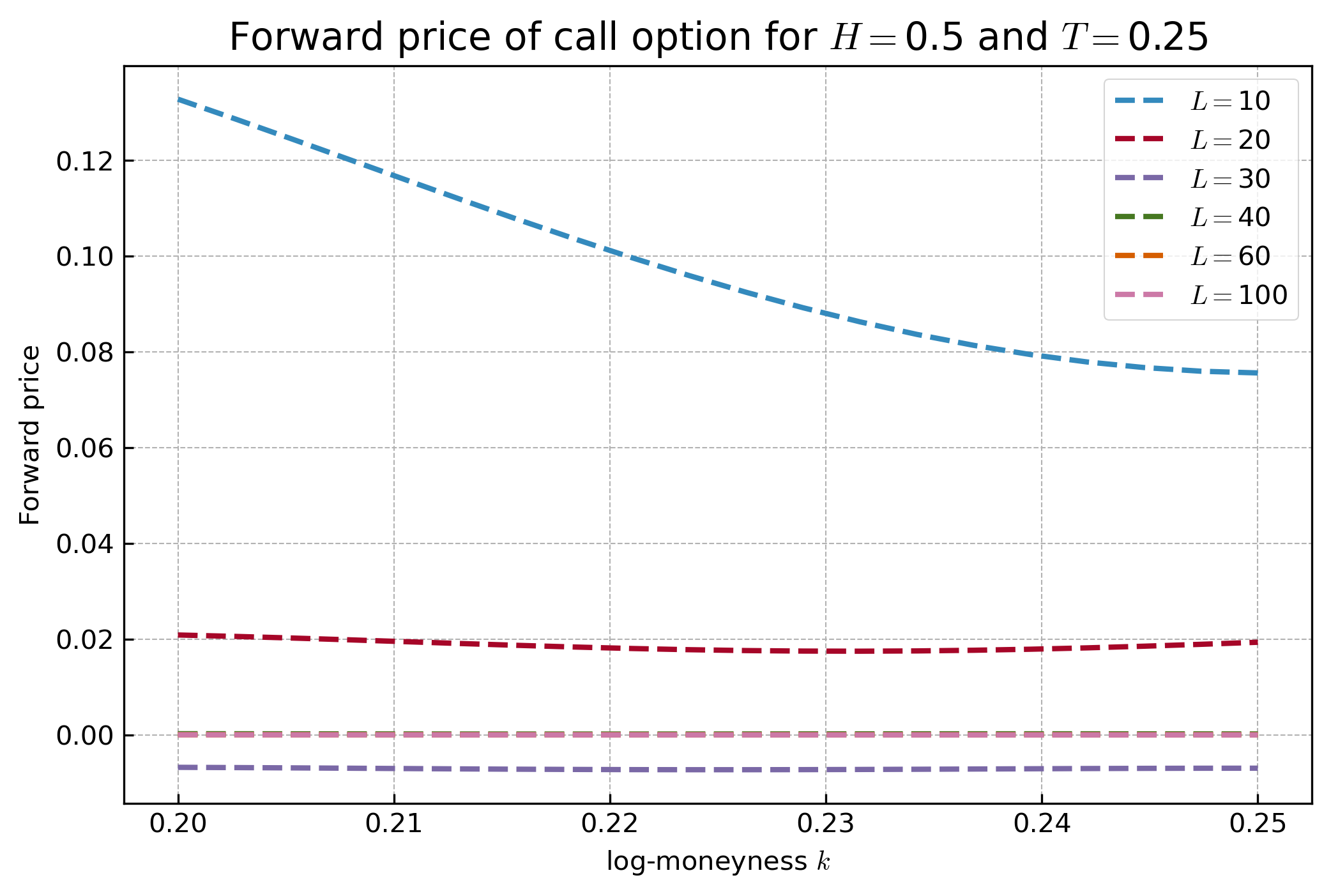}\caption{\protect\label{fig:GL_pricing_T_025}
Forward price of call option for $T=0.25$ using \cite{key-18} pricing formula with Gauss-Laguerre quadrature numerical integration by varying $L$ (number of weights). Hull-White and Stein-Stein model with parameters: $\kappa_r=-0.03$, $\eta_r=0.01$, $I_0^T=100$, ${\nu}_{0}=0.1$,
${\theta}_{\nu}=0.1,$ ${\eta}_{\nu}= 0.125$, $\kappa_\nu = 0$, ${\rho}_{I\nu}=-0.7,$ ${\rho}_{r\nu}=-0.25,$
${\rho}_{Ir}=-0.25$.}
\end{figure}

\begin{figure}[H]
\centering{}\includegraphics[width=0.4\textwidth]{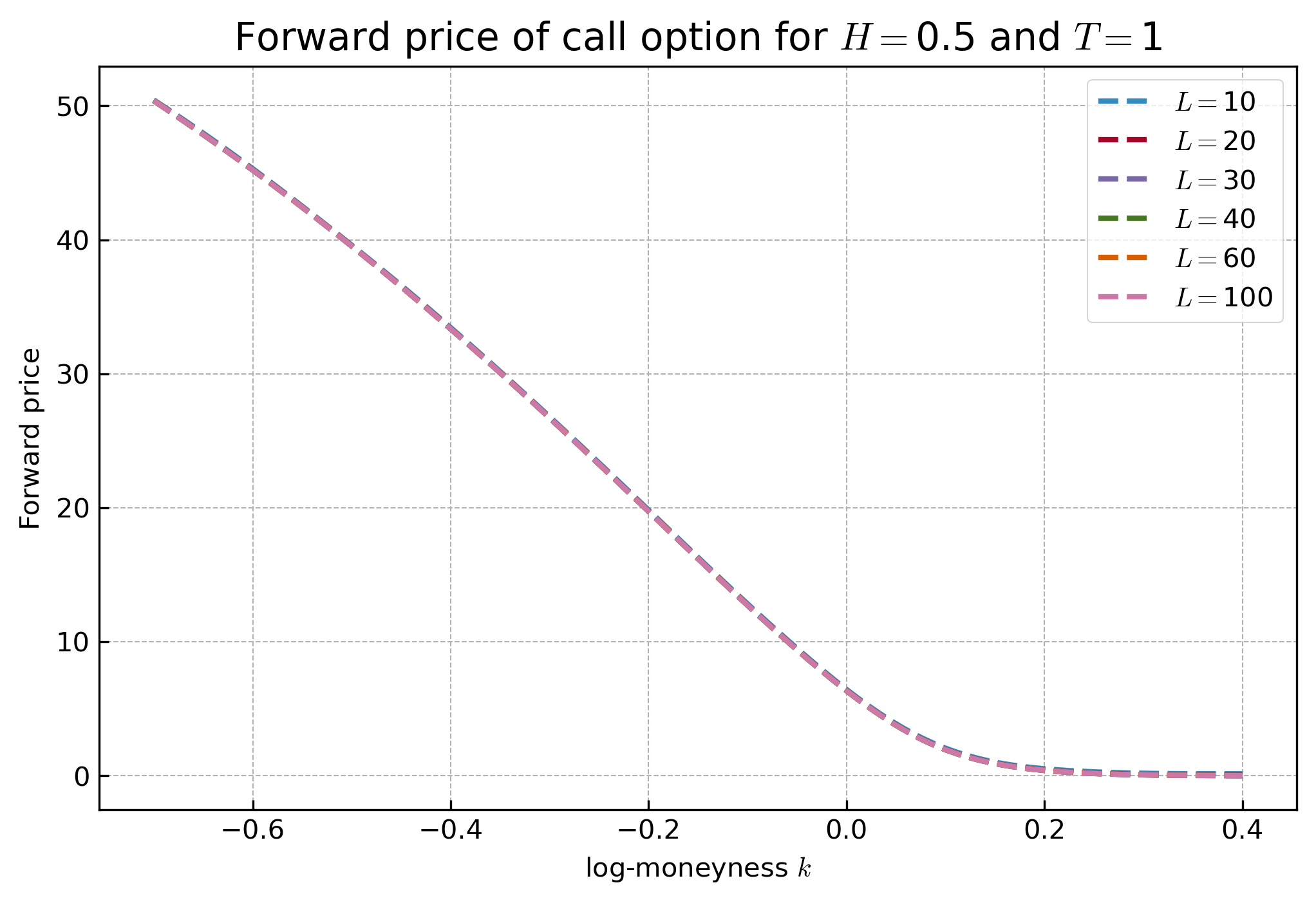} \\
\includegraphics[width=0.4\textwidth]{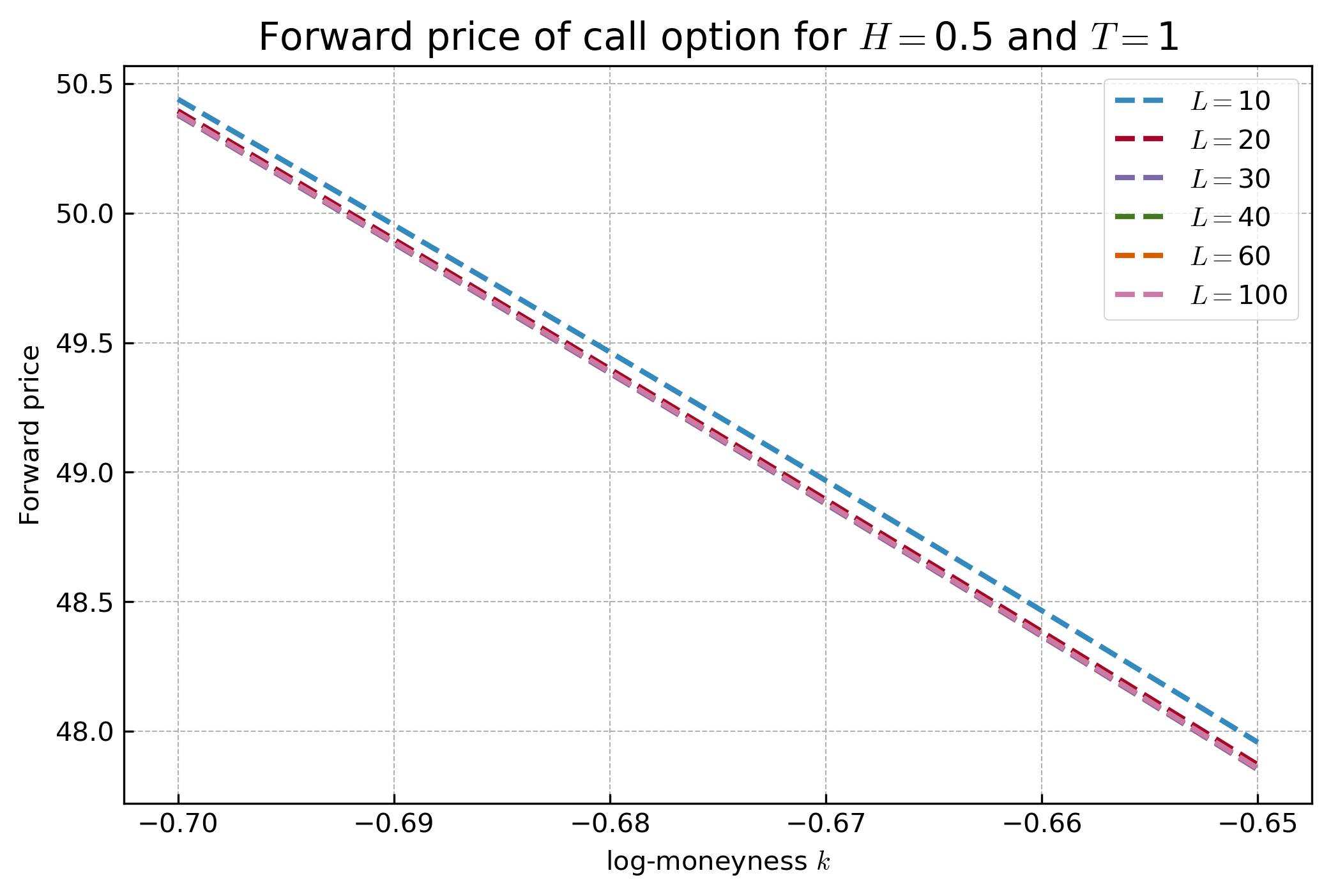} \includegraphics[width=0.4\textwidth]{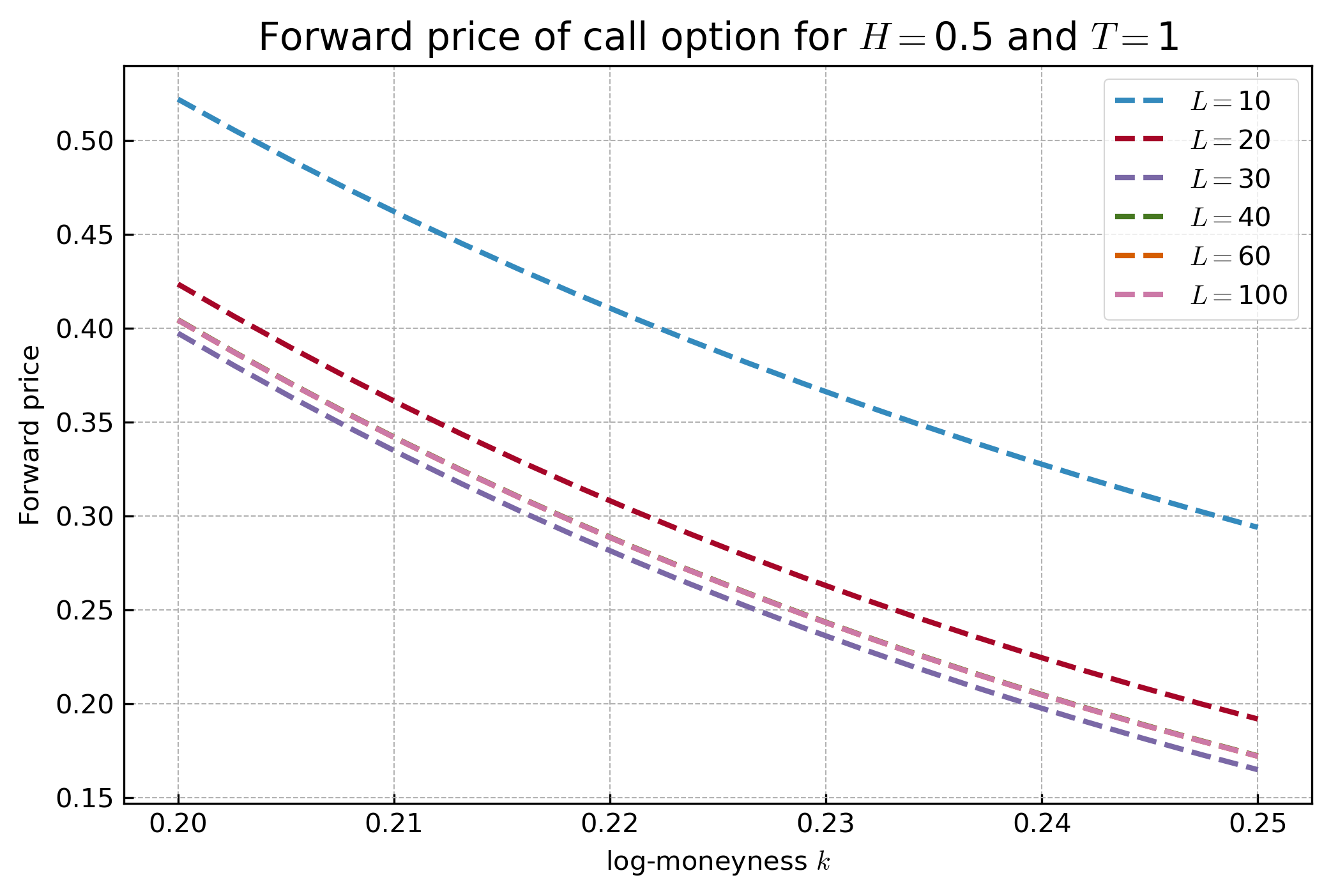}\caption{\protect\label{fig:GL_pricing_T_1}
Forward price of call option for $T=1$ using \cite{key-18} pricing formula with Gauss-Laguerre quadrature numerical integration by varying $L$ (number of weights). Hull-White and Stein-Stein model with parameters: $\kappa_r=-0.03$, $\eta_r=0.01$, $I_0^T=100$, ${\nu}_{0}=0.1$,
${\theta}_{\nu}=0.1,$ ${\eta}_{\nu}= 0.125$, $\kappa_\nu = 0$, ${\rho}_{I\nu}=-0.7,$ ${\rho}_{r\nu}=-0.25,$
${\rho}_{Ir}=-0.25$.}
\end{figure}

\begin{figure}[H]
\centering{}\includegraphics[width=0.75\textwidth]{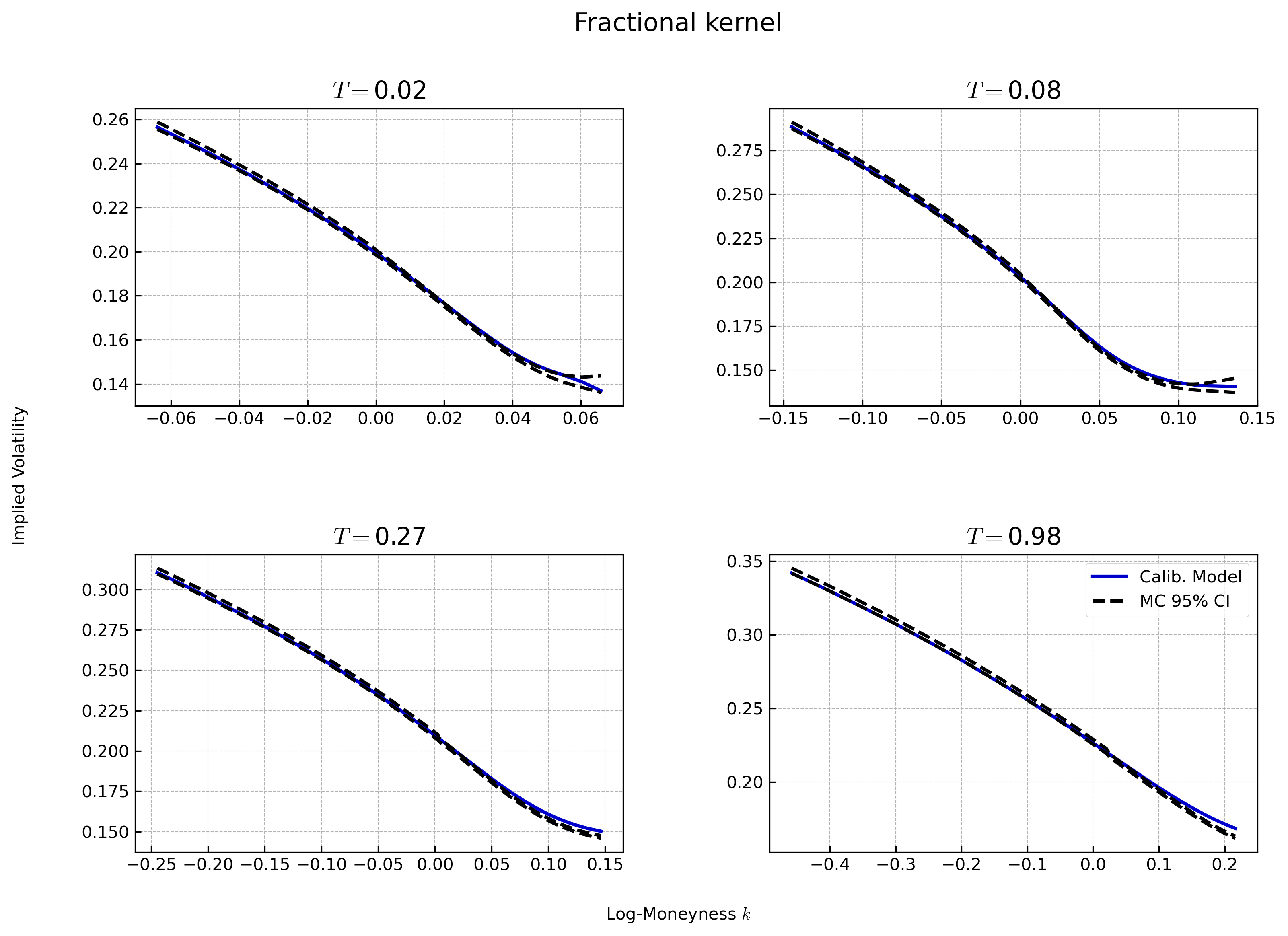}\caption{\protect\label{fig:IV_calib_vs_MC_fract}
Implied volatility for fractional kernel: operator discretization approximation with $N=40$ vs Monte Carlo confidence intervals ($N_\text{sim}=200,000$ and $n_\text{steps}=500$). Calibrated parameters: $\hat{\nu}_{0}=0.1964$,
$\hat{\theta}_{\nu}=-0.0248,$ $\hat{\eta}_{\nu}=0.2123,$ $\hat{\rho}_{I\nu}=-0.7981,$
$\hat{\rho}_{Ir}=-0.5971$ and $\hat{H}_{\nu}=0.2992$. Note that, as shown in Table \ref{tab:conv_MC_short_mat}, the MC implied volatility converges from $n_\text{steps}=500$ for the shortest maturities.}
\end{figure}

\begin{figure}[H]
\centering{}\includegraphics[width=0.75\textwidth]{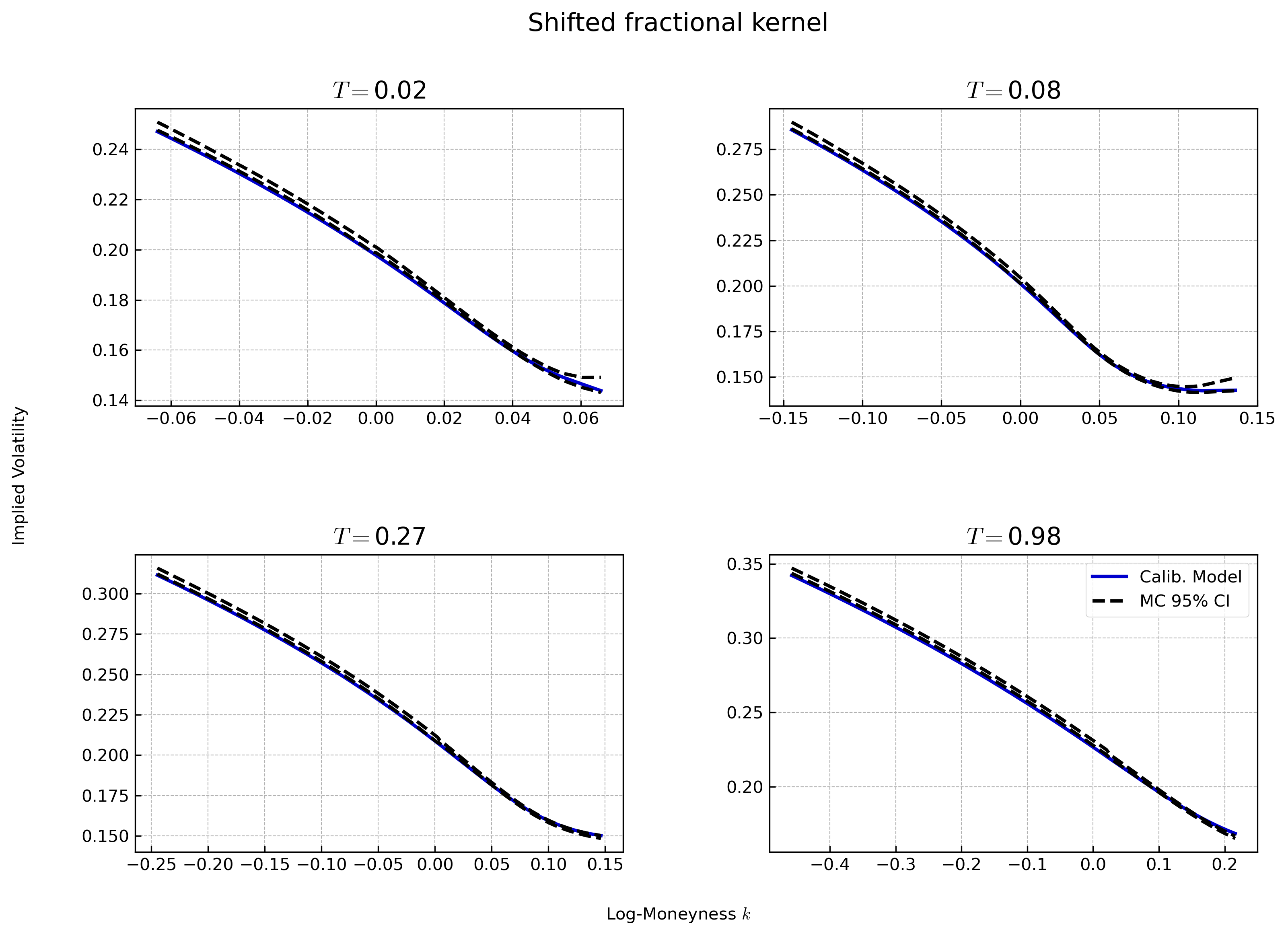}\caption{\protect\label{fig:IV_calib_vs_MC_PD}
Implied volatility for shifted fractional kernel: operator discretization approximation with $N=40$ vs Monte Carlo confidence intervals ($N_\text{sim}=200,000$ and $n_\text{steps}=500$). Calibrated parameters: $\hat{\nu}_{0}=0.1978$,
$\hat{\theta}_{\nu}=-0.0259,$ $\hat{\eta}_{\nu}= 0.2164,$ $\hat{\rho}_{I\nu}=-0.7868,$
$\hat{\rho}_{Ir}=-0.6107$ and $\hat{H}_{\nu}=0.2273$. Note that, as shown in Table \ref{tab:conv_MC_short_mat_shifted}, the MC implied volatility converges
from $n_\text{steps}=500$ for the shortest maturities.}
\end{figure}

\begin{table}[h!] 
\footnotesize
\centering 
\begin{tabular}{|c||c|c|c|c||c|c|c|c|c|}
\hline & \multicolumn{4}{c||}{$T=0.02$} & &  \multicolumn{4}{c|}{$T=0.08$} \\ 
\hline $k$\textbackslash $n_\text{steps}$ & $100$& $250$ & $500$ & $1000$ & $k$\textbackslash $n_\text{steps}$& $100$ &  $250$ & $500$ & $1000$ \\ 
\hline 
$-0.02$ & $0.201$ & $0.201$ & $0.199$ &$0.199$ &$-0.05$ &$0.239$ &$0.239$ &$0.238$ &$0.238$\\ 
$0$ &$0.222$ & $0.221$ &$0.220$ & $0.220$& $0$& $0.204$ &$0.204$ &$0.202$ & $0.202$\\ 
$0.02$&$0.178$ &$0.176$ &$0.176$ &$0.176$ &$0.05$ & $0.163$ &$0.162$ & $0.162$ & $0.162$\\ 
\hline 
\end{tabular} 
\caption{\protect \label{tab:conv_MC_short_mat} Monte Carlo implied volatility with $N_\text{sim}=200,000$ for the fractional kernel with calibrated parameters.}
\end{table}

\begin{table}[h!] 
\footnotesize
\centering 
\begin{tabular}{|c||c|c|c|c||c|c|c|c|c|}
\hline & \multicolumn{4}{c||}{$T=0.02$} & &  \multicolumn{4}{c|}{$T=0.08$} \\ 
\hline $k$\textbackslash $n_\text{steps}$ & $100$& $250$ & $500$ & $1000$ & $k$\textbackslash $n_\text{steps}$& $100$ &  $250$ & $500$ & $1000$ \\ 
\hline 
$-0.02$ & $0.219$ & $0.218$ & $0.217$ &$0.217$ &$-0.05$ &$0.239$ &$0.238$ &$0.237$ &$0.237$\\ 
$0$ &$0.201$ & $0.201$ &$0.199$ & $0.199$& $0$& $0.204$ &$0.204$ &$0.202$ & $0.202$\\ 
$0.02$&$0.182$ &$0.180$ &$0.180$ &$0.180$ &$0.05$ & $0.165$ &$0.163$ & $0.163$ & $0.163$\\ 
\hline 
\end{tabular} 
\caption{\protect \label{tab:conv_MC_short_mat_shifted} Monte Carlo implied volatility with $N_\text{sim}=200,000$ for the shited fractional kernel with calibrated parameters.}
\end{table}

\section*{Acknowledgments}
Eduardo Abi Jaber is grateful for the financial support from the Chaires FiME-FDD, Financial Risks, Deep Finance \& Statistics and Machine Learning and systematic methods in finance at Ecole Polytechnique.
Edouard Motte was financially supported by the Fonds de la Recherche
Scientifique (F.R.S. - FNRS) through a FRIA grant. The authors thank the anonymous referees for their relevant comments and suggestions, which helped to improve the quality of the paper. 

\section*{Disclosure of interest}
No potential conflict of interest was reported by the authors.

\bibliographystyle{plainnat}
\bibliography{bibl}

\end{document}